%% file: main.tex
\def\llncs{0}
\def\fullpage{1}
\def\anonymous{0}
\def\authnote{0}
\def\notxfont{0}
\def\submission{0}

\ifnum\llncs=1
	\documentclass{llncs}
\else
	\documentclass[letterpaper,hmargin=1.05in,vmargin=1.05in,11pt]{article}
			\ifnum\fullpage=1
		\usepackage{fullpage}
		\fi
\fi

\input{preamble_usepackages.tex}
\input{macros}
\usepackage{fancyhdr}
\usepackage{breakcites}

\title{From the Hardness of Detecting Superpositions to Cryptography: Quantum Public Key Encryption and Commitments}
\ifnum\anonymous=0 
\ifnum\llncs=0
\author[1]{Minki Hhan}
\author[2]{\hskip 1em Tomoyuki Morimae}
\author[2,3]{\hskip 1em Takashi Yamakawa}
\affil[1]{{\small KIAS, Seoul, Republic of Korea}\authorcr{\small minkihhan@kias.re.kr}}
\affil[2]{{\small Yukawa Institute for Theoretical Physics, Kyoto University, Kyoto, Japan}\authorcr{\small tomoyuki.morimae@yukawa.kyoto-u.ac.jp}}
\affil[3]{{\small NTT Social Informatics Laboratories, Tokyo, Japan}\authorcr{\small takashi.yamakawa.ga@hco.ntt.co.jp}}
\else
\author{
 Minki Hhan\inst{1} \and Tomoyuki Morimae\inst{2} \and Takashi Yamakawa\inst{2,3}
}
\institute{KIAS, Seoul, Republic of Korea \and
 Yukawa Institute for Theoretical Physics, Kyoto University, Kyoto, Japan \and NTT Social Informatics Laboratories, Tokyo, Japan
}
\fi
\else
\author{}\institute{}
\fi
\date{}

\begin{document}
\ifnum\submission=0
\begin{flushright}
YITP-22-109
\end{flushright}
\fi

{\let\newpage\relax\maketitle}

\ifnum\llncs=1
\thispagestyle{plain}
\fi

\begin{abstract} 
Recently, Aaronson et al. (arXiv:2009.07450) showed that detecting interference between two orthogonal states is as hard as swapping these states. While their original motivation was from quantum gravity, we show its applications in quantum cryptography. 
\begin{enumerate}
    \item
    We construct the first public key encryption scheme from cryptographic \emph{non-abelian} group actions. Interestingly, the ciphertexts of our scheme are quantum even if messages are classical. 
    This resolves an open question posed by  Ji et al. (TCC '19).  We construct the scheme through a new abstraction called swap-trapdoor function pairs, which may be of independent interest. 
    \item We give a simple and efficient compiler that converts the flavor of quantum bit commitments.  
    More precisely, for any prefix
    $\mathrm{X},\mathrm{Y}\in \{\text{computationally,statistically,perfectly}\}$,  
    if the base scheme is X-hiding and Y-binding, then the resulting scheme is Y-hiding and X-binding. Our compiler calls the base scheme only once. Previously, all known compilers call the base schemes polynomially many times (Crépeau et al., Eurocrypt '01 and Yan, Asiacrypt '22).
    For the security proof of the conversion,  we generalize the result of Aaronson et al. by considering quantum auxiliary inputs. 
\end{enumerate}
\end{abstract}

\ifnum\llncs=0
\thispagestyle{empty}
\clearpage

\newpage
\setcounter{tocdepth}{2}
\tableofcontents
\thispagestyle{empty}
\clearpage
\pagenumbering{arabic}
\fi

\ifnum\llncs=0 
\input{Introduction}

\input{preliminaries}
\input{QPKE}

\input{swap_and_distinguish}
\input{conversion}
\input{application}

\ifnum\anonymous=1
\else
{\bf Acknowledgements.}
MH is supported by
a KIAS Individual Grant QP089801.
TM is supported by
JST Moonshot R\verb|&|D JPMJMS2061-5-1-1, 
JST FOREST, 
MEXT QLEAP, 
the Grant-in-Aid for Scientific Research (B) No.JP19H04066, 
the Grant-in Aid for Transformative Research Areas (A) 21H05183,
and 
the Grant-in-Aid for Scientific Research (A) No.22H00522.
\fi

\bibliographystyle{alpha} 
\bibliography{abbrev3,crypto,reference}

\appendix
\input{Appendix_One-Shot_Signatures}
\input{Appendix_Commitments}

\else 

\input{Introduction}
\input{preliminaries}
\input{QPKE}
\input{swap_and_distinguish}
\input{conversion}
\smallskip
{\bf Acknowledgements.}
MH is supported by
a KIAS Individual Grant QP089801.
TM is supported by
JST Moonshot R\verb|&|D JPMJMS2061-5-1-1, 
JST FOREST, 
MEXT QLEAP, 
the Grant-in-Aid for Scientific Research (B) No.JP19H04066, 
the Grant-in Aid for Transformative Research Areas (A) 21H05183,
and 
the Grant-in-Aid for Scientific Research (A) No.22H00522.

\bibliographystyle{splncs04} 
\bibliography{abbrev3,crypto,reference}

\fi

\end{document}

%% file: preamble_usepackages.tex
\usepackage[%
  colorlinks=true,
  citecolor=blue,
  pagebackref=true
]{hyperref}

\usepackage{amsmath, amsfonts, amssymb, mathtools,amscd}

\usepackage{amsthm}

\usepackage{lmodern}
\usepackage[T1]{fontenc}
\usepackage[utf8]{inputenc}

\usepackage{arydshln} 
\usepackage{url}
\usepackage{ifthen}
\usepackage{bm}
\usepackage{multirow}
\usepackage[dvips]{graphicx}
\usepackage[usenames]{color}
\usepackage{xcolor,colortbl} 
\usepackage{threeparttable}
\usepackage{comment}
\usepackage{paralist,verbatim}
\usepackage{cases}
\usepackage{booktabs}
\usepackage{braket}
\usepackage{cancel} 
\usepackage{ascmac} 
\usepackage{framed}
\usepackage{authblk}
\usepackage{pifont}
\usepackage{qcircuit}
\definecolor{darkblue}{rgb}{0,0,0.6}
\definecolor{darkgreen}{rgb}{0,0.5,0}
\definecolor{maroon}{rgb}{0.5,0.1,0.1}
\definecolor{dpurple}{rgb}{0.2,0,0.65}

\usepackage[capitalise,noabbrev]{cleveref}
\usepackage[absolute]{textpos}
\usepackage[final]{microtype}
\usepackage[absolute]{textpos}
\usepackage{everypage}
\DeclareMathAlphabet{\mathpzc}{OT1}{pzc}{m}{it}

\newtheoremstyle{thicktheorem}%
{\topsep}
{\topsep}
{\itshape}{}%
{\bfseries}%
{.}
{ }%
{\thmname{#1}\thmnumber{ #2}%
		\thmnote{ (#3)}%
}

\newtheoremstyle{remark}
{\topsep}
{\topsep}
	{}
	{}
	{}
	{.}
	{ }
	{\textit{\thmname{#1}}\thmnumber{ #2}
			\thmnote{ (#3)}%
	}

\ifnum\llncs=0
	\theoremstyle{thicktheorem}
	\newtheorem{theorem}{Theorem}[section]
	\newtheorem{lemma}[theorem]{Lemma}
	\newtheorem{corollary}[theorem]{Corollary}
	
	\newtheorem{definition}[theorem]{Definition}

	\theoremstyle{remark}
	
	\newtheorem{remark}[theorem]{Remark}

\else
\fi
\newtheorem{MyClaim}[theorem]{Claim}
\Crefname{MyClaim}{Claim}{Claims}

	\crefname{theorem}{Theorem}{Theorems}
	\crefname{assumption}{Assumption}{Assumptions}
	\crefname{construction}{Construction}{Constructions}
	\crefname{corollary}{Corollary}{Corollaries}
	\crefname{conjecture}{Conjecture}{Conjectures}
	\crefname{definition}{Definition}{Definitions}
	\crefname{exmaple}{Example}{Examples}
	\crefname{experiment}{Experiment}{Experiments}
	\crefname{counterexample}{Counterexample}{Counterexamples}
	\crefname{lemma}{Lemma}{Lemmata}
	\crefname{observation}{Observation}{Observations}
	\crefname{proposition}{Proposition}{Propositions}
	\crefname{remark}{Remark}{Remarks}
	\crefname{claim}{Claim}{Claims}
	\crefname{fact}{Fact}{Facts}
	\crefname{note}{Note}{Notes}

\ifnum\llncs=1
 \crefname{appendix}{App.}{Appendices}
 \crefname{section}{Sec.}{Sections}
\else
\fi

\ifnum\llncs=1
\renewcommand*{\backref}[1]{}
\else
	\renewcommand*{\backref}[1]{(Cited on page~#1.)}
	\ifnum\notxfont=1
	\else
		\usepackage{newtxtext}
	\fi
\fi

%% file: macros.tex
\ifnum\authnote=0  
\newcommand{\mor}[1]{}
\newcommand{\minki}[1]{}
\newcommand{\takashi}[1]{}

\else
\newcommand{\mor}[1]{$\ll$\textsf{\color{red} Tomoyuki: { #1}}$\gg$}
\newcommand{\takashi}[1]{$\ll$\textsf{\color{orange} Takashi: { #1}}$\gg$}
\newcommand{\minki}[1]{$\ll$\textsf{\color{darkgreen} Minki: { #1}}$\gg$}

\fi

\newcommand{\calY}{\mathcal{Y}}
\newcommand{\calX}{\mathcal{X}}
\newcommand{\swap}{\mathsf{Swap}}
\newcommand{\eval}{\mathsf{Eval}}

\newcommand{\Tr}{\mathrm{Tr}}

\newcommand{\garbage}{\mathrm{garbage}}

\newcommand{\alive}{\mathrm{Alive}}
\newcommand{\dead}{\mathrm{Dead}}

\newcommand{\qsk}{\mathpzc{sk}}
\newcommand{\qct}{\mathpzc{ct}}

\newcommand{\dms}{\mathsf{DMS}}

\newcommand{\halmic}{\mathsf{HM}}
\newcommand{\inj}{\mathsf{inj}}
\newcommand{\col}{\mathsf{col}}
\newcommand{\injsetup}{\mathsf{keyed}\text{-}\mathsf{inj}}
\newcommand{\gl}{\mathsf{GL}}
\newcommand{\glsetup}{\mathsf{GL}\text{-}\mathsf{setup}}

\newcommand{\ywlq}{\mathsf{YWLQ}}
\newcommand{\my}{\mathsf{MY}}

\newcommand{\StateGen}{\mathsf{StateGen}}

\newcommand{\sample}{\overset{\hspace{0.1em}\mathsf{\scriptscriptstyle\$}}{\leftarrow}}

\newcommand{\seteq}{\coloneqq}

\newcommand{\cA}{\mathcal{A}}

\def\makeuppercase#1{
\expandafter\newcommand\csname tl#1\endcsname{\widetilde{#1}}
}

\def\makelowercase#1{
\expandafter\newcommand\csname tl#1\endcsname{\widetilde{#1}}
}

\newcommand{\regC}{\mathbf{C}}
\newcommand{\regR}{\mathbf{R}}
\newcommand{\regZ}{\mathbf{Z}}
\newcommand{\regD}{\mathbf{D}}

\newcommand{\regB}{\mathbf{B}}
\newcommand{\regA}{\mathbf{A}}
\newcommand{\regX}{\mathbf{X}}
\newcommand{\regY}{\mathbf{Y}}
\newcommand{\regPP}{\mathbf{P}}

\newcommand{\regRX}{\mathbf{R}_{\mathsf{X}}}
\newcommand{\regRK}{\mathbf{R}_{\mathsf{K}}}

\newcommand{\secp}{\lambda}

\newcommand{\crs}{\mathsf{crs}}

\newcommand{\A}{\entity{A}}
\newcommand{\B}{\entity{B}}

\newcommand*{\sk}{\keys{sk}}
\newcommand*{\pk}{\keys{pk}}

\newcommand*{\pp}{\keys{pp}}

\newcommand*{\vk}{\keys{vk}}

\newcommand*{\td}{\keys{td}}

\newcommand*{\keys}[1]{\mathsf{#1}}

\newcommand*{\algo}[1]{\ensuremath{\mathsf{#1}}}

\newcommand*{\entity}[1]{\mathcal{#1}}

\newenvironment{boxfig}[2]{\begin{figure}[#1]\fbox{\begin{minipage}{0.97\linewidth}
                        \vspace{0.2em}
                        \makebox[0.025\linewidth]{}
                        \begin{minipage}{0.95\linewidth}
            {{
                        #2 }}
                        \end{minipage}
                        \vspace{0.2em}
                        \end{minipage}}}{\end{figure}}

\newcommand{\bit}{\{0,1\}}

\newcommand{\setup}{\algo{Setup}}

\newcommand{\keygen}{\algo{KeyGen}}

\newcommand{\enc}{\algo{Enc}}
\newcommand{\dec}{\algo{Dec}}

\newcommand{\sign}{\algo{Sign}}
\newcommand{\vrfy}{\algo{Vrfy}}

\newcommand{\negl}{{\mathsf{negl}}}

\newcommand{\poly}{{\mathrm{poly}}}

\makeatletter
\DeclareRobustCommand
  \myvdots{\vbox{\baselineskip4\p@ \lineskiplimit\z@
    \hbox{.}\hbox{.}\hbox{.}}}
\makeatother

%% file: introduction.tex
\section{Introduction}\label{sec:intro}
When can we efficiently distinguish a superposition of two orthogonal states from their probabilistic mix?
A folklore answer to this question was that we can efficiently distinguish them whenever we can efficiently map one of the states to the other. Recently, Aaronson, Atia and, Susskind~\cite{AAS20} gave a complete answer to the question. They confirmed that the folklore was almost correct but what actually characterizes the distinguishability is the ability to \emph{swap} the two states rather than the ability to map one of the states to the other.\footnote{We remark that the meaning of ``swap'' here is different from that of the SWAP gate as explained below.} 

We explain their result in more detail by using the example of Schr{\"o}dinger's cat following~\cite{AAS20}. Let $\ket{\alive}$ and $\ket{\dead}$ be orthogonal states, which can be understood as the states of alive and dead cats in Schr{\"o}dinger's cat experiment. Then, the authors showed that one can efficiently swap $\ket{\alive}$ and $\ket{\dead}$ (i.e., there is an efficiently computable unitary $U$ such that $U\ket{\dead}=\ket{\alive}$ and $U\ket{\alive}=\ket{\dead}$) if and only if there is an efficient distinguisher that distinguishes $\frac{\ket{\alive}+\ket{\dead}}{\sqrt{2}}$ and $\frac{\ket{\alive}-\ket{\dead}}{\sqrt{2}}$ with certainty. 
Note that distinguishing $\frac{\ket{\alive}+\ket{\dead}}{\sqrt{2}}$ and $\frac{\ket{\alive}-\ket{\dead}}{\sqrt{2}}$ is equivalent to distinguishing $\frac{\ket{\alive}+\ket{\dead}}{\sqrt{2}}$ and the uniform probabilistic mix of $\ket{\alive}$ and $\ket{\dead}$.\footnote{
The distinguishing advantage is (necessarily) halved.
This can be seen by the following equality: \tiny\begin{align*}
 &\frac{1}{2}\left(\ket{\alive}\bra{\alive}+\ket{\dead}\bra{\dead}\right)\\
 =
 &\frac{1}{2}\left(\left(\frac{\ket{\alive}+\ket{\dead}}{\sqrt{2}}\right)\left(\frac{\bra{\alive}+\bra{\dead}}{\sqrt{2}}\right)+\left(\frac{\ket{\alive}-\ket{\dead}}{\sqrt{2}}\right)\left(\frac{\bra{\alive}-\bra{\dead}}{\sqrt{2}}\right)\right).
 \end{align*}
 } 
 Moreover, they showed that the equivalence is robust in the sense that a partial ability to swap $\ket{\alive}$ and $\ket{\dead}$, i.e., $|\bra{\dead}U\ket{\alive}+\bra{\alive}U\ket{\dead}|=\Gamma$ for some $\Gamma > 0$
is equivalent to distinguishability of $\frac{\ket{\alive}+\ket{\dead}}{\sqrt{2}}$ and $\frac{\ket{\alive}-\ket{\dead}}{\sqrt{2}}$ with advantage $\Delta=\Gamma/2$. 
They gave an interpretation of their result that observing interference between alive and dead cats is ``necromancy-hard'', i.e., at least as hard as bringing a dead cat back to life. 

While their original motivation was from quantum gravity, we find their result interesting from cryptographic perspective. Roughly speaking, the task of swapping $\ket{\alive}$ and $\ket{\dead}$ can be thought of as a kind of search problem where one is given $\ket{\alive}$ (resp. $\ket{\dead}$) and asked to ``search'' for $\ket{\dead}$ (resp. $\ket{\alive}$). On the other hand, the task of distinguishing $\frac{\ket{\alive}+\ket{\dead}}{\sqrt{2}}$ and $\frac{\ket{\alive}-\ket{\dead}}{\sqrt{2}}$ is apparently a decision problem.  
From this perspective, we can view their result as a ``search-to-decision'' reduction. Search-to-decision reductions have been playing the central role in cryptography, e.g., the celebrated Goldreich-Levin theorem~\cite{STOC:GolLev89}.  
Based on this observation, we tackle the following two problems in quantum cryptography.\footnote{It may be a priori unclear why these problems are related to \cite{AAS20}. This will become clearer in the technical overview in \Cref{sec:overview}.} 



\smallskip
\noindent\textbf{Public key encryption from non-abelian group actions.}
Brassard and Yung~\cite{C:BraYun90} initiated the study of cryptographic group actions. We say that a group $G$ acts on a set $S$ by an action $\star:G\times S\rightarrow S$ if the following are satisfied:
\begin{enumerate}
    \item For the identity element $e\in G$ and any $s\in S$, we have $e\star s=s$.
    \item For any $g,h\in G$ and any $s\in S$, we have $(gh)\star s = g\star(h\star s)$.
\end{enumerate}
For a cryptographic purpose, we assume (at least) that the group action is one-way, i.e., it is hard to find $g'$ such that $g'\star s=g\star s$ given $s$ and $g\star s$. The work of \cite{C:BraYun90} proposed instantiations of such cryptographic group actions based on the hardness of discrete logarithm, factoring, or graph isomorphism problems. 

Cryptographic group actions are recently gaining a renewed attention from the perspective of \emph{post-quantum} cryptography. Ji et al. \cite{TCC:JQSY19} proposed new instantiations based on general linear group actions on tensors. Alamati et al. \cite{AC:ADMP20} proposed isogeny-based instantiations based on earlier works~\cite{Couveignes06,cryptoeprint:2006/145,AC:CLMPR18}. Both of them are believed to be secure against quantum adversaries. 

An important difference between the instantiations in \cite{TCC:JQSY19} and \cite{AC:ADMP20} is that the former considers \emph{non-abelian} groups whereas the latter considers \emph{abelian} groups. 
Abelian group actions are particularly useful because they give rise to a non-interactive key exchange protocol similar to Diffie-Hellman key exchange~\cite{DH76}.
Namely, suppose that $s\in S$ is published as a public parameter, Alice publishes $g_A \star s$  as a public key while keeping $g_A$ as her secret key, and Bob publishes $g_B \star s$  as a public key while keeping $g_B$ as his secret key. Then, they can establish a shared key $g_A\star (g_B\star s) = g_B\star (g_A\star s)$. On the other hand, an eavesdropper Eve cannot know the shared key since she cannot know $g_A$ or $g_B$ by the one-wayness of the group action.\footnote{For the actual security proof, we need a stronger assumption than the one-wayness. This is similar to the necessity of decisional Diffie-Hellman assumption, which is stronger than the mere hardness of the discrete logarithm problem, for proving security of Diffie-Hellman key exchange.} 
This also naturally gives a public key encryption (PKE) scheme similar to ElGamal encryption~\cite{ElGamal85}.  
On the other hand, the above construction does not work if $G$ is a non-abelian group. Indeed, cryptographic applications given in \cite{TCC:JQSY19} are limited to \emph{Minicrypt} primitives~\cite{Impagliazzo95}, i.e., those that do not imply PKE in a black-box manner.  Thus, \cite{TCC:JQSY19} raised the following open question:\footnote{The statement of the open problem in~\cite{TCC:JQSY19} is quoted as follows: ``{\it Finally, it is an important open problem to build quantum-secure public-key encryption schemes based on
hard problems about GLAT or its close variations.}'' 
Here, GLAT stands for General Linear Action on Tensors, which is their instantiation of non-abelian group action. 
Thus, {\bf Question 1} is slightly more general than what they actually ask. \label{footnote:quotation}
}

\begin{quote}
{\bf Question 1:} \emph{Can we construct PKE from non-abelian group actions?
}
\end{quote}

\smallskip
\noindent\textbf{Flavor conversion for quantum bit commitments.}
Commitments are one of the most important primitives in cryptography. It enables one to ``commit'' to a  (classical) bit\footnote{We can also consider commitments for multi-bit strings. But we focus on \emph{bit} commitments in this paper.} in such a way that the committed bit is hidden from other parties before the committer reveals it, which is called the \emph{hiding} property, and the committer cannot change the committed bit after sending the commitment, which is called the \emph{binding} property. One can easily see that it is impossible for \emph{classical} commitments to achieve both hiding and binding properties against unbounded-time adversaries. It is known to be impossible even with \emph{quantum} communication~\cite{LC97,May97}. Thus, it is a common practice in cryptography to relax either of them to hold only against computationally bounded adversaries. We say that a commitment scheme is computationally (resp. statistically) binding/hiding, if it holds against (classical or quantum depending on the context) polynomial-time (resp. unbounded-time) adversaries. Then, there are the following two \emph{flavors} of commitments: One is computationally hiding and statistically binding, and the other is computationally binding and statistically hiding.\footnote{Of course, we can also consider computationally hiding and computationally binding one, which is weaker than both flavors.} 
In the following, whenever we require statistical hiding or binding, the other one should be understood as computational since it is impossible to statistically achieve both of them as already explained. 

In classical cryptography, though commitments of both flavors are known to be equivalent to the existence of one-way functions~\cite{JC:Naor91,SIAM:HILL99,STOC:HaiRei07}, 
there is no known direct conversion between them that preserves efficiency or the number of interactions. 
Thus, their constructions have been studied separately. 

Recently, Yan~\cite{AC:Yan22}, based on an earlier work by Crépeau, Légaré, and Salvail~\cite{EC:CreLegSal01}, showed that the situation is completely different for quantum bit commitments, which rely on quantum communication between the sender and receiver. First, he showed a round-collapsing theorem, which means that any interactive quantum bit commitments can be converted into non-interactive ones. 
Then he gave a conversion that converts the flavor of any non-interactive quantum bit commitments using the round-collapsing theorem.  

Though Yan's conversion gives a beautiful equivalence theorem, a disadvantage of the conversion is that it does not preserve the efficiency. 
Specifically, it calls the base scheme polynomially many times (i.e., $\Omega(\secp^2)$ times for the security parameter $\secp$). Then, it is natural to ask the following question: 

\begin{quote}
{\bf Question 2:} \emph{Is there an efficiency-preserving flavor conversion for quantum bit commitments?
}
\end{quote}

\subsection{Our Results}
We answer both questions affirmatively using (a generalization of) the result of~\cite{AAS20}. 

For {\bf Question 1}, we construct a PKE scheme with quantum ciphertexts based on non-abelian group actions. 
This resolves the open problem posed by~\cite{TCC:JQSY19}.\footnote{
The statement of their open problem (quoted in \Cref{footnote:quotation}) does not specify if we are allowed to use quantum ciphertexts. Thus, we claim to resolve the problem even though we rely on quantum ciphertexts. 
If they mean \emph{post-quantum} PKE (which has classical ciphertexts), this is still open. 
}
Our main construction only supports classical one-bit messages, but we can convert it into one that supports quantum multi-qubit messages by hybrid encryption with quantum one-time pad as showin in~\cite{C:BroJef15}.
Interestingly, ciphertexts of our scheme are quantum even if messages are classical.  
We show that our scheme is IND-CPA secure if the group action satisfies \emph{pseudorandomness}, which is a stronger assumption than the one-wayness introduced in \cite{TCC:JQSY19}. In addition, we show a ``win-win'' result similar in spirit to \cite{EC:Zhandry19b}.  We show that if the group action is one-way, then  our PKE scheme is IND-CPA secure \emph{or} we can use the group action to construct one-shot signatures~\cite{STOC:AGKZ20}.\footnote{This is a simplified claim and some subtle issues about uniformness of the adversary and ``infinitely-often security'' are omitted here. See \Cref{lem:Gap_CF_and_CH} for the formal statement.}  Note that constructing one-shot signatures has been thought to be a very difficult task. The only known construction is relative to a classical oracle and there is no known construction in the standard model. Even for its significantly weaker variant called tokenized signatures~\cite{BDS17}, the only known construction in the standard model is based on indistinguishability obfuscation~\cite{C:CLLZ21}. 
Given the difficulty of constructing tokenized signatures, let alone one-shot signatures, it is reasonable to conjecture that our PKE scheme is IND-CPA secure if we built it on ``natural'' one-way group actions. Our PKE scheme is constructed through an abstraction called \emph{swap-trapdoor function pairs} (STFs), which may be of independent interest.  

For {\bf Question 2},
We give a new conversion between the two flavors of quantum commitments. That is, for $\mathrm{X},\mathrm{Y}\in \{\text{computationally,statistically,perfectly}\}$, 
    if the base scheme is X-hiding and Y-binding, then the resulting scheme is Y-hiding and X-binding.
Our conversion calls the base scheme only once in superposition. 
Specifically, if $Q_b$ is the unitary applied by the sender when committing to $b\in \bit$ in the base scheme, 
the committing procedure of the resulting scheme consists of a single call to $Q_0$ or $Q_1$ controlled by an additional qubit  (i.e., application of a unitary such that $\ket{b}\ket{\psi}\mapsto \ket{b}(Q_b\ket{\psi})$) and additional constant number of gates. 
For the security proof of our conversion, we develop a generalization of the result of~\cite{AAS20} where we consider auxiliary quantum inputs. 


We show several applications of our conversion. 
We remark that our conversion does not give any new feasibility results since similar conversions with worse efficiency were already known~\cite{EC:CreLegSal01,AC:Yan22}. However, our conversion gives schemes with better efficiency in terms of the number of calls to the building blocks.   
\ifnum\llncs=0
\input{summary_applications}
\else
\fi

\ifnum\llncs=0
\input{related_work}
\else
\fi

\section{Technical Overview}\label{sec:overview} 
We give a technical overview of our results. In the overview, we assume that the reader has read the informal explanation of the result of \cite{AAS20} at the beginning of \Cref{sec:intro}.
\subsection{Part I: PKE from Group Actions}\label{sec:overview_QPKE}
Suppose that a (not necessarily abelian) group $G$ acts on a finite set $S$ by a group action $\star:G\times S \rightarrow S$. Suppose that it is one-way, i.e., it is hard to find $g'$ such that $g'\star s=g\star s$ given $s$ and $g\star s$.\footnote{We will eventually need pseudorandomness, which is stronger than one-wayness, for the security proof of our PKE scheme. We defer the introduction of pseudorandomness for readability. 
} 

Our starting point is the observation made in \cite{C:BraYun90} that one-way group actions give claw-free function pairs as follows. Let $s_0$ and $s_1\seteq g\star s_0$ be public parameters where $s_0\in S$ and $g\in G$ are uniformly chosen.
Then if we define a function $f_b:G\rightarrow S$ by $f_b(h)\seteq h\star s_b$ for $b\in \bit$, the pair $(f_0,f_1)$ is claw-free, i.e., it is hard to find $h_0$ and $h_1$ such that $f_0(h_0)=f_1(h_1)$. This is because if one can find such $h_0$ and $h_1$, then one can break the one-wayness 
of the group action 
by outputting $h_1^{-1}h_0$,
since $f_0(h_0)=f_1(h_1)$ implies $(h_1^{-1}h_0)\star s_0 =s_1$.  

Unfortunately, claw-free function pairs are not known to imply PKE. The reason of the difficulty of constructing PKE is that claw-free function pairs do not have trapdoors. Indeed, it is unclear if there is a trapdoor that enables us to invert $f_0$ and $f_1$ for the above group-action-based construction. 
Our first observation is that the above construction actually has a weak form of a trapdoor: If we know $g$ as a trapdoor, then we can find $h_1$ such that $f_0(h_0)=f_1(h_1)$ from $h_0$ by simply setting $h_1\seteq h_0 g^{-1}$ and vice versa. Though this trapdoor $g$ does not give the power to invert $f_0$ or $f_1$, this enables us to break claw-freeness in a strong sense. We formalize such function pairs as swap-trapdoor function pairs (STFs).\footnote{The intuition of the name is that one can ``swap'' $h_0$ and $h_1$ given a trapdoor.} For the details of STFs, see Sec.~\ref{sec:STF}.

Next, we explain our construction of a PKE scheme with quantum ciphertexts. Though it is a generic construction based on STFs with certain properties, we here focus on the group-action-based instantiation for simplicity. (For the generic construction based on STFs, see Sec.~\ref{sec:QPKE_construction}.)
A public key of our PKE scheme consists of $s_0$ and $s_1= g\star s_0$ and a secret key is $g$. For encrypting a bit $b$, the ciphertext is set to be 
\begin{align} \label{eq:ct}
    \qct_b\seteq \frac{1}{\sqrt{2}}\left(
    \ket{0}\ket{f^{-1}_0(y)} + (-1)^b \ket{1}\ket{f^{-1}_1(y)} 
    \right)
\end{align}
for a random $y\in S$.\footnote{Precisely, $y$ is distributed as $h \star s_0$ for uniformly random $h\in G$.} 
Here, $\ket{f^{-1}_{b'}(y)}$ is the uniform superposition over $f^{-1}_{b'}(y)\seteq\{h\in G: f_{b'}(h)=y\}$ for $b'\in\bit$.
The above state can be generated by a standard technique similar to \cite{FOCS:BCMVV18,FOCS:Mahadev18b}. Specifically, we first prepare
\begin{align*}
\frac{1}{\sqrt{2}}(\ket{0} + (-1)^b \ket{1})
    \otimes 
    \frac{1}{\sqrt{|G|}}\sum_{h\in G}\ket{h},
\end{align*}
compute a group action by $h$ in the second register on $s_0$ or $s_1$ controlled by the first register to get 
\begin{align*}
\frac{1}{\sqrt{2|G|}}\Big(\sum_{h\in G}\ket{0}\ket{h}\ket{h\star s_0} + (-1)^b \sum_{h\in G}\ket{1}\ket{h}\ket{h\star s_1}\Big),
\end{align*}
and measure the third register to get $y\in S$. At this point, the first and second registers collapse to the state in \Cref{eq:ct}.\footnote{Note that $|f_0^{-1}(y)|=|f_1^{-1}(y)|$ for all $y\in S$.}
Decryption can be done as follows. Given a ciphertext $\qct_b$, we apply a unitary $\ket{h}\rightarrow \ket{h g}$ on the second register controlled on the first register. Observe that the unitary maps $\ket{f^{-1}_1(y)}$ to $\ket{f^{-1}_0(y)}$. Then, the resulting state is  $\frac{1}{\sqrt{2}}\left(
    \ket{0}\ket{f^{-1}_0(y)} + (-1)^b \ket{1}\ket{f^{-1}_0(y)} 
    \right)$. 
Thus, measuring the first register in the Hadamard basis results in message $b$. 


Next, we discuss how to prove security. Our goal is to prove that the scheme is IND-CPA secure, i.e., $\qct_0$ and $\qct_1$ are computationally indistinguishable. Here, we rely on the result of \cite{AAS20}. According to their result, one can distinguish $\qct_0$ and $\qct_1$ if and only if one can swap $\ket{0}\ket{f^{-1}_0(y)}$ and $\ket{1}\ket{f^{-1}_1(y)}$. Thus, it suffices to prove the hardness of swapping $\ket{0}\ket{f^{-1}_0(y)}$ and $\ket{1}\ket{f^{-1}_1(y)}$ with a non-negligible advantage.\footnote{See \Cref{thm:AAS} for the precise meaning of the advantage for swapping.} 
Unfortunately, we do not know how to prove this solely assuming the claw-freeness of $(f_0,f_1)$. Thus, we introduce a new assumption called \emph{conversion hardness}, which requires that one cannot find $h_1$ such that $f_1(h_1)=y$ given $\ket{f^{-1}_0(y)}$ with a non-negligible probability. Assuming it, the required hardness of swapping follows straightforwardly since if one can swap $\ket{0}\ket{f^{-1}_0(y)}$ and $\ket{1}\ket{f^{-1}_1(y)}$, then one can break the conversion hardness by first mapping $\ket{0}\ket{f^{-1}_0(y)}$ to $\ket{1}\ket{f^{-1}_1(y)}$ and then measuring the second register. 

The remaining issue is how to prove conversion hardness based on a reasonable assumption on the group action. We show that pseudorandomness introduced in \cite{TCC:JQSY19} suffices for this purpose. Pseudorandomness requires the following two properties:
\begin{enumerate}
\item \label{item:orbit_overview}
The probability that there exists $g\in G$ such that  $g\star s_0 =s_1$ is negligible where $s_0,s_1\in S$ are uniformly random.
\item \label{item:PR_overvoew}
The distribution of 
$(s_0,s_1\seteq g\star s_0)$ where 
$s_0\in S$ and $g \in G$ are uniformly random is computationally indistinguishable from the uniform distribution over $S^2$. 
\end{enumerate}
Note that we require \Cref{item:orbit_overview} because otherwise \Cref{item:PR_overvoew} may unconditionally hold, in which case there is no useful cryptographic application. 
We argue that pseudorandomness implies conversion hardness as follows. By \Cref{item:PR_overvoew}, the attack against the conversion hardness should still succeed with almost the same probability even if we replace $s_1$ with a uniformly random element of $S$. However, then there should exist no solution by \Cref{item:orbit_overview}. Thus, the original success probability should be negligible.  

While \cite{TCC:JQSY19} gave justification of pseudorandomness of their instantiation of group actions, it is a stronger assumption than one-wayness. Thus, it is more desirable to get PKE scheme solely from one-wayness. Toward this direction, we show the following ``win-win'' result inspired by \cite{EC:Zhandry19b}. If $(f_0,f_1)$ is claw-free but not conversion hard, then we can construct a one-shot signatures. 
Roughly one-shot signatures are a quantum primitive which enables us to generate a classical verification key $\vk$ along with a quantum signing key $\qsk$ in such a way that one can use $\qsk$ to generate a signature for whichever message of one's choice, but cannot generate signatures for different messages simultaneously. 
\ifnum\submission=0
(See \Cref{def:OSS} for the formal definition.) 
\fi
For simplicity, suppose that  $(f_0,f_1)$ is claw-free but its conversion hardness is totally broken. 
That is, we assume that we can efficiently find $h_1$ such that $f_1(h_1) = y$ given $\ket{f^{-1}_0(y)}$. Our idea is to set $\ket{f^{-1}_0(y)}$ to be the secret key  and $y$ to be the corresponding verification key. 
For signing to $0$, the signer simply measures $\ket{f^{-1}_0(y)}$ to get $h_0\in f^{-1}_0(y)$ and set $h_0$ to be the signature for the message 0. For signing to $1$, the signer runs the adversary against conversion hardness to get $h_1$ such that $f_1(h_1) = y$ and set $h_1$ to be the signature for the message 1. If one can generate signatures to $0$ and $1$ simultaneously, we can break claw-freeness since $f_0(h_0)=f_1(h_1)=y$. Thus, the above one-shot signature is secure if $(f_0,f_1)$ is claw-free. In the general case where the conversion hardness is not necessarily completely broken, our idea is to amplify the probability of finding $h_1$ from $\ket{f_0^{-1}(y)}$ by a parallel repetition. 
\ifnum\submission=0
See \Cref{sec:proof_Gap_CF_and_CH} for the full proof.  
\fi
Based on this result, we can see that if the group action is one-way, then our PKE scheme is IND-CPA secure or we can construct one-shot signatures.

\subsection{Part II: Flavor Conversion for Commitments}\label{sec:overview_conversion}
\noindent\textbf{Definition of quantum bit commitments.}
First, we recall the definition of quantum bit commitments as formalized by Yan \cite{AC:Yan22}.
He (based on earlier works~\cite{ICALP:ChaKerRos11,YWLQ15,FUYZ20}) showed that any (possibly interactive) quantum bit commitment scheme can be written in the following (non-interactive) canonical form. A canonical quantum bit commitment scheme is characterized by a pair of unitaries $(Q_0,Q_1)$ over two registers $\regC$ (called the commitment register) and $\regR$ (called the reveal register) and works as follows.
\begin{itemize}
    \item[{\bf Commit phase}:] 
    For committing to a bit $b\in \bit$, the sender generates the state $Q_b\ket{0}_{\regC,\regR}$ and sends $\regC$ to the receiver while keeping $\regR$ on its side.\footnote{We write $\ket{0}$ to mean $\ket{0\ldots 0}$ for simplicity.} 
    \item[{\bf Reveal phase}:] 
    For revealing the committed bit, the sender sends $\regR$ along with the committed bit $b$ to the receiver. Then, the receiver applies $Q_b^\dagger$ to $\regC$ and $\regR$ and measures both registers. If the measurement outcome is $0\ldots 0$, the receiver accepts and otherwise rejects.
\end{itemize}

We require a canonical quantum bit commitment scheme to satisfy the following hiding and binding properties. 
The hiding property is defined analogously to that of classical commitments.  
That is, the computational (resp. statistical) hiding property requires that quantum polynomial-time (resp. unbounded-time) receiver (possibly with quantum advice) cannot distinguish commitments to $0$ and $1$ if only given $\regC$. 

On the other hand, the binding property is formalized in a somewhat different way from the classical case. The reason is that a canonical quantum commitment scheme cannot satisfy the binding property in the classical sense. The classical binding property roughly requires that a malicious sender can open a commitment to either of $0$ or $1$ except for a negligible probability. On the other hand, in canonical quantum bit commitment schemes, if the sender generates a uniform superposition of commitments to $0$ and $1$, it can open the commitment to $0$ and $1$ with probability $1/2$ for each.\footnote{A recent work by Bitansky and Brakerski~\cite{TCC:BitBra21} showed that a quantum commitment scheme may satisfy the classical binding property if the receiver performs a measurement in the commit phase. However, such a measurement is not allowed for canonical quantum bit commitments.} 
Thus, we require a weaker binding property called the honest-binding property, which intuitively requires that it is difficult to map an honestly generated commitment of $0$ to that of $1$ without touching $\regC$. More formally, the computational (resp. statistical) honest-binding property requires that for any polynomial-time computable (resp. unbounded-time computable) unitary $U$ over $\regR$ and an additional register $\regZ$ and an auxiliary state $\ket{\tau}_{\regZ}$, we have   
\begin{align*}
    \left\|(Q_1\ket{0}\bra{0}Q_1^\dagger)_{\regC,\regR}(I_{\regC}\otimes U_{\regR,\regZ})((Q_0\ket{0})_{\regC,\regR}\ket{\tau}_{\regZ})\right\|=\negl(\secp).
\end{align*}
One may think that honest-binding is too weak because it only considers honestly generated commitments. However, somewhat surprisingly, \cite{AC:Yan22} proved that it is equivalent to another binding notion called the \emph{sum-binding}~\cite{EC:DumMaySal00}.\footnote{The term ``sum-binding'' is taken from \cite{EC:Unruh16}.} The sum-binding property requires that the sum of probabilities that any (quantum polynomial-time, in the case of computational binding) \emph{malicious} sender can open a commitment to $0$ and $1$ is at most $1+\negl(\secp)$. In addition, it has been shown that the honest-binding property is sufficient for cryptographic applications including zero-knowledge proofs/arguments (of knowledge), oblivious transfers, and multi-party computation~
\cite{YWLQ15,FUYZ20,AC:Yan21,C:MorYam22}. In this paper, we refer to honest-binding if we simply write binding.  


\smallskip
\noindent\textbf{Our conversion.}
We propose an efficiency-preserving flavor conversion for quantum bit commitments inspired by the result of \cite{AAS20}. 
Our key observation is that the swapping ability and distinguishability look somewhat similar to breaking binding and hiding of quantum commitments, respectively.
The correspondence between distinguishability and breaking hiding is easier to see: The hiding property directly requires that distinguishing commitments to $0$ and $1$ is hard. The correspondence between the swapping ability and breaking binding is less clear, but one can find similarities by recalling the definition of (honest-)binding for quantum commitments: Roughly, the binding property requires that it is difficult to map the commitment to $0$ to that to $1$. Technically, a binding adversary does not necessarily give the ability to swap commitments to $0$ and $1$ since it may map the commitment to $1$ to an arbitrary state instead of to the commitment to $0$. But ignoring this issue (which we revisit later), breaking binding property somewhat corresponds to swapping.

However, an important difference between security notions of quantum commitments and the setting of the theorem of \cite{AAS20} is that the former put some restrictions on registers the adversary can touch: For hiding, the adversary cannot touch the reveal register $\regR$, and for binding, the adversary cannot touch the commitment register $\regC$. To deal with this issue, we make another key observation that the equivalence between swapping and distinguishing shown in~\cite{AAS20} preserves \emph{locality}. That is, if the swapping unitary does not touch some qubits of $\ket{\alive}$ or $\ket{\dead}$, then the corresponding distinguisher does not touch those qubits either, and vice versa.  

The above observations suggest the following conversion. Let $\{Q_0,Q_1\}$ be a canonical quantum bit commitment scheme. Then, we construct another scheme $\{Q'_0,Q'_1\}$ as follows: 
\begin{itemize}
 \item The roles of commitment and reveal registers are swapped from $\{Q_0,Q_1\}$ and the commitment register is augmented by an additional one-qubit register. 
  That is, if $\regC$ and $\regR$ are the commitment and reveal registers of  $\{Q_0,Q_1\}$,  then the commitment and reveal registers of $\{Q'_0,Q'_1\}$ are defined as $\regC':= (\regR,\regD)$ and $\regR':= \regC$ where $\regD$ is a one-qubit register.  
  \item For $b\in \bit$, the unitary $Q_b'$ is defined as follows: 
\begin{align}\label{eq:Q'_b}
    Q_b'\ket{0}_{\regC,\regR}\ket{0}_{\regD}:= \frac{1}{\sqrt{2}}\left((Q_{0}\ket{0})_{\regC,\regR}\ket{0}_{\regD}
    +(-1)^{b} (Q_{1}\ket{0})_{\regC,\regR}\ket{1}_{\regD} \right),
\end{align}
where $(\regC',\regR')$ is rearranged as $(\regC,\regR,\regD)$.\footnote{We only present how $Q_b'$ works on $\ket{0}_{\regC,\regR}\ket{0}_{\regD}$ for simplicity. Its definition on general states can be found in \Cref{thm:conversion}.}  
\end{itemize}

One can see that $\{Q'_0,Q'_1\}$ is almost as efficient as $\{Q_0,Q_1\}$: For generating, $Q'_b\ket{0}_{\regC,\regR}\ket{0}_{\regD}$ one can first prepare $\ket{0}_{\regC,\regR}(\ket{0}+(-1)^b\ket{1})_{\regD}$ and then apply $Q_0$ or $Q_1$ to $(\regC,\regR)$ controlled by $\regD$. 
We prove that the hiding and binding properties of $\{Q_0,Q_1\}$ imply binding and hiding properties of $\{Q'_0,Q'_1\}$, respectively. Moreover, the reduction preserves all three types of computational/statistical/perfect security. Thus, this gives a conversion between different flavors of quantum bit commitments.

\smallskip
\noindent\textbf{Security proof.} 
At an intuitive level, the theorem of \cite{AAS20} with the above ``locality-preserving'' observation seems to easily give a reduction from security of $\{Q'_0,Q'_1\}$ to that of $\{Q_0,Q_1\}$: If we can break the hiding property of $\{Q'_0,Q'_1\}$, then we can distinguish $Q_b'\ket{0}_{\regC,\regR}\ket{0}_{\regD}$ without touching $\regR'=\regC$. Then, their theorem with the above observation gives a swapping algorithm that swaps $(Q_{0}\ket{0}_{\regC,\regR})\ket{0}_{\regD}$ and $(Q_{1}\ket{0}_{\regC,\regR})\ket{1}_{\regD}$ without touching $\regR'=\regC$, which clearly breaks the binding property of $\{Q_0,Q_1\}$. One may expect that the reduction from binding to hiding works analogously.  However, it is not as easy as one would expect due to the following reasons.
\begin{enumerate}
  \item \label{item:swap}
    An adversary that breaks the binding property is weaker than a ``partial'' swapping unitary that swaps $Q'_{0}\ket{0}_{\regC',\regR'}$ and $Q'_{1}\ket{0}_{\regC',\regR'}$ needed for \cite{AAS20}. For example, suppose that we have a unitary $U$ such that $UQ'_{0}\ket{0}_{\regC',\regR'}=Q'_{1}\ket{0}_{\regC',\regR'}$ and $UQ'_{1}\ket{0}_{\regC',\regR'}=-Q'_{0}\ket{0}_{\regC',\regR'}$. Clearly, this completely breaks the binding property of $\{Q'_0,Q'_1\}$. However, this is not sufficient for applying  \cite{AAS20} since $|\bra{0}{Q'_1}^\dagger UQ'_{0}\ket{0}+\bra{0}{Q'_0}^\dagger UQ'_{1}\ket{0}|=0$.
    \item \label{item:ancilla}
    For security of quantum bit commitments, we have to consider adversaries with quantum advice, or at least those with ancilla qubits even for security against uniform adversaries. However, the theorem of \cite{AAS20} does not consider any ancilla qubits. 
\end{enumerate}

Both issues are already mentioned in \cite{AAS20}. 
In particular, \Cref{item:swap} is an essential issue. They prove the existence of a pair of orthogonal states $\ket{\alive}$ and $\ket{\dead}$ such that we can map $\ket{\alive}$ to $\ket{\dead}$ by an efficient unitary, but $|\bra{\dead} U\ket{\alive}+\bra{\alive}U\ket{\dead}|\approx 0$ for all efficient unitaries $U$~\cite[Theorem~3]{AAS20}.
For \Cref{item:ancilla}, they (with acknowledgment to Daniel Gottesman) observe that the conversion from a distinguisher to a swapping unitary works even with any quantum advice, but the other direction does not work if there are ancilla qubits~\cite[Footnote 2]{AAS20}.

One can see that the above issues are actually not relevant to the reduction from the hiding of $\{Q'_0,Q'_1\}$ to the binding of $\{Q_0,Q_1\}$. However, for the reduction from the binding of $\{Q'_0,Q'_1\}$ to the hiding of $\{Q_0,Q_1\}$, both issues are non-trivial. Below, we show how to resolve those issues.


\smallskip
\noindent\textbf{Solution to \Cref{item:swap}.}
By the result of \cite[Theorem~3]{AAS20} as already explained, this issue cannot be resolved if we think of $Q'_{0}\ket{0}_{\regC',\regR'}$ and $Q'_{1}\ket{0}_{\regC',\regR'}$ as general orthogonal states. Thus, we look into the actual form of them presented in \Cref{eq:Q'_b}.
Then, we observe that an adversary against the binding property does not touch $\regD$ since that is part of the commitment register $\regC'$ of $\{Q'_0,Q'_1\}$. Therefore, he cannot cause any interference between $(Q_{0}\ket{0})_{\regC,\regR}\ket{0}_{\regD}$ and  $(Q_{1}\ket{0})_{\regC,\regR}\ket{1}_{\regD}$. Therefore, if it maps $$\frac{1}{\sqrt{2}}\left((Q_{0}\ket{0})_{\regC,\regR}\ket{0}_{\regD}
    + (Q_{1}\ket{0})_{\regC,\regR}\ket{1}_{\regD}\right)\mapsto \frac{1}{\sqrt{2}}\left((Q_{0}\ket{0})_{\regC,\regR}\ket{0}_{\regD}
    - (Q_{1}\ket{0})_{\regC,\regR}\ket{1}_{\regD}\right),$$ then it should also map 
    $$\frac{1}{\sqrt{2}}\left((Q_{0}\ket{0})_{\regC,\regR}\ket{0}_{\regD}
    - (Q_{1}\ket{0})_{\regC,\regR}\ket{1}_{\regD}\right)
    \mapsto \frac{1}{\sqrt{2}}\left((Q_{0}\ket{0})_{\regC,\regR}\ket{0}_{\regD}
    + (Q_{1}\ket{0})_{\regC,\regR}\ket{1}_{\regD}\right).$$ Thus, the ability to map $Q'_{0}\ket{0}_{\regC',\regR'}$ to $Q'_{1}\ket{0}_{\regC',\regR'}$ is equivalent to swapping them for this particular construction when one is not allowed to touch $\regD$. A similar observation extends to the imperfect case as well. Therefore, \Cref{item:swap} is not an issue for the security proof of this construction. 

\smallskip
\noindent\textbf{Solution to \Cref{item:ancilla}.}
To better understand the issue, we review how the conversion from a swapping unitary to a distinguisher works. For simplicity, we focus on the perfect case here, i.e., we assume that there is a unitary $U$ such that $U\ket{\dead}=\ket{\alive}$ and $U\ket{\alive}=\ket{\dead}$ for orthogonal states $\ket{\alive}$ and $\ket{\dead}$. Then, we can construct a distinguisher $\A$ that distinguishes $\frac{\ket{\alive}+\ket{\dead}}{\sqrt{2}}$ and $\frac{\ket{\alive}-\ket{\dead}}{\sqrt{2}}$ as follows: Given a state $\ket{\eta}$, which is either of the above two states $\frac{\ket{\alive}+\ket{\dead}}{\sqrt{2}}$ or $\frac{\ket{\alive}-\ket{\dead}}{\sqrt{2}}$, it prepares $\frac{\ket{0}+\ket{1}}{\sqrt{2}}$ in an ancilla qubit, applies $U$ controlled by the ancilla, and measures the ancilla in Hadamard basis. An easy calculation shows that the measurement outcome is $1$ with probability $1$ if  $\ket{\eta}=\frac{\ket{\alive}+\ket{\dead}}{\sqrt{2}}$ and $0$ with probability $1$ if $\ket{\eta}=\frac{\ket{\alive}-\ket{\dead}}{\sqrt{2}}$.

Then, let us consider what happens if the swapping unitary uses ancilla qubits. That is, suppose that we have $U\ket{\dead}\ket{\tau}=\ket{\alive}\ket{\tau'}$ and $U\ket{\alive}\ket{\tau}=\ket{\dead}\ket{\tau'}$ for some ancilla states $\ket{\tau}$ and $\ket{\tau'}$. When $\ket{\tau}$ and $\ket{\tau'}$ are orthogonal, the above distinguisher does not work because there does not occur interference between states with $0$ and $1$ in the control qubit. To resolve this issue, our idea is to ``uncompute'' the ancilla state. A naive idea to do so is to apply $U^\dagger$, but then this is meaningless since it just goes back to the original state. Instead, we prepare a ``dummy'' register that is initialized to be $\frac{\ket{\alive}+\ket{\dead}}{\sqrt{2}}$. Then, we add an application of $U^\dagger$ to the ancilla qubits and the dummy register controlled by the control qubit. Then, the ancilla qubit goes back to $\ket{\tau}$ while the state in the dummy register does not change because it is invariant under the swapping of $\ket{\alive}$ and $\ket{\dead}$. Then, we can see that this modified distinguisher distinguishes  $\frac{\ket{\alive}+\ket{\dead}}{\sqrt{2}}$ and $\frac{\ket{\alive}-\ket{\dead}}{\sqrt{2}}$ with advantage $1$. 

Unfortunately, when the swapping ability is imperfect, the above distinguisher does not work. However, we show that the following slight variant of the above works: Instead of preparing $\frac{\ket{\alive}+\ket{\dead}}{\sqrt{2}}$, it prepares $\frac{\ket{\alive}\ket{0}+\ket{\dead}\ket{1}}{\sqrt{2}}$. After the controlled application of $U^{\dagger}$, it flips the rightmost register (i.e., apply Pauli $X$ to it). In the perfect case, this variant also works with advantage $1$ since the state in the dummy register becomes $\frac{\ket{\dead}\ket{0}+\ket{\alive}\ket{1}}{\sqrt{2}}$ after the application of the controlled $U^\dagger$, which goes back to the original state $\frac{\ket{\alive}\ket{0}+\ket{\dead}\ket{1}}{\sqrt{2}}$ by the flip. 
Our calculation shows that this version is robust, i.e., it works even for the imperfect case. 

There are several caveats for the above. First, it requires the distinguisher to take an additional quantum advice $\frac{\ket{\alive}\ket{0}+\ket{\dead}\ket{1}}{\sqrt{2}}$, which is not necessarily efficiently generatable in general.\footnote{We remark that they are efficiently generatable in our application where $\ket{\alive}$ and $\ket{\dead}$ correspond to commitments to $0$ and $1$.} 
Second, there occurs a quadratic reduction loss unlike the original theorem in \cite{AAS20} without ancilla qubits. Nonetheless, they are not a problem for our purpose.

%% file: summary_applications.tex
\ifnum\llncs=1
\section{Summary of Applications of Our Conversion}\label{sec:summary_applications}
\fi

\begin{enumerate}
    \item In \Cref{sec:construction_Naor}, we apply our conversion to the statistically binding scheme from PRGs by Yan, Weng, Lin, and Quan~\cite{YWLQ15}. Then, we obtain the first statistically hiding quantum commitment scheme from PRGs that makes only a single call to the PRGs. 
    
    \item In \Cref{sec:construction_PRS}, 
    based on a recent work, by Morimae and Yamakawa~\cite{C:MorYam22}, we show that we can use (single-copy-secure) pseudorandom state generators (PRSGs)~\cite{C:JiLiuSon18} instead of PRGs in the above construction. As a result, we obtain the  the first statistically hiding quantum commitment scheme from PRSGs that makes only a single call to the PRSGs. 

 \item 
 In \Cref{sec:new_inj}, we give a novel simple construction of a perfectly hiding quantum commitment scheme from injective one-way functions that makes a single call to the base function. By applying our conversion to it, we obtain a perfectly binding quantum commitment scheme from injective one-way functions that makes a single call to the base function. Though there is a classical construction of such a scheme based on the Goldreich-Levin theorem~\cite{STOC:GolLev89}, our construction has a shorter commitment length since a commitment does not need to include a seed for the hardcore predicate.
       
 \item In \Cref{sec:new_collapsing}, we show that replacing injective one-way functions with (sufficiently length-decreasing) collapsing functions~\cite{EC:Unruh16} in the above constructions yields commitment schemes with the other flavor. As a result, we obtain the first statistically binding quantum commitment scheme from collapsing hash functions that makes a single call to the collapsing hash function.   
   
    
\end{enumerate}

We provide more detailed comparisons with existing constructions after the presentation of each construction in \Cref{sec:application}. 
In addition, we present more applications of our conversion (including applications to the schemes of \cite{C:HalMic96,EC:DumMaySal00}) in \Cref{sec:more_application}.

%% file: related_work.tex
\ifnum\llncs=0
\subsection{Related Work}
\else
\section{More Related Work}\label{sec:more_related_work}
\fi
\noindent\textbf{Cryptographic group actions.} 
Brassard and Yung~\cite{C:BraYun90} initiated the study of cryptographic group actions and proposed instantiations based on the hardness of graph isomorphism,
discrete logarithm, or factoring. However, they are not suitable for our purpose since it turns out that the graph isomorphism problem can be solved in (classical) quasi polynomial-time~\cite{STOC:Babai16}\footnote{Another issue is that the graph isomorphism problem is easy for a uniformly random instance, and thus it cannot satisfy our definition of one-wayness (\Cref{def:one-way}) that requires average case hardness. If we modify the definition of the one-wayness to choose the hardest instance, the graph isomorphism-based construction may satisfy it, and such a version suffices for our applications. However, since such a construction can be broken in quasi-polynomial time by Babai's algorithm \cite{STOC:Babai16}, we do not consider this instantiation  and simply consider average case version in the definition of one-wayness. A similar remark can be found in \cite[Remark~1]{TCC:JQSY19}.} and discrete logarithm and factoring problems can be solved in quantum polynomial time~\cite{Shor99}.

 Alamati et al. \cite{AC:ADMP20} gave an abstraction of isogeny-based cryptography as group actions. However, the isogeny-based construction only supports limited functionality formalized as \emph{Restricted Effective Group Action} (REGA). Though it might be possible to modify our definition of group actions (\Cref{def:GA_efficient}) to capture isogeny-based construction by considering similar restrictions, we do not do so because isogeny-based PKE is already known even without relying on quantum ciphetexts~\cite{Couveignes06,cryptoeprint:2006/145,PQCRYPTO:JaoDeFo11,AC:CLMPR18}.

We consider the general linear group action on tensors proposed by \cite{TCC:JQSY19} as a main instantiation for our construction of PKE.  Though their security is a newly introduced assumption by \cite{TCC:JQSY19}, they justify it by pointing out reductions to many important problems in different areas including coding theory, computational group theory, and multivariate cryptography~\cite{FGS19}. They also discuss potential cryptanalyses and demonstrate that none of them seems to work. See  \cite{TCC:JQSY19} for the details.


\smallskip\noindent\textbf{Quantum key distribution.} 
Bennett and Brassard~\cite{BB84} constructed an unconditionally secure key exchange protocol with quantum communication, which is known as \emph{quantum key distribution}. 
We remark that quantum key distribution protocols are inherently interactive unlike our quantum PKE with quantum ciphertexts. Indeed, it is easy to see that unconditionally secure PKE with classical keys and quantum ciphertexts is impossible since a brute-force search for the correct decryption key would totally break security.

\smallskip
\noindent\textbf{Quantum public key encryption.}
There are several works that proposed ``quantum PKE'' schemes. We compare them with our PKE with quantum ciphertexts. 

The ``quantum PKE'' in \cite{C:OkaTanUch00} is entirely classical except that the key generation algorithm can be quantum. The security of their scheme relies on the hardness of the subset-sum problem. Thus, 
their quantum PKE is incomparable to our PKE with quantum ciphertexts where key generation is classical, and their underlying assumption is also incomparable to ours.

The ``quantum PKE'' in \cite{EC:KKNY05} is PKE with quantum ciphertexts and \emph{quantum public keys}. On the other hand, our quantum-ciphertext PKE uses quantum ciphertexts and \emph{classical} public keys. 
Thus, their quantum PKE is a weaker primitive than our PKE with quantum ciphertexts. 
We remark classical public keys are much more desirable since we can certify classical public keys by using digital signatures while generating signatures on quantum messages is known to be impossible~\cite{Alagic_2021}.
The technical aspect of our PKE scheme is somewhat similar to \cite{EC:KKNY05} in the sense that both embed messages into phases of quantum states.

\smallskip
\noindent\textbf{Quantum bit commitments.}
Bennett, Brassard, and Crépeau  \cite{BB84,C:BraCre90} initiated the study of quantum bit commitments. Unfortunately, it turned out to be impossible to construct an unconditionally secure quantum bit commitments~\cite{LC97,May97}. 
Thus, later works constructed quantum bit commitments relying on complexity assumptions~\cite{EC:DumMaySal00,EC:CreLegSal01,KO09,KO11,YWLQ15,C:MorYam22,C:AnaQiaYue22}. 
A seminal work by Yan~\cite{AC:Yan22} showed that any (possibly interactive) quantum bit commitments can be converted into one in a non-interactive \emph{canonical} form. His definition of quantum bit commitments in the canonical form requires a seemingly weak binding property called \emph{honest-binding}. However, he showed that it is actually equivalent to \emph{sum-binding}, which has been traditionally used as a definition of a binding property of quantum bit commitments~\cite{EC:DumMaySal00,EC:CreLegSal01,KO09,KO11,C:MorYam22}. In addition, some works~\cite{YWLQ15,FUYZ20,AC:Yan21,C:MorYam22} showed that quantum bit commitments in the canonical form can be used as a building block of other cryptographic primitives including zero-knowledge proofs or arguments (of knowledge), oblivious transfers, and multi-party computations. 
Thus, we use quantum bit commitments in the canonical form (with the honest-binding property) as defined in~\cite{AC:Yan22} as a default definition of quantum bit commitments in this paper.

\noindent\textbf{Other notions of binding.}
As explained above, we use honest-binding as a default definition of binding. On the other hand, there are several other definitions of binding for quantum commitments. We review them and give comparisons with honest-binding. (Similar discussions can also be found in \cite{AC:Yan22}.)  

Bitansky and Brakerski~\cite{TCC:BitBra21}  introduced the notion of \emph{classical-binding} for quantum commitments. It roughly requires that the committed message is uniquely determined by the commitment. Though this is impossible to achieve for canonical quantum bit commitments, they avoid the impossibility by having the receiver \emph{measure} the commitment in a certain way. 
The advantage of the classical binding property is that it is conceptually similar to the binding of classical commitments, and thus it is easy to give security proofs when plugging it into some protocol as a substitute for classical commitments. On the other hand, existing works~\cite{YWLQ15,FUYZ20,C:MorYam22} show that the statistical honest-binding quantum commitments are already useful for many applications. Indeed, there seems no known application for which classical-binding suffices but honest-binding does not. 

Ananth, Qian, and Yuen~\cite{C:AnaQiaYue22} introduced a new definition of a statistical binding property for quantum commitments, which we call AQY-binding. The motivation of this definition is for the application to quantum oblivious transfers and multi-party computation~\cite{C:BCKM21b}. However, \cite[Appendix B]{C:MorYam22} observed that the statistical honest-binding property implies the AQY-binding property based on the technique of \cite{FUYZ20}. A full proof is given in \cite[Appendix B]{AC:Yan22}.

Yan~\cite{AC:Yan21} proved that the \emph{computational} honest-binding property implies what is called the computational \emph{predicate-binding} property, which is sufficient for implementing Blum's Hamiltonicity protocol.

There are several other definitions of \emph{computational} binding for quantum (string) commitments~\cite{TCC:CDMS04,C:DamFehSal04} that are shown to be more useful in applications than computational honest binding ones. However, there is no known construction that satisfies the definition of \cite{TCC:CDMS04}, and the only known construction that satisfies \cite{C:DamFehSal04} is in the CRS model and based on a special assumption that is tailored to their construction. (See \cite{EC:Unruh16,AC:Yan21} for more details of these definitions.) 

\subsection{Concurrent Work}
A concurrent work by Gunn, Ju, Ma, and Zhandry~\cite{GJMZ22} defines commitments to quantum states and shows duality between binding and hiding for them. In particular, as the special case of commitments to classical strings, they give a similar flavor conversion to ours~\cite[Section 4.4.1]{GJMZ22}. 
The difference from our work is that their definitions of binding and hiding are different from ours. 
Specifically, they require what they call ``$Z$-binding'', which is similar to collapse-binding introduced by Unruh~\cite{EC:Unruh16}, and ``$X$-hiding'', which requires indistinguishability against adversaries that can submit a superposition of classical messages as a challenge message.\footnote{$Z$ and $X$ for $Z$-binding and $X$-binding stand for Pauli operators. Do not confuse them with our notation $\mathrm{X},\mathrm{Y}\in \{\text{computationally,statistically,perfectly}\}$.} 
We make the following observation about these definitions:
\begin{itemize}
    \item Though $Z$-binding is seemingly stronger than Yan's binding (\cref{def:binding}), we can prove that they are actually equivalent using the result of \cite{cryptoeprint:2022/786}.\footnote{We thank James Bartusek, Fermi Ma, and Jun Yan for suggesting it.}
    \item 
    Though $X$-hiding seems stronger than Yan's hiding (\cref{def:hiding}), there is a very simple compiler to upgrade Yan's hiding into $X$-hiding. For committing to a bit $b$, we simply add an additional one-qubit register that is set to be $\ket{b}$ in the reveal register. This additional qubit prevents a hiding adversary from causing interference between commitments to different bits, and thus $X$-hiding immediately follows from Yan's hiding for this scheme.
\end{itemize}
Based on the above observation, we can recover our coversion from that in \cite{GJMZ22}. 
Specifically, starting from 
a quantum bit commitment scheme 
that satisfies (statistical or computational) Yan's hiding and Yan's biding, we first apply the compiler in the second item above to upgrade Yan's hiding to $X$-hiding and then apply the conversion of \cite{GJMZ22}. This results in exactly the same conversion as the one presented in this paper.


%% file: preliminaries.tex
\section{Preliminaries}
\ifnum\llncs=0
\input{definitions}
\else
Notations used throughout the paper and definitions of basic cryptographic primitives are given in the full version. 
\fi

\subsection{Canonical Quantum Bit Commitments}
We define \emph{canonical} quantum bit commitments as defined in \cite{AC:Yan22}.  

\begin{definition}[Canonical quantum bit commitments]\label{def:canonical_com}
A canonical quantum bit commitment scheme is represented by a family $\{Q_0(\secp),Q_1(\secp)\}_{\secp\in \mathbb{N}}$ of polynomial-time computable unitaries over two registers $\regC$ (called the \emph{commitment} register) and $\regR$ (called the \emph{reveal} register).  
In the rest of the paper, we often omit $\secp$ and simply write $Q_0$ and $Q_1$ to mean $Q_0(\secp)$ and $Q_1(\secp)$. 
\end{definition}
\begin{remark}
Canonical quantum bit commitments are supposed to be used as follows. In the commit phase, to commit to a bit $b\in \bit$, 
the sender generates a state $Q_b\ket{0}_{\regC,\regR}$ and sends $\regC$ to the receiver while keeping $\regR$. In the reveal phase, the sender sends $b$ and $\regR$ to the receiver. The receiver projects the state on $(\regC,\regR)$ onto   $Q_b\ket{0}_{\regC,\regR}$, and accepts if it succeeds and otherwise rejects.  
\end{remark}


\begin{definition}[Hiding]\label{def:hiding}
We say that a canonical quantum bit commitment scheme $\{Q_0,Q_1\}$ is computationally (rep. statistically) \emph{hiding} if $\Tr_{\regR}(Q_0(\ket{0}\bra{0})_{\regC,\regR}Q_0^\dagger)$ is computationally (resp. statistically) indistinguishable from $\Tr_{\regR}(Q_1(\ket{0}\bra{0})_{\regC,\regR}Q_1^\dagger)$. 
We say that it is perfectly hiding if they are identical states.  
\end{definition}

\begin{definition}[Binding]\label{def:binding}
We say that a canonical quantum bit commitment scheme $\{Q_0,Q_1\}$ is computationally (rep. statistically) \emph{binding} if for any polynomial-time computable   
(resp. unbounded-time) unitary $U$ over $\regR$ 
and an additional register $\regZ$ and any polynomial-size state $\ket{\tau}_{\regZ}$, 
it holds that 
\begin{align*}
    \left\|(Q_1\ket{0}\bra{0}Q_1^\dagger)_{\regC,\regR}(I_{\regC}\otimes U_{\regR,\regZ})((Q_0\ket{0})_{\regC,\regR}\ket{\tau}_{\regZ})\right\|=\negl(\secp).
\end{align*}
We say that it is perfectly binding if the LHS is $0$ for all unbounded-time unitary $U$. 
\end{definition}

\subsection{Equivalence between Swapping and Distinguishing}
\minki{Does this section is appropriate for preliminaries? I usually put things that are well-known to communities... But I agree that there is no other good place for this theorem.}
The following theorem was proven in \cite{AAS20}. 

\begin{theorem}[{\cite[Theorem 2]{AAS20}}] \label{thm:AAS}~
\begin{enumerate}
    \item \label{item:AAS_one}
    Let $\ket{x},\ket{y}$ be orthogonal $n$-qubit states.
    Let $U$ be a polynomial-time computable unitary over $n$-qubit states and define $\Gamma$ as
    \begin{align*}
        \Gamma :=  
        \left|\bra{y}U\ket{x}
        + 
        \bra{x}U\ket{y}
        \right|.
    \end{align*} 
    Then, there exists a QPT distinguisher $\A$ that makes a single black-box access to controlled-$U$  and distinguishes  
    $\ket{\psi}\coloneqq\frac{\ket{x}+\ket{y}}{\sqrt{2}}$ and  $\ket{\phi}\coloneqq\frac{\ket{x}-\ket{y}}{\sqrt{2}}$ with advantage $\frac{\Gamma}{2}$.
    Moreover, if $U$ does not act on some qubits, then $\A$ also does not act on those qubits.
    \item  \label{item:AAS_two}
    Let $\ket{\psi},\ket{\phi}$ be orthogonal $n$-qubit states, and suppose that a QPT distinguisher $\A$ 
    distinguishes $\ket{\psi}$ and $\ket{\phi}$ with advantage $\Delta$ without using any ancilla qubits.  
    Then, there exists a polynomial-time computable unitary $U$ over $n$-qubit states such that 
    \begin{align*}
        \frac{|\bra{y}U\ket{x}+\bra{x}U\ket{y}|}{2}=\Delta
    \end{align*}
    where  $|x\rangle\coloneqq\frac{|\psi\rangle+|\phi\rangle}{\sqrt{2}}$ and $|y\rangle\coloneqq\frac{|\psi\rangle-|\phi\rangle}{\sqrt{2}}$.
     Moreover, if $\A$ does not act on some qubits, then $U$ also does not act on those qubits.
\end{enumerate}
\end{theorem}
\begin{remark}[Descriptions of quantum circuits.]
For the reader's convenience, we give the concrete descriptions of quantum circuits for the above theorem, which are presented in \cite{AAS20}. 

For \Cref{item:AAS_one},  let $\widetilde{U}\seteq e^{i\theta} U$ for $\theta$ such that 
\[
\mathrm{Re}(\bra{y}\widetilde{U}\ket{x}
        + 
        \bra{x}\widetilde{U}\ket{y})=\left|\bra{y}U\ket{x}
        + 
        \bra{x}U\ket{y}
        \right|.\]
Then, $\A$ is described in \Cref{fig:AAS_one}.

For \Cref{item:AAS_two}, let $V_{\A}$ be a unitary such that 
\begin{align*}
    &V_{\A}\ket{\psi}=\sqrt{p}\ket{1}\ket{\psi_1}+\sqrt{1-p}\ket{0}\ket{\psi_0}\\
    &V_{\A}\ket{\phi}=\sqrt{1-p+\Delta}\ket{0}\ket{\phi_0}+\sqrt{p-\Delta}\ket{1}\ket{\phi_1}
\end{align*}
for some $\ket{\psi_0}$, $\ket{\psi_1}$, $\ket{\phi_0}$, and $\ket{\phi_1}$. 
That is, $V_{\A}$ is the unitary part of $\A$.
Then, $U$ is described in \Cref{fig:AAS_two}.

\begin{figure}[t]
\[\Qcircuit @C=1em @R=0.5em {
\lstick{\ket{0}}  
        &\gate{H} &\ctrl{1}             &\gate{H}   &\meter \\
\lstick{\ket{\psi}\text{~or~}\ket{\phi}} 
        &\qw      &\gate{\widetilde{U}} &\qw        &\qw\\
}\]
\caption{Quantum circuit for $\A$ in \Cref{item:AAS_one} of \Cref{thm:AAS}.}
\label{fig:AAS_one}
\end{figure}

\begin{figure}[t]
\centering
\[\Qcircuit @C=0.8em @R=0.8em {
& \multigate{3}{V_\A}   &\qw&\qw        &\qw&\gate{Z}   &\qw&\qw        &\qw&\multigate{3}{V_\A^\dagger}& \qw\\
& \ghost{V_\A}          &\qw&\qw        &\qw&\qw        &\qw&\qw        &\qw&\ghost{V_\A^\dagger}       & \qw \\
& \ghost{V_\A}          &\qw&\myvdots   &   &\qw        &\qw&\myvdots   &   &\ghost{V_\A^\dagger}       & \qw\\
& \ghost{V_\A}          &\qw&\qw        &\qw&\qw        &\qw&\qw        &\qw&\ghost{V_\A^\dagger}       & \qw 
\inputgroupv{1}{4}{0.8em}{2.6em}{\ket{x}\text{~or~}\ket{y}\phantom{aaaaaa}}\\
}\]
\caption{Quantum circuit for $U$ in \Cref{item:AAS_two} of \Cref{thm:AAS}.}
\label{fig:AAS_two}
\end{figure}

\end{remark}
\begin{remark}
Though the final requirement in both items (``Moreover,...'') is not explicitly stated in \cite[Theorem 2]{AAS20}, it is easy to see from \Cref{fig:AAS_one,fig:AAS_two}. This observation is important for our application to commitments and PKE.
\end{remark}

%% file: definitions.tex
\ifnum\llncs=1
\section{Omitted Preliminaries}\label{sec:omitted_preliminaries}
\fi

\smallskip
\noindent\textbf{Basic notations.} 
We use $\secp$ to mean the security parameter throughout the paper. The dependence on $\secp$ is often implicit. For example, we simply write a function $f:\bit^n \rightarrow \bit^m$ to mean a collection $\{f_\secp:\bit^{n(\secp)} \rightarrow \bit^{m(\secp)}\}_{\secp\in \mathbb{N}}$ for some functions $n(\secp)$ and $m(\secp)$ etc.
For a finite set $X$, we write $x\gets X$ to mean that we uniformly take $x$ from $X$. 
For a (possibly randomized) classical or quantum algorithm $\A$, we write $y\gets \A(x)$ to mean that $\A$ takes $x$ as input and outputs $y$. 
For a function $f:X\rightarrow Y$ and $y\in Y$, we write $f^{-1}(y)$ to mean the set of all preimages of $y$, i.e., $f^{-1}(y)\seteq \{x\in X:f(x)=y\}$. 
We say that a probability distribution is statistically close to another probability distribution if their statistical distance is negligible.

\smallskip
\noindent\textbf{Notations for quantum computations.} 
For simplicity, $|0...0\rangle$ is sometimes written as $|0\rangle$.
Quantum registers are denoted by bold fonts, e.g., $\regA,\regB$ etc. 
$\Tr_\regA(\rho_{\regA,\regB})$ is the partial trace over the register $\regA$ of the bipartite state $\rho_{\regA,\regB}$.
For simplicity, the tensor product $\otimes$ is sometimes omitted: for example, $|\psi\rangle\otimes|\phi\rangle$
is sometimes written as $|\psi\rangle|\phi\rangle$.
$I$ is the identity operator on a single qubit.
For simplicity, we often write $I^{\otimes m}$ just as $I$ when the dimension is clear from the context. 
For any two states $\rho_1$ and $\rho_2$,
$F(\rho_1,\rho_2)$ is the fidelity between them. 
For a set $S$ of classical strings, we define $\ket{S}\seteq \frac{1}{\sqrt{|S|}}\sum_{x\in S}\ket{x}$. 

\smallskip
\noindent\textbf{Computational models.}
We say that a classical algorithm is probabilistic polynomial time (PPT) if it can be computed by a polynomial-time (classical) probabilistic Turing machine.  
We say that a quantum algorithm is quantum polynomial time (QPT) if it can be computed by a polynomial-time quantum Turing machine (or equivalently a quantum circuit generated by a polynomial-time Turing machine). 
We say that a quantum algorithm is non-uniform QPT if it can be computed by a polynomial-size quantum circuits (or polynomial-time quantum Turing machine) with quantum advice. We use non-uniform QPT algorithms as a default model of adversaries unless otherwise noted.

We say that a sequence $\{U_\secp\}_{\secp\in \mathbb{N}}$ of unitary operators is polynomial-time computable if there is a polynomial-time Turing machine that on input $1^\secp$ outputs a description of a quantum circuit that computes $U_\secp$. We often omit the dependence on $\secp$ and simply write $U$ is polynomial-time computable to mean the above. 

\smallskip
\noindent\textbf{Distinguishing advantage.}
For a quantum algorithm $\A$ and quantum states $\ket{\psi}$ and $\ket{\phi}$, we say that $\A$ distinguishes $\ket{\psi}$ and $\ket{\phi}$ with advantage $\Delta$ if 
\begin{align*}
    |\Pr[\A(\ket{\phi})=1]-\Pr[\A(\ket{\psi})=1]|=\Delta.
\end{align*}

\subsection{Basic Cryptographic Primitives}
\begin{definition}[One-way functions]
We say that a classical polynomial-time computable function $f:\bit^n \rightarrow \bit^m$ is a \emph{one-way function (OWF)} if for any non-uniform QPT adversary $\A$, we have 
\begin{align*}
    \Pr[f(x')=f(x):x\gets \bit^n,x' \gets \A(1^\secp,f(x))]=\negl(\secp).
\end{align*}
We say that a one-way function 
$f$ is an \emph{injective} one-way function if $f$ is injective, and that a one-way function 
$f$ is a one-way \emph{permutation} if $f$ is a permutation.
\end{definition}

\begin{definition}[Keyed one-way functions]
We say that a family $\{f_k:\bit^n \rightarrow \bit^m\}_{k\in \mathcal{K}}$ of classical polynomial-time computable functions is a \emph{keyed} one-way function if for any non-uniform QPT adversary $\A$, we have 
\begin{align*}
    \Pr[f_k(x')=f_k(x):k\gets \mathcal{K}, x\gets \bit^n,x' \gets \A(1^\secp,k,f_k(x))]=\negl(\secp).
\end{align*}
We say that 
a keyed one-way fucntion $\{f_k:\bit^n \rightarrow \bit^m\}_{k\in \mathcal{K}}$ is a keyed \emph{injective} one-way function if $f_k$ is injective for all $k\in \mathcal{K}$.
\end{definition}

\begin{definition}[Pseudorandom generators]
We say that a classical polynomial-time computable function $G:\bit^n \rightarrow \bit^m$ is a \emph{pseudorandom generator (PRG)} if $m> n$ and for any non-uniform QPT adversary $\A$, we have 
\begin{align*}
    |\Pr[\A(y)=1:y\gets \bit^{m}]-\Pr[\A(G(x))=1:x\gets \bit^n]|=\negl(\secp).
\end{align*}
\end{definition}

It is well-known that PRGs exist if and only if OWFs exist~\cite{SIAM:HILL99}.

\begin{definition}[Collapsing functions~{\cite{EC:Unruh16}}]\label{def:collapsing}
For a polynomial-time computable function family $\mathcal{H}=\{H_k:\bit^{L}\rightarrow \bit^\ell\}_{k\in \mathcal{K}_{\mathcal{H}}}$ and an adversary $\A$, we define an experiment $\mathsf{Exp}_{\A}^{\mathsf{collapse}}(1^\secp)$ as follows:
\begin{enumerate}
    \item The challenger generates $k\gets \mathcal{K}_{\mathcal{H}}$.
    \item \label{step:send_X}
    $\A$ is given $k$ as input and generates a hash value $y\in \bit^\ell$ and a quantum state $\sigma$ over registers $(\regX,\regA)$ where
    $\regX$ stores an element of $\bit^L$ and $\regA$ is $\A$'s internal register.
    Then it sends $y$ and register $\regX$ to the challenger, and keeps $\regA$ on its side.
    \item The challenger picks $b\gets \bit$. If $b=0$, the challenger does nothing and if $b=1$, the challenger measures register $\regX$ in the computational basis.
    The challenger returns register $\regX$ to $\A$.
    \label{step:collapsing_game_measure}
    \item $\A$ outputs a bit $b'$. The experiment outputs $1$ if $b'=b$ and $0$ otherwise. 
\end{enumerate}
We say that $\A$ is a valid adversary if the following is satisfied:
if we measure the state in $\regX$ right after Step \ref{step:send_X}, then the outcome $x$ satisfies $H_k(x)=y$ with probability $1$. 

We say that $\mathcal H$ is \emph{collapsing} if
for any non-uniform QPT valid adversary $\A,$ we have
\[
|\Pr[1\gets \mathsf{Exp}_{\A}^{\mathsf{collapse}}(1^\secp)]-1/2|= \negl(\secp).
\]
\minki{I changed the orders of notions due to the distance between ``collapsing'' and the definition itself. The previous definition is hidden below.}
\end{definition}

As shown in \cite{EC:Unruh16}, the collapsing property implies the collision-resistance. That is, if $\mathcal{H}=\{H_k:\bit^{L}\rightarrow \bit^\ell\}_{k\in \mathcal{K}_{\mathcal{H}}}$ is collapsing, then it is also collision-resistant, i.e., no non-uniform QPT adversary can find $x\neq x'$ such that $H_k(x)=H_k(x')$ with non-negligible probability given $k\gets \mathcal{K}_\mathcal{H}$.
It is clear that injective functions are collapsing.

Unruh~\cite{AC:Unruh16} showed that there is a collapsing function family with arbitrarily long (or even unbounded) input-length under the LWE assumption (or more generally, under the existence of lossy functions in a certain parameter regime).

\begin{definition}[Single-copy-secure PRSGs~\cite{C:MorYam22}]
\label{definition:1PRS}
A \emph{single-copy-secure pseudorandom quantum states generator (PRSG)} is a QPT
algorithm $\StateGen$ that, on input $k\in\{0,1\}^n$,
outputs an $m$-qubit quantum state $|\phi_k\rangle$.
As the security, we require the following: 
for any non-uniform QPT adversary $\cA$,  
\begin{eqnarray*}
|\Pr_{k\leftarrow \{0,1\}^n}[\cA(|\phi_k\rangle)\to1]
-\Pr_{|\psi\rangle\leftarrow \mu_m}[\cA(|\psi\rangle)\to1]|
=\negl(\secp),
\end{eqnarray*}
where $\mu_m$ is the Haar measure on $m$-qubit states.\footnote{Intuitively,
$|\psi\rangle\leftarrow \mu_m$ means that an $m$-qubit pure state is sampled uniformly at random
from the set of all $m$-qubit pure states.}
\end{definition}
Single-copy-secure PRSGs are a restricted version of (poly-copy-secure) PRSGs introduced in \cite{C:JiLiuSon18},
where any polynomially many copies of $|\phi_k\rangle$ are computationally indistnguishable from the same number of
copies of Haar random states.
If one-way functions exist, (poly-copy-secure) PRSGs exist~\cite{C:JiLiuSon18}.
On the other hand, there is an evidence that (poly-copy-secure) PRSGs do not imply one-way functions~\cite{Kre21}.


%% file: QPKE.tex
\section{Quantum-Ciphertext Public Key Encryption}\label{sec:QPKE}
In \Cref{sec:STF}, we introduce a notion of swap-trapdoor function pairs, which can be seen as a variant of trapdoor claw-free function pairs~\cite{FOCS:GolMicRiv84}. 
In \Cref{sec:QPKE_construction}, we define quantum-ciphertext PKE and construct it based on STFs.  
In \Cref{sec:STF_from_GA}, we construct STFs based on group actions.

\subsection{Swap-Trapdoor Function Pairs}\label{sec:STF}
We introduce a notion of swap-trapdoor function pairs (STFs). 
Similarly to trapdoor claw-free function pairs, a STF consists of two functions $f_0,f_1: \calX\rightarrow \calY$. 
We require that there is a trapdoor which enables us to ``swap'' preimages under $f_0$ and $f_1$, i.e., given $x_b$, we can find $x_{b\oplus 1}$ such that $f_{b\oplus 1 }(x_{b\oplus 1})=f_b(x_b)$. 
The formal definition of STFs is given below. 


\begin{definition}[Swap-trapdoor function pair]\label{def:STF}
A \emph{swap-trapdoor function pair (STF)} consists of algorithms $(\setup,\eval,\swap)$.
\begin{description}
\item[$\setup(1^\secp)\to(\pp,\td)$:] This is a PPT algorithm that takes the security parameter $1^\secp$ as input, and outputs a public parameter $\pp$ and a trapdoor $\td$. The public parameter $\pp$ specifies functions $f_b^{(\pp)}:\calX\rightarrow \calY$ for each $b\in\bit$. We often omit the dependence on $\pp$ and simply write $f_b$ when it is clear from the context.  
\item[$\eval(\pp,b,x)\to y$:] This is a deterministic classical polynomial-time algorithm that takes a public parameter $\pp$, a bit $b\in\bit$, and an element $x\in \calX$ as input, and outputs $y\in \calY$. 
\item[$\swap(\td,b,x)\to x'$:] This is a deterministic classical polynomial-time algorithm that takes a trapdoor $\td$, a bit $b\in \bit$, and an element $x\in \calX$ as input, and outputs $x'\in \calX$.
\end{description}
We require a STF to satisfy the following:

\smallskip\noindent\textbf{Evaluation correctness.}
For any $(\pp,\td)\gets \setup(1^\secp)$ , $b\in\bit$, and $x\in \calX$, we have $\eval(\pp,b,x)=f_b(x)$.

\smallskip\noindent\textbf{Swapping correctness.}
For any $(\pp,\td)\gets \setup(1^\secp)$, $b\in \bit$, and $x\in \calX$, 
 if we let $x'\gets \swap(\td,b,x)$, then we have $f_{b\oplus 1}(x')=f_b(x)$ and $\swap(\td,b\oplus 1,x')=x$.
In particular, $\swap(\td,b,\cdot)$ induces an efficiently computable and invertible one-to-one mapping between $f^{-1}_0(y)$ and $f^{-1}_1(y)$ for any $y\in \calY$. 
 

\smallskip\noindent\textbf{Efficient random sampling over $\calX$.} 
There is a PPT algorithm that samples an almost uniform element of $\calX$ (i.e., the distribution of the sample is statistically close to the uniform distribution). 

\smallskip\noindent\textbf{Efficient superposition over $\calX$.} 
There is a QPT algorithm
 that generates a state whose trace distance from $\ket{\calX}=\frac{1}{\sqrt{|\calX|}}\sum_{x\in\calX}|x\rangle$ is $\negl(\secp)$.
 \end{definition}

\begin{remark}[A convention on ``Efficient random sampling over $\calX$'' and ``Efficient superposition over $\calX$'' properties]\label{rem:random_sampling_and_superposition}
In the rest of this paper, we assume that we can sample elements from \emph{exactly} the uniform distribution of $\calX$.  Similarly, we assume that we can \emph{exactly} generate $\ket{\calX}$ in QPT. They are just for simplifying the presentations of our results, and all the results hold with the above imperfect version with additive negligible loss for security or correctness. 
\end{remark}


We define two security notions for STFs which we call \emph{claw-freeness} and \emph{conversion hardness}.  Looking ahead, what we need in our construction of quantum-ciphertext PKE in \Cref{sec:QPKE_construction} is only conversion hardness. However, since there are interesting relations between them as we show later, we define both of them here.  



\begin{definition}[Claw-freeness]\label{def:claw-free}
We say that a STF $(\setup,\eval,\swap)$ satisfies \emph{claw-freeness} if for any non-uniform QPT algorithm $\A$, we have 
\begin{align*}
    \Pr[f_0(x_0)=f_1(x_1):(\pp,\td)\gets \setup(1^\secp),(x_0,x_1)\gets \A(\pp)]=\negl(\secp).
\end{align*}
\end{definition}

\begin{definition}[Conversion hardness]\label{def:conv_hard}
We say that a STF $(\setup,\eval,\swap)$ satisfies \emph{conversion hardness} if for any non-uniform QPT algorithm $\A$, we have 
\begin{align*}
    \Pr[f_1(x_1)=y:(\pp,\td)\gets \setup(1^\secp),x_0\gets \calX, y\seteq f_0(x_0),  x_1\gets \A(\pp,\ket{f_0^{-1}(y)})]=\negl(\secp)
\end{align*} 
where we remind that $\ket{f_0^{-1}(y)}\seteq\frac{1}{\sqrt{|f_0^{-1}(y)|}}\sum_{x\in f_0^{-1}(y)}\ket{x}$.  
\end{definition}

\begin{remark}[On asymmetry of $f_0$ and $f_1$.]
Conversion hardness requires that it is hard to find $x_1$ such that $f_1(x_1)=y$ given $\ket{f_0^{-1}(y)}$. We could define it in the other way, i.e., it is hard to find $x_0$ such that $f_0(x_0)=y$ given $\ket{f_1^{-1}(y)}$. These two definitions do not seem to be equivalent. However, it is easy to see that if there is a STF that satisfies one of them, then it can be modified to satisfy the other one by just swapping the roles of $f_0$ and $f_1$. In this sense, the choice of the definition from these two versions is arbitrary. 
\end{remark}


We show several lemmas on the relationship between claw-freeness and conversion hardness.


First, we show that claw-freeness implies conversion hardness if $f_0$ is \ifnum\submission=0
collapsing~
(\Cref{def:collapsing}).\footnote{Collapsingness of $f_0$ can be naturally defined 
by ignoring $f_1$ and simply considering $\pp$ as a function index for $f_0$.}
\else
collapsing.
\fi

\begin{lemma}[Claw-free and collapsing $\rightarrow$ Conversion hard]\label{lem:collapsing_conversion_hard}
If $f_0$ is collapsing, then claw-freeness implies conversion hardness.
\end{lemma}
\ifnum\llncs=0 \input{proof_collapsing_conversion_hard}
\else
We defer the proof to the full version. 
\fi

As a special case of \Cref{lem:collapsing_conversion_hard}, claw-freeness implies conversion hardness  when $f_0$ is \emph{injective} (in which case $f_1$ is also injective). This is because any injective function is trivially collapsing.

\minki{It is almost obvious to construct a conversion hard but not claw-free STF assuming the existence of conversion hard STFs. I added this below just for completeness.}

We remark that a conversion hard STF is not necessarily claw-free, because a claw can be augmented in STF without hurting the conversion hardness.

Next, we show a ``win-win'' result inspired from \cite{EC:Zhandry19b}. We roughly show that a claw-free but non-conversion-hard STF can be used to construct one-shot signatures~\cite{STOC:AGKZ20}. Roughly one-shot signatures are a genuinely quantum primitive which enables us to generate a classical verification key $\vk$ along with a quantum signing key $\qsk$ in such a way that one can use $\qsk$ to generate a signature for whichever message of one's choice, but cannot generate signatures for different messages simultaneously. 
\ifnum\submission=0
(See \Cref{def:OSS} for the formal definition.) 
\fi
The only known construction of one-shot signatures is relative to a classical oracle and there is no known construction in the standard model. Even for its weaker variant called tokenized signatures~\cite{BDS17}, the only known construction in the standard model is based on indistinguishability obfuscation~\cite{C:CLLZ21}.
Given the difficulty of constructing tokenized signatures, let alone one-shot signatures, it is reasonable to conjecture that natural candidate constructions 
of STFs satisfy conversion hardness if it satisfies claw-freeness. 
This is useful because claw-freeness often follows from weaker assumptions than conversion hardness, which is indeed the case for the group action-based construction in \Cref{sec:STF_from_GA}.

Before stating the lemma, we remark some subtlety about the lemma. 
Actually, we need to assume a STF that is claw-free but not \emph{infinitely-often} \emph{uniform} conversion hard. Here, ``infinitely-often'' means that it only requires the security to hold for infinitely many security parameters rather than all security parameters. (See \cite[Sec. 4.1]{EC:Zhandry19b} for more explanations about infinitely-often security.) 
The ``uniform'' means that security is required to hold only against uniform adversaries as opposed to non-uniform ones. 
Alternatively, we can weaken the assumption to a STF that is claw-free but not uniform conversion hard if we weaken the goal to be \emph{infinitely-often} one-shot signatures. 
We remark that similar limitations also exist for the ``win-win'' result in \cite{EC:Zhandry19b}.  

Then, the lemma is given below. 
\begin{lemma}[Claw-free and non-conversion hard STF $\rightarrow$ One-shot signatures]\label{lem:Gap_CF_and_CH}
Let $(\setup,\eval,\swap)$ be a STF that satisfies claw-freeness. Then, the following statements hold:
\begin{enumerate}
    \item \label{item:gap_one}
    If $(\setup,\eval,\swap)$ is not infinitely-often uniform conversion hard, then we can use it to construct one-shot signatures.
     \item \label{item:gap_two}
     If $(\setup,\eval,\swap)$ is not uniform conversion hard, then we can use it to construct infinitely-often one-shot signatures.
\end{enumerate}
\end{lemma}
\ifnum\llncs=0 \begin{proof}(sketch.)
We give a proof sketch here. The full proof can be found in \Cref{sec:proof_Gap_CF_and_CH}. 

For simplicity, suppose that  $(\setup,\eval,\swap)$ is claw-free but its conversion hardness is totally broken. 
That is, we assume that we can efficiently find $x_1$ such that $f_1(x_1) = y$ given $(\pp,\ket{f^{-1}_0(y)})$. Our idea is to set $\pp$ to be the public parameter of one-shot signatures, $\ket{f^{-1}_0(y)}$ to be the secret key, and $y$ to be the corresponding verification key. 
For signing to $0$, the signer simply measures $\ket{f^{-1}_0(y)}$ to get $x_0\in f^{-1}_0(y)$ and set $x_0$ to be the signature. For signing to $1$, the signer runs the adversary against conversion hardness to get $x_1$ such that $f_1(x_1) = y$. If one can generate signatures to $0$ and $1$ simultaneously, we can break claw-freeness since $f_0(x_0)=f_1(x_1)=y$. Thus, the above one-shot signature is secure if $(\setup,\eval,\swap)$ is claw-free. 

In the actual proof, we only assume an adversary that finds $x_1$ with a non-negligible (or noticeable) probability rather than $1$. Then, our idea is to simply repeat the above construction parallelly many times so that at least one of the execution of the adversary succeeds with overwhelming probability. 

We need to assume that the adversary is uniform since that is used as part of the signing algorithm. The ``infinitely-often'' restriction comes from the fact that an inverse of non-negligible function may not be polynomial.
\end{proof}
\else
We defer the proof to the full version  
since the idea is already explained in \cref{sec:overview_QPKE}. 
\fi

\smallskip\noindent\textbf{Instantiations.} Our main instantiation of STFs is based on group actions, which is given in \Cref{sec:STF_from_GA}. 

A lattice-based instantiation is also possible if we relax the requirements to allow some ``noises'' similarly to \cite{FOCS:BCMVV18}. The noisy version is sufficient for our construction of quantum-ciphertext PKE given in \Cref{sec:QPKE_construction}. However, since lattice-based (classical) PKE schemes are already known~\cite{JACM:Regev09,STOC:GenPeiVai08}, we do not try to capture lattice-based instantiations in the definition of STFs. 

\subsection{Quantum-Ciphertext Public Key Encryption}\label{sec:QPKE_construction}
In this section, we define quantum-ciphertext PKE and construct it based on STFs.

\smallskip\noindent\textbf{Definition.} 
We define quantum-ciphertext PKE for one-bit messages for simplicity.
The multi-bit message version can be defined analogously, and a simple parallel repetition works to expand the message length. Moreover, we can further extend the message space to quantum states by a hybrid encryption with quantum one-time pad as in \cite{C:BroJef15}, i.e., we encrypt a quantum message by a quantum one-time pad, and then encrypt the key of the quantum one-time pad by quantum PKE for classical messages.  

\begin{definition}[Quantum-ciphertext public key encryption]\label{def:QPKE}
A \emph{quantum-ciphertext public key encryption (quantum-ciphertext PKE)} scheme (with single-bit messages) consists of algorithms $(\keygen,\enc,\dec)$.
\begin{description}
\item[$\keygen(1^\secp)\rightarrow (\pk,\sk)$:]
This is a PPT algorithm that takes the security parameter $1^\secp$ as input, and outputs a classical public key $\pk$ and a classical secret key $\sk$.
\item[$\enc(\pk,b)\rightarrow \qct$:] 
This is a QPT algorithm that takes a public key $\pk$ and a message $b\in \bit$ as input, and outputs a quantum ciphertext $\qct$. 
\item[$\dec(\sk,\qct)\rightarrow b'/\bot$:]
This is a QPT algorithm that takes a secret key $\sk$ and a ciphertext $\qct$ as input, and outputs a message $b'\in \bit$ or $\bot$. 
\end{description}
It must satisfy correctness as defined below:

\noindent\textbf{Correctness.}
For any $m\in \bit$, we have 
\begin{align*}
    \Pr[m'=m:(\pk,\sk)\gets \keygen(1^\secp),\qct \gets \enc(\pk,m),m'\gets \dec(\sk,\qct)]=1-\negl(\secp).
\end{align*}
\end{definition}

We define IND-CPA security for quantum-ciphertext PKE similarly to that for classical PKE as follows.
\begin{definition}[IND-CPA security]\label{def:IND-CPA}
We say that a quantum-ciphertext PKE scheme $(\keygen,\enc,\dec)$ is \emph{IND-CPA secure} if for any non-uniform QPT adversary $\A$, we have 
\begin{align*}
    \left|
    \Pr\left[
    \A(\pk,\qct_0)=1
    \right]
    -
    \Pr\left[
    \A(\pk,\qct_1)=1
    \right]
    \right|
    =\negl(\secp),
\end{align*}
where $ (\pk,\sk)\gets \keygen(1^\secp)$, 
    $\qct_0\gets \enc(\pk,0)$, and $\qct_1\gets \enc(\pk,1)$.
\end{definition}

\smallskip\noindent\textbf{Construction.}
Let $(\setup,\eval,\swap)$ be a STF. 
We construct a quantum-ciphertext PKE scheme $(\keygen,\enc,\dec)$ as follows.

\begin{description}
\item[$\keygen(1^\secp)$:] 
Generate $(\pp,\td)\gets \setup(1^\secp)$ and output $\pk\seteq \pp$ and $\sk \seteq \td$. 
\item[$\enc(\pk,b\in \bit)$:] 
Parse $\pk=\pp$. 
Prepare two registers $\regD$ and $\regX$.
Generate the state  
\begin{align*}
    \frac{1}{\sqrt{2}}(\ket{0}+(-1)^b\ket{1})_{\regD}\ket{\calX}_{\regX} =
    \frac{1}{\sqrt{2|\calX|}}(\ket{0}+(-1)^b\ket{1})_{\regD}\sum_{x\in \calX}\ket{x}_{\regX}.
\end{align*}
Prepare another register $\regY$, 
coherently compute $f_0$ or $f_1$ into $\regY$ controlled by $\regD$ to get 
\begin{align*}
    \sum_{x\in \calX}\frac{1}{\sqrt{2|\calX|}}(\ket{0}_{\regD}\ket{x}_{\regX}\ket{f_0(x)}_{\regY}+(-1)^b \ket{1}_{\regD}\ket{x}_{\regX}\ket{f_1(x)}_{\regY}),
\end{align*}
and measure $\regY$ to get $y\in \calY$. At this point, $\regD$ and $\regX$ collapse to the following state:\footnote{Note that the swapping correctness implies that $|f^{-1}_0(y)|=|f^{-1}_1(y)|$ for any $y\in \calY$.}
\begin{align*}
    \frac{1}{\sqrt{2}}(\ket{0}_{\regD}\ket{f^{-1}_0(y)}_{\regX}+(-1)^b \ket{1}_{\regD}\ket{f^{-1}_1(y)}_{\regX}).
\end{align*}
The above state is set to be $\qct$.\footnote{Remark that one does not need to include $y$ in the ciphertext.} 
\item[$\dec(\sk,\qct)$:]
Parse $\sk=\td$. 
Let $U_\td$ be a unitary over $\regD$ and $\regX$ such that\footnote{Note that the second operation is possible because $\swap(\td,0,\swap(\td,1,x))=x$.}
\begin{align*}
    U_\td \ket{0}_{\regD} \ket{x}_{\regX} &= \ket{0}_{\regD}\ket{x}_{\regX},\\
    U_\td \ket{1}_{\regD} \ket{x}_{\regX} &= \ket{1}_{\regD}\ket{\swap(\td,1,x)}_{\regX}.
\end{align*}
Apply $U_\td$ on the register $(\regD,\regX)$ and measure $\regD$ in the Hadamard basis and output the measurement outcome $b'\in\bit$. 
\end{description}


\smallskip
\noindent\textbf{Correctness.}
\begin{theorem}
$(\keygen,\enc,\dec)$ satisfies correctness.
\end{theorem}
\begin{proof}
An honestly generated ciphertext $\qct$ is of the form 
\begin{align*}
    \frac{1}{\sqrt{2}}(\ket{0}_{\regD}\ket{f^{-1}_0(y)}_{\regX}+(-1)^b \ket{1}_{\regD}\ket{f^{-1}_1(y)}_{\regX}).
\end{align*}
By the definition of $U_\td$ and the swapping correctness, it is easy to see that we have 
\begin{align*}
   &U_\td \ket{0}_{\regD}\ket{f^{-1}_0(y)}_{\regX}=\ket{0}_{\regD}\ket{f^{-1}_{0}(y)}_{\regX},\\
   &U_\td \ket{1}_{\regD}\ket{f^{-1}_1(y)}_{\regX}=\ket{1}_{\regD}\ket{f^{-1}_{0}(y)}_{\regX}.
\end{align*}
Thus, applying $U_\td$ on $\qct$ results in the following state:
\begin{align*}
    \frac{1}{\sqrt{2}}(\ket{0}_{\regD}\ket{f^{-1}_0(y)}_{\regX}+(-1)^b \ket{1}_{\regD}\ket{f^{-1}_0(y)}_{\regX})=\frac{1}{\sqrt{2}}(\ket{0}_{\regD}+(-1)^b \ket{1}_{\regD})\otimes \ket{f^{-1}_0(y)}_{\regX}.
\end{align*}
The measurement of $\regD$ in the Hadamard basis therefore results in $b$. 
\end{proof}

\smallskip
\noindent\textbf{Security.}
\begin{theorem}\label{thm:IND-CPA}
If $(\setup,\eval,\swap)$ satisfies conversion hardness, $(\keygen,\enc,\dec)$ is IND-CPA secure.
\end{theorem}
\ifnum\llncs=0
\input{proof_IND-CPA}.
\else
We can prove \Cref{thm:IND-CPA} by using \Cref{item:AAS_two} of \Cref{thm:AAS}.
We defer the proof to 
the full version
since the idea is already explained in \cref{sec:overview_QPKE}.
\fi


\subsection{Instantiation from Group Actions}\label{sec:STF_from_GA}

We review basic definitions about cryptographic group actions and their one-wayness and pseudorandomness
following \cite{TCC:JQSY19}. Then, we construct a STF based on it. 

\smallskip
\noindent\textbf{Basic definitions.}

\begin{definition}[Group actions]
Let $G$ be a (not necessarily abelian) group, $S$ be a set, and $\star:G\times S \rightarrow S$ be a function where we write $g\star s$ to mean $\star(g,s)$. We say that $(G,S,\star)$ is a \emph{group action} if it satisfies the following:
\begin{enumerate}
    \item For the identity element $e\in G$ and any $s\in S$, we have $e\star s=s$.
    \item For any $g,h\in G$ and any $s\in S$, we have $(gh)\star s = g\star(h\star s)$.
\end{enumerate}
\end{definition}

To be useful for cryptography, we have to at least assume that basic operations about $(G,S,\star)$ have efficient algorithms. We require the following efficient algorithms similarly to \cite{TCC:JQSY19}.
\begin{definition}[Group actions with efficient algorithms]\label{def:GA_efficient}
We say that a group action $(G,S,\star)$ has \emph{efficient algorithms} if it satisfies the following:\footnote{Strictly speaking, we have to consider a family $\{(G_\secp ,S_\secp ,\star_\secp)\}_{\secp \in \mathbb{N}}$ of group actions parameterized by the security parameter to meaningfully define the efficiency requirements. We omit the dependence on $\secp$ for notational simplicity throughout the paper.} 
\begin{itemize}
    \item[{\bf Unique representations:}] Each element of $G$ and $S$ can be represented as a bit string of length $\poly(\secp)$ in a unique manner. Thus, we identify these elements and their representations.
    \item[{\bf Group operations:}] There are classical deterministic polynomial-time algorithms that compute $gh$ from $g\in G$ and $h \in G$ and $g^{-1}$ from $g \in G$.
    \item[{\bf Group action:}] There is a classical deterministic polynomial-time algorithm that computes $g\star s$ from $g\in G$ and $s\in S$.
    \item[{\bf Efficient recognizability:}] There are classical deterministic polynomial-time algorithms that decide if a given bit string represents an element of $G$ or $S$, respectively.
    \item[{\bf Random sampling:}] There are PPT algorithms that sample almost uniform elements of $G$ or $S$ (i.e., the distribution of the sample is statistically close to the uniform distribution), respectively.
    \item[{\bf Superposition over $G$:}] There is a QPT algorithm that generates a state whose trace distance from $\ket{G}$ is $\negl(\secp)$. 
\end{itemize}
\end{definition}
\begin{remark}[A convention on ``Random sampling'' and ``Superposition over $G$'' properties]
In the rest of this paper, we assume that we can sample elements from \emph{exactly} uniform distributions of $G$ and $S$. Similarly, we assume that we can \emph{exactly} generate $\ket{G}$ in QPT. They are just for simplifying the presentations of our results, and all the results hold with the above imperfect version with additive negligible loss for security or correctness. 
\end{remark}

The above requirements are identical to those in \cite{TCC:JQSY19} except for the ``superposition over $G$'' property. We remark that all candidate constructions proposed in  \cite{TCC:JQSY19} satisfy this property as explained later. 

\smallskip
\noindent\textbf{Assumptions.}
We define one-wayness and pseudorandomness following \cite{TCC:JQSY19}.
\begin{definition}[One-wayness]\label{def:one-way}
We say that a group action $(G,S,\star)$ with efficient algorithms is \emph{one-way} if for any non-uniform QPT adversary $\A$, we have 
\begin{align*}
    \Pr\left[
    g'\star s = g\star s
    :
    \begin{array}{l}
     s\gets S,
    g\gets G,
    g' \leftarrow \A(s,g\star s)
    \end{array}
    \right]=\negl(\secp).
\end{align*}
\end{definition}

\begin{definition}[Pseudorandomness]\label{def:PR}
We say that a group action $(G,S,\star)$ with efficient algorithms is \emph{pseudorandom} if it satisfies the following:
\begin{enumerate}
\item \label{item:orbit}
We have 
\begin{align*}
\Pr[\exists g\in G\text{~s.t.~}g\star s =t:s,t\gets S]=\negl(\secp).
\end{align*}
\item \label{item:PR}
For any non-uniform QPT adversary $\A$, we have 
\begin{align*}
    \left|
    \Pr\left[
    1\gets \A(s,t):
     s\gets S,
    g\gets G,
    t\seteq g\star s
    \right]
    -
    \Pr\left[
    1\gets \A(s,t):
     s,t\gets S
    \right]
    \right|
    =
    \negl(\secp).
\end{align*}
\end{enumerate}
\end{definition}
\begin{remark}[On \Cref{item:orbit}]
We require \Cref{item:orbit} to make \Cref{item:PR} non-trivial. For example, if $(G,S,\star)$ is transitive, i.e., for any $s,t\in S$, there is $g\in G$ such that $g\star s =t$, \Cref{item:PR} trivially holds because the distributions of $t=g\star s$ is uniformly distributed over $S$ for any fixed $s$ and random $g\gets G$.  
\end{remark}
\begin{remark}[Pseudorandom $\rightarrow$ One-way]
We remark that the pseudorandomness immediately implies the one-wayness as noted in \cite{TCC:JQSY19}. 
\end{remark}

\noindent\textbf{Instantiations.}
Ji et al. \cite{TCC:JQSY19} gave several candidate constructions of one-way and pseudorandom group actions with efficient algorithms based on general linear group actions on tensors. 
We briefly describe one of their candidates below. 
Let $\mathbb{F}$ be a finite field, and $k$, $d_1,d_2...,d_k$ be positive integers (which are typically set as 
$k=3$ and $d_1=d_2=d_3$). 
We set $G\seteq \prod_{j=1}^{k}GL_{d_j}(\mathbb{F})$, $S\seteq \bigotimes_{j=1}^{k} \mathbb{F}^{d_j}$, and define the group action by the matrix-vector multiplication as
\begin{align*}
    (M_j)_{j\in[k]} \star T\seteq \left(\bigotimes_{j=1}^{k}M_j\right) T
\end{align*}
for $(M_j)_{j\in[k]}\in \prod_{j=1}^{k}GL_{d_j}(\mathbb{F})$ and $T\in \bigotimes_{j=1}^{k} \mathbb{F}^{d_j}$. 
See \cite{TCC:JQSY19} for attempts of cryptanalysis and justification of the one-wayness and pseudorandomness. 
We remark that we introduced an additional requirement of the ``superposition over $G$'' property in \Cref{def:GA_efficient}, but their candidates satisfy this property. 
In their candidates, the group $G$ is a direct product of general linear groups over finite fields (or symmetric groups for one of the candidates), and a uniformly random matrix over finite fields is invertible with overwhelming probability for appropriate parameters.

\smallskip\noindent\textbf{Construction of STF.}
We construct a STF based on group actions. Let $(G,S,\star)$ be a group action with efficient algorithms (as defined in \Cref{def:GA_efficient}). Then, we construct a STF as follows.

\begin{description}
\item[$\setup(1^\secp)$:] Generate $s_0\gets S$ and $g\gets G$, set $s_1\seteq g\star s_0$, and output $\pp\seteq (s_0,s_1)$ and $\td\seteq g$. 
For $b\in \bit$, we define $f_b:G\rightarrow S$ by 
    $f_b(h)\seteq h\star s_b$.
\item[$\eval(\pp=(s_0,s_1),b,h)$:] Output $f_b(h)=h\star s_b$. 
\item[$\swap(\td=g,b,h)$:] 
If $b=0$, output $hg^{-1}$. If $b=1$, output $hg$. 
\end{description}

The evaluation correctness is trivial. The swapping correctness can be seen as follows: For any $h\in G$, $f_1(\swap(\td,0,h))=f_1(hg^{-1})=(hg^{-1})\star s_1 = h \star s_0 = f_0(h)$. Similarly, for any $h\in G$,  $f_0(\swap(\td,1,h))=f_0(hg)=(hg)\star s_0 = h \star s_1 = f_1(h)$. For any $h\in G$, $\swap(\td,1,\swap(\td,0,h))=\swap(\td,1,hg^{-1})=(hg^{-1})g=h$. 

The efficient sampling and efficient superposition properties directly follow from the corresponding properties of the group action.

We prove the following theorem.
\begin{theorem}\label{thm:GA_to_STF}
The following hold:
\begin{enumerate}
    \item If $(G,S,\star)$ is one-way, then $(\setup,\eval,\swap)$ is claw-free. \label{item:OW_to_CF}
    \item If $(G,S,\star)$ is pseudorandom, then $(\setup,\eval,\swap)$ is conversion hard. \label{item:PR_to_CH}
\end{enumerate}
\end{theorem}
\ifnum\llncs=0 \input{proof_GA_to_STF}
\else
We defer the proof to the full version
because it is easy. 
\fi

\noindent\textbf{Quantum-ciphertext PKE from group actions.}
Recall that conversion hard STFs suffice for constructing IND-CPA secure quantum ciphertext PKE (\Cref{thm:IND-CPA}). 
Then, by \Cref{lem:collapsing_conversion_hard,lem:Gap_CF_and_CH,thm:GA_to_STF}, we obtain the following corollaries. 
\begin{corollary}
If there exists a pseudorandom group action with efficient algorithms, there exists an IND-CPA secure quantum-ciphertext PKE.
\end{corollary}
\begin{remark}[Lossy encryption]
Actually, we can show that the quantum-ciphertext PKE constructed from a pseudorandom group action is lossy encryption \cite{EC:BelHofYil09}, which is stronger than IND-CPA secure one. We omit the detail since our focus is on constructing IND-CPA secure schemes. 
\end{remark}
\begin{corollary}
If there exists a one-way group action with efficient algorithms such that $f_0$ is collapsing,\footnote{We currently have no candidate of such a one-way group action.} there exists a uniform IND-CPA secure quantum-ciphertext PKE scheme. 
\end{corollary}
\begin{corollary}
If there exists a one-way group action with efficient algorithms, there exists a uniform IND-CPA secure quantum-ciphertext PKE scheme or infinitely-often one-shot signatures.\footnote{The uniform IND-CPA security is defined similarly to the IND-CPA security in \Cref{def:IND-CPA} except that the adversary is restricted to be \emph{uniform} QPT.}
\end{corollary}

%% file: proof_collapsing_conversion_hard.tex
\begin{proof}
Suppose that $(\setup,\eval,\swap)$ does not satisfy conversion hardness. Then, there is a non-uniform QPT adversary $\A$ such that 
\begin{align*}
    \Pr[f_1(x_1)=y:(\pp,\td)\gets \setup(1^\secp),x_0\gets \calX, y\seteq f_0(x_0),  x_1\gets \A(\pp,\ket{f_0^{-1}(y)})]
\end{align*}
is non-negligible. 
By the assumption that $f_0$ is collapsing, we can show that 
\begin{align*}
   \left|\Pr[f_1(x_1)=y:(\pp,\td)\gets \setup(1^\secp),x_0\gets \calX, y\seteq f_0(x_0),  x_1\gets \A(\pp,\ket{f_0^{-1}(y)})]\right.\\
   \left.
   -
   \Pr[f_1(x_1)=y:(\pp,\td)\gets \setup(1^\secp),x_0\gets \calX, y\seteq f_0(x_0),  x_1\gets \A(\pp,\ket{x_0})]
   \right|\\
   =\negl(\secp).
\end{align*}

Combining the above,
\begin{align*}
    \Pr[f_1(x_1)=y:(\pp,\td)\gets \setup(1^\secp),x_0\gets \calX, y\seteq f_0(x_0),  x_1\gets \A(\pp,\ket{x_0})]
\end{align*}
is non-negligible. 

Then, we use $\A$ to construct a non-uniform QPT adversary $\B$ that breaks claw-freeness as follows.

\begin{description}
\item[$\B(\pp)$:] 
Pick $x_0\gets \calX$, run $x_1\gets \A(\pp,x_0)$, and output $(x_0,x_1)$. 
\end{description} 
Then, $\B$ breaks claw-freeness.
\end{proof}

%% file: proof_IND-CPA.tex
\begin{proof}
First, we remark that the IND-CPA security is identical to computational indistinguishability of the following two states $\ket{\psi_0}$ and $\ket{\psi_1}$ against any non-uniform QPT distinguisher $\A$ that does not act on $(\regY,\regPP')$:
\begin{align*}
    \ket{\psi_b}\seteq 
    \sum_{\pp}\sqrt{D(\pp)}\ket{\pp}_{\regPP}\ket{\pp}_{\regPP'}\sum_{x\in \calX}\frac{1}{\sqrt{2|\calX|}}(\ket{0}_{\regD}\ket{x}_{\regX}\ket{f_0(x)}_{\regY}+(-1)^b \ket{1}_{\regD}\ket{x}_{\regX}\ket{f_1(x)}_{\regY})
\end{align*}
where $D(\pp)\seteq \Pr[\pp'=\pp:(\pp',\td')\gets \setup(1^\secp)]$.

Suppose that there is a non-uniform QPT distinguisher $\A$ with an advice $\ket{\tau}_{\regZ}$ that does not act on $(\regY,\regPP')$ and distinguishes $\ket{\psi_0}$ and $\ket{\psi_1}$ with a non-negligible advantage $\Delta$. 
 
Since $\ket{\psi_0}$ and $\ket{\psi_1}$ are orthogonal, by \Cref{item:AAS_two} of \Cref{thm:AAS}, there exists a polynomial-time computable unitary $U$ over $(\regPP,\regD,\regX,\regZ)$ such that\footnote{For applying \Cref{item:AAS_two} of \Cref{thm:AAS}, we assume that $\A$ does not use an additional ancilla qubits besides $\ket{\tau}_\regZ$ w.l.o.g. (Sufficiently many qubits that are initialized to be $\ket{0}$ could be included in $\ket{\tau}_{\regZ}$.)}
 \begin{align*}
        \frac{1}{2}\left|
        \begin{array}{ll}
        &\bra{\psi'_1}_{\regPP,\regPP',\regD,\regX,\regY}\bra{\tau}_{\regZ}(U_{\regPP,\regD,\regX,\regZ}\otimes I_{\regPP',\regY})\ket{\psi'_0}_{\regPP,\regPP',\regD,\regX,\regY}\ket{\tau}_{\regZ}\\
    &+\bra{\psi'_0}_{\regPP,\regPP',\regD,\regX,\regY}\bra{\tau}_{\regZ}(U_{\regPP,\regD,\regX,\regZ}\otimes I_{\regPP',\regY})\ket{\psi'_1}_{\regPP,\regPP',\regD,\regX,\regY}\ket{\tau}_{\regZ}
    \end{array}
        \right|
        =\Delta
    \end{align*}
where 
\begin{align}
    \ket{\psi'_b}=\frac{\ket{\psi_0}+(-1)^b\ket{\psi_1}}{\sqrt{2}}=
   \sum_{\pp}\sqrt{\frac{D(\pp)}{|\calX|}}\ket{\pp}_{\regPP}\ket{\pp}_{\regPP'}\sum_{x\in \calX}\ket{b}_{\regD}\ket{x}_{\regX}\ket{f_b(x)}_{\regY}. \label{eq:def_psi_prime}
\end{align}
In the following, we simply write $U$ to mean $U_{\regPP,\regD,\regX,\regZ}$ for notational simplicity. 
Thus, we must have 
\begin{align}
|\bra{\psi'_1}_{\regPP,\regPP',\regD,\regX,\regY}\bra{\tau}_{\regZ}(U\otimes I_{\regPP',\regY})\ket{\psi'_0}_{\regPP,\regPP',\regD,\regX,\regY}\ket{\tau}_{\regZ}|\geq \Delta
\end{align}
or
\begin{align}
|\bra{\psi'_0}_{\regPP,\regPP',\regD,\regX,\regY}\bra{\tau}_{\regZ}(U\otimes I_{\regPP',\regY})\ket{\psi'_1}_{\regPP,\regPP',\regD,\regX,\regY}\ket{\tau}_{\regZ}|\geq \Delta. 
\end{align}
Without loss of the generality, we assume that the former inequality holds.
Then we show that this contradicts conversion hardness of $(\setup,\eval,\swap)$. 
We construct a non-uniform QPT adversary $\B$ that takes $\ket{\tau}$ as an advice and breaks the conversion hardness (Definition~\ref{def:conv_hard}) of $(\setup,\eval,\swap)$ as follows. 

\begin{description}
\item[$\B\left((\pp,\ket{f^{-1}_0(y)}_{\regX});\ket{\tau}_{\regZ}\right)$:] 
On input $(\pp,\ket{f^{-1}_0(y)}_{\regX})$ and a quantum advice $\ket{\tau}_{\regZ}$, prepare a single qubit register $\regD$ that is initialized to be $\ket{0}_{\regD}$,  
apply $U$ on $\ket{\pp}_{\regPP}\ket{0}_\regD\ket{f^{-1}_0(y)}_{\regX}\ket{\tau}_{\regZ}$, measure $\regX$ to obtain an outcome $x'$, and output $x'$. 
\end{description}

For any $\pp$, we have
\begin{align}
    &\Pr\left[f_1(x')=y :x\gets \calX,y\seteq f_0(x), x'\gets \B\left((\pp,\ket{f_0^{-1}(y)}_{\regX});\ket{\tau}_{\regZ}\right)\right] \notag \\
    &=\sum_{\substack{y\in \calY\\ x'\in f_1^{-1}(y)}  }
    \frac{|f_0^{-1}(y)|}{|\calX|}
    \left\|\bra{x'}_{\regX}U\ket{\pp}_{\regPP}\ket{0}_{\regD}\ket{f^{-1}_0(y)}_{\regX}\ket{\tau}_{\regZ}\right\|^2 \label{eq:qcPKE_Pr}\\
    &\geq \frac{1}{|\calX|}\left(\sum_{\substack{y\in \calY\\ x'\in f_1^{-1}(y)}  }
    \sqrt{\frac{|f_0^{-1}(y)|}{|\calX|}}
    \left\|\bra{x'}_{\regX}U\ket{\pp}_{\regPP}\ket{0}_{\regD}\ket{f^{-1}_0(y)}_{\regX}\ket{\tau}_{\regZ}\right\|\right)^2 \label{eq:qcPKE_CS}\\
    &\geq \frac{1}{|\calX|^2}\left\|\sum_{\substack{y\in \calY\\ x'\in f_1^{-1}(y)}  }
    \sqrt{|f_0^{-1}(y)|}
    \bra{x'}_{\regX}U\ket{\pp}_{\regPP}\ket{0}_{\regD}\ket{f^{-1}_0(y)}_{\regX}\ket{\tau}_{\regZ}\right\|^2 \label{eq:qcPKE_Tri}\\
&\geq \frac{1}{|\calX|^2}\left|\sum_{\substack{y\in \calY\\x\in f_0^{-1}(y)\\ x'\in f_1^{-1}(y)}  }
    \bra{\pp}_{\regPP}\bra{1}_{\regD}\bra{x'}_{\regX}\bra{\tau}_{\regZ}U\ket{\pp}_{\regPP}\ket{0}_{\regD}\ket{x}_{\regX}\ket{\tau}_{\regZ}\right|^2 \label{eq:qcPKE_Sum}\\
&= \frac{1}{|\calX|^2}
\left|
\begin{array}{l}
\left(\sum_{x'\in \calX}\bra{\pp}_{\regPP}\bra{\pp}_{\regPP'}\bra{1}_{\regD}\bra{x'}_{\regX}\bra{f_1(x')}_{\regY}\bra{\tau}_{\regZ}\right)\\
\left( U\otimes I_{\regPP',\regY}\right)
\left(\sum_{x\in \calX}\ket{\pp}_{\regPP}\ket{\pp}_{\regPP'}\ket{0}_{\regD}\ket{x}_{\regX}\ket{f_0(x)}_{\regY}\ket{\tau}_{\regZ}\right)
\end{array}
\right|^2  \label{eq:qcPKE_split},
\end{align}  
where \Cref{eq:qcPKE_Pr} follows from the definition of $\B$,  
\Cref{eq:qcPKE_CS} follows from Cauchy–Schwarz inequality and $\sum_{y\in \calY}|f^{-1}_1(y)|=|\calX|$, \Cref{eq:qcPKE_Tri} follows from the triangle inequality, and \Cref{eq:qcPKE_Sum} follows from the 
definition $\ket{f_0^{-1}(y)}= \frac{1}{|f_0^{-1}(y)|^{1/2}}\sum_{x\in f_0^{-1}(y)}\ket{x}$ and the
fact that inserting $\bra{\pp}_{\regPP}\bra{1}_{\regD}\bra{\tau}_{\regZ}$ can only decrease the norm. 

Therefore, we have 
\begin{align*}
    &\Pr[f_1(x')=y :(\pp,\td)\gets \setup(1^\secp), x\gets \calX,y\seteq f_0(x), x'\gets \B((\pp,\ket{f_0^{-1}(y)}_{\regX});\ket{\tau}_{\regZ})]\\
    &=\sum_{\pp}D(\pp)\left[\Pr[f_1(x')=y :x\gets \calX,y\seteq f_0(x), x'\gets \B((\pp,\ket{f_0^{-1}(y)}_{\regX});\ket{\tau}_{\regZ})]\right]\\
    &\geq 
    \sum_{\pp}
    \frac{D(\pp)}{|\calX|^2}
    \left|
\begin{array}{l}
\left(\sum_{x'\in \calX}\bra{\pp}_{\regPP}\bra{\pp}_{\regPP'}\bra{1}_{\regD}\bra{x'}_{\regX}\bra{f_1(x')}_{\regY}\bra{\tau}_{\regZ}\right)\\
\left( U\otimes I_{\regPP',\regY}\right)
\left(\sum_{x\in \calX}\ket{\pp}_{\regPP}\ket{\pp}_{\regPP'}\ket{0}_{\regD}\ket{x}_{\regX}\ket{f_0(x)}_{\regY}\ket{\tau}_{\regZ}\right)
\end{array}
\right|^2\\
     &\geq 
    \left|\sum_{\pp}
    \frac{D(\pp)}{|\calX|}\left(
\begin{array}{l}
\left(\sum_{x'\in \calX}\bra{\pp}_{\regPP}\bra{\pp}_{\regPP'}\bra{1}_{\regD}\bra{x'}_{\regX}\bra{f_1(x')}_{\regY}\bra{\tau}_{\regZ}\right)\\
\left( U\otimes I_{\regPP',\regY}\right)
\left(\sum_{x\in \calX}\ket{\pp}_{\regPP}\ket{\pp}_{\regPP'}\ket{0}_{\regD}\ket{x}_{\regX}\ket{f_0(x)}_{\regY}\ket{\tau}_{\regZ}\right)
\end{array}
\right)
\right|^2\\
    &= 
    \left|\bra{\psi'_1}_{\regPP,\regPP',\regD,\regX,\regY}\bra{\tau}_{\regZ}(U\otimes I_{\regPP',\regY})\ket{\psi'_0}_{\regPP,\regPP',\regD,\regX,\regY}\ket{\tau}_{\regZ}\right|^2,
\end{align*}
where the first inequality follows from \Cref{eq:qcPKE_split}, the second inequality follows from Jensen's inequality,
and the final equality follows from \Cref{eq:def_psi_prime}.  

This is non-negligible by our assumption. Therefore, $\B$ breaks the conversion hardness of $(\setup,\eval,\swap)$, which is a contradiction. Thus, $(\keygen,\enc,\dec)$ is IND-CPA secure.
\end{proof}

%% file: proof_GA_to_STF.tex
\begin{proof}~\\
\noindent\textbf{Proof of \Cref{item:OW_to_CF}.}
Suppose that $(\setup,\eval,\swap)$ is not claw-free. Then there is a non-uniform QPT adversary $\A$ such that 
\begin{align*}
    \Pr[f_0(h_0)=f_1(h_1):(\pp,\td)\gets \setup(1^\secp),(h_0,h_1)\gets \A(\pp)]
\end{align*}
is non-negligible. We use $\A$ to construct a non-uniform QPT adversary $\B$ that breaks one-wayness of $(G,S,\star)$ as follows:

\begin{description}
\item[$\B(s_0,s_1)$:] 
Set $\pp\seteq (s_0,s_1)$, run $(h_0,h_1)\gets \A(\pp)$, and outputs $h_1^{-1}h_0$.
\end{description}

By the assumption, we have $f_0(h_0)=f_1(h_1)$ with a non-negligible probability.
By the definition of $f_0$ and $f_1$, $f_0(h_0)=f_1(h_1)$ is equivalent to
$h_0\star s_0=h_1\star s_1$, which means $h_1^{-1}h_0\star s_0=s_1$. Since this occurs with a non-negligible probability $\B$ breaks one-wayness of $(G,S,\star)$. Thus,  $(\setup,\eval,\swap)$ is claw-free.

\smallskip
\noindent\textbf{Proof of \Cref{item:PR_to_CH}.}
Suppose that $(\setup,\eval,\swap)$ is not conversion hard. Then there is a non-uniform QPT algorithm $\A$ such that
\begin{align*}
    \Pr[f_1(x_1)=y:(\pp,\td)\gets \setup(1^\secp),x_0\gets \calX, y\seteq f_0(x_0),  x_1\gets \A(\pp,\ket{f_0^{-1}(y)})]
\end{align*}
is non-negligible. 
This is equivalent to that 
\begin{align*}
\Pr\left[
    h_0\star s_0 = h_1 \star s_1
    :
    \begin{array}{ll}
    &s_0\gets S,
    g,h_0\gets G,\\
    &s_1  \seteq g \star s_0, 
    y\seteq h_0\star s_0, \\
    &h_1 \leftarrow \A(s_0,s_1,\ket{f_0^{-1}(y)})
    \end{array}
    \right]
\end{align*}
is non-negligible.
On the other hand, by \Cref{item:orbit} of \Cref{def:PR}, we have 
\begin{align*}
  \Pr\left[
    h_0\star s_0 = h_1 \star s_1
    :
    \begin{array}{ll}
    &s_0,s_1\gets S,
    h_0\gets G,\\
    & y\seteq h_0\star s_0, \\
    &h_1 \leftarrow \A(s_0,s_1,\ket{f_0^{-1}(y)})
    \end{array}
    \right]= \negl(\secp).
\end{align*}
Therefore, 
\begin{align}\label{eq:distinguish_PR}
    \left|
        \begin{array}{ll}
    &\Pr\left[
    h_0\star s_0 = h_1 \star s_1
    :
    \begin{array}{ll}
    &s_0\gets S,
    g,h_0\gets G,\\
    &s_1  \seteq g \star s_0, 
    y\seteq h_0\star s_0, \\
    &h_1 \leftarrow \A(s_0,s_1,\ket{f_0^{-1}(y)})
    \end{array}
    \right]\\
    &-
  \Pr\left[
    h_0\star s_0 = h_1 \star s_1
    :
    \begin{array}{ll}
    &s_0,s_1\gets S,
    h_0\gets G,\\
    & y\seteq h_0\star s_0, \\
    &h_1 \leftarrow \A(s_0,s_1,\ket{f_0^{-1}(y)})
    \end{array}
    \right]
     \end{array}
     \right|
\end{align}
is non-negligible. 

Then, we construct the following non-uniform QPT adversary $\B$ that breaks pseudorandomness of $(G,S,\star)$:
\begin{description}
\item[$\B(s_0,s_1)$:] Generate a state $\frac{1}{\sqrt{|G|}}\sum_{h_0\in G}\ket{h_0}\ket{h_0\star s_0}$ and measure the second register to get $y\in S$. Then, the first register collapses to $\ket{f_0^{-1}(y)}$. Run $h_1\gets \A(s_0,s_1,\ket{f_0^{-1}(y)})$.
Output $1$ if $h_1\star s_1=y$ and otherwise $0$. 
\end{description}

We can see that $\B$'s advantage to distinguish the two cases ($s_0\gets S$, $g\gets G$, $s_1\seteq g\star s_0$ or $s_0,s_1\gets S$) is exactly  \Cref{eq:distinguish_PR}, which is non-negligible. This contradicts pseudorandomness of  $(G,S,\star)$ (\Cref{item:PR} of \Cref{def:PR}). 
Thus, $(\setup,\eval,\swap)$ is conversion hard.
\end{proof}

%% file: swap_and_distinguish.tex
\section{Equivalence between Swapping and Distinguishing with Auxiliary States} 
For our application to conversion for commitments, we need a generalization of \Cref{thm:AAS} that considers auxiliary quantum states. While it is straightforward to generalize \Cref{item:AAS_two} to such a setting,\footnote{Indeed, such a generalization of \Cref{item:AAS_two} is already implicitly used in the proof of \Cref{thm:IND-CPA}.} a generalization of \Cref{item:AAS_one} is non-trivial. The problems is that the unitary $U$ may not preserve the auxiliary state when it ``swaps'' $\ket{x}$ and $\ket{y}$.\footnote{This is also observed in \cite[Footnote 2]{AAS20}.} 
Intuitively, we overcome this issue by ``uncomputing'' the auxiliary state in a certain sense.

\begin{theorem}[Generalization of \Cref{thm:AAS} with auxiliary states] \label{thm:AAS_generalized}~
\begin{enumerate}
    \item \label{item:AAS_one_generalized}
    Let $\ket{x},\ket{y}$ be orthogonal $n$-qubit states and $\ket{\tau}$ be an $m$-qubit state. 
    Let $U$ be a polynomial-time computable unitary over $(n+m)$-qubit states and define $\Gamma$ as
    \begin{align*}
        \Gamma :=  
        \left\|(\bra{y}\otimes I^{\otimes m})U\ket{x}\ket{\tau}
        + 
        (\bra{x}\otimes I^{\otimes m})U\ket{y}\ket{\tau}
        \right\|.
    \end{align*} 
    Then, there exists a non-uniform QPT distinguisher $\A$ with advice $\ket{\tau'}=\ket{\tau}\otimes \frac{\ket{x}\ket{0}+\ket{y}\ket{1}}{\sqrt{2}}$ that distinguishes  
    $\ket{\psi}=\frac{\ket{x}+\ket{y}}{\sqrt{2}}$ and  $\ket{\phi}=\frac{\ket{x}-\ket{y}}{\sqrt{2}}$ with advantage $\frac{\Gamma^2}{4}$.
    Moreover, if $U$ does not act on some qubits, then $\A$ also does not act on those qubits.
    \item  \label{item:AAS_two_generalized}
    Let $\ket{\psi},\ket{\phi}$ be orthogonal $n$-qubit states, and suppose that a non-uniform QPT distinguisher $A$ with an $m$-qubit advice $\ket{\tau}$ 
    distinguishes $\ket{\psi}$ and $\ket{\phi}$ with advantage $\Delta$ without using additional ancilla qubits besides $\ket{\tau}$. 
    Then, there exists a polynomial-time computable unitary $U$ over $(n+m)$-qubit states such that 
    \begin{align*}
        \frac{|\bra{y}\bra{\tau}U\ket{x}\ket{\tau}+\bra{x}\bra{\tau}U\ket{y}\ket{\tau}|}{2}=\Delta
    \end{align*}
       where  $|x\rangle\coloneqq\frac{|\psi\rangle+|\phi\rangle}{\sqrt{2}}$ and $|y\rangle\coloneqq\frac{|\psi\rangle-|\phi\rangle}{\sqrt{2}}$.
     Moreover, if $\A$ does not act on some qubits, then $U$ also does not act on those qubits.
\end{enumerate}
\end{theorem}
\begin{remark}
We remark that  \Cref{item:AAS_one_generalized}  does \emph{not} preserve the auxiliary state unlike \Cref{item:AAS_two_generalized}. Though this does not capture the intuition that ``one can distinguish $\ket{\psi}$ and $\ket{\phi}$ whenever he can swap $\ket{x}$ and $\ket{y}$'', 
this is good enough for our purpose. 
We also remark that there is a quadratic reduction loss in  \Cref{item:AAS_one_generalized}. We do not know if it is tight while both items of \Cref{thm:AAS} is shown to be tight in \cite{AAS20}.
\end{remark}
\begin{proof}[Proof of \Cref{thm:AAS_generalized}]
\Cref{item:AAS_two_generalized} directly follows from \Cref{item:AAS_two} of \Cref{thm:AAS} by considering $\ket{x}\ket{\tau}$ and $\ket{y}\ket{\tau}$ as $\ket{x}$ and $\ket{y}$ in \Cref{thm:AAS}.  
We prove \Cref{item:AAS_one_generalized} below.

\smallskip
\noindent\textbf{Proof of  \Cref{item:AAS_one_generalized}.}
Let $\regA$ and $\regA'$ be $n$-qubit registers, $\regZ$ be an $m$-qubit register, and $\regB$ be a $1$-qubit register. 
We define a unitary $\widetilde{U}$ over $(\regA,\regZ,\regA',\regB)$ as follows: 
\begin{align} \label{eq:def_tildeU}
    \widetilde{U}:=X_{\regB}U^{\dagger}_{\regA',\regZ}U_{\regA,\regZ}
\end{align}
where 
$X_{\regB}$ is the Pauli $X$ operator on $\regB$ and
$U^{\dagger}_{\regA',\regZ}$ means the inverse of $U_{\regA',\regZ}$, which works similarly to $U_{\regA,\regZ}$ except that it acts on $\regA'$ instead of on $\regA$. 

Then, we prove the following claim. 
\begin{MyClaim}\label{cla:for_AAS_one_generalized}
Let $\ket{x},\ket{y},\ket{\tau}$, and $\Gamma$ be as in \Cref{item:AAS_one_generalized} of \Cref{thm:AAS_generalized}, $\widetilde{U}$ be as defined in \Cref{eq:def_tildeU}, and
$\ket{\sigma}_{\regA',\regB}$ be the state over $(\regA',\regB)$ defined as follows:
\begin{align} \label{eq:def_sigma}
    \ket{\sigma}_{\regA',\regB}:=\frac{\ket{x}_{\regA'}\ket{0}_{\regB}+\ket{y}_{\regA'}\ket{1}_{\regB}}{\sqrt{2}}.
\end{align}
Then, it holds that 
\begin{align*}
\left|\bra{y}_{\regA}\bra{\tau}_{\regZ}\bra{\sigma}_{\regA',\regB}
    \widetilde{U}\ket{x}_{\regA}\ket{\tau}_{\regZ}\ket{\sigma}_{\regA',\regB}
+
\bra{x}_{\regA}\bra{\tau}_{\regZ}\bra{\sigma}_{\regA',\regB}
    \widetilde{U}\ket{y}_{\regA}\ket{\tau}_{\regZ}\ket{\sigma}_{\regA',\regB}\right|
    =\frac{\Gamma^2}{2}.
\end{align*}
\end{MyClaim}
We first finish the proof of \Cref{item:AAS_one_generalized} assuming that \Cref{cla:for_AAS_one_generalized} is correct. 
By \Cref{item:AAS_one} of \Cref{thm:AAS}, \Cref{cla:for_AAS_one_generalized} implies that there is a QPT distinguisher $\widetilde{\A}$ that distinguishes 
\begin{align*}
\ket{\widetilde{\psi}}=\frac{(\ket{x}+\ket{y})_{\regA}\ket{\tau}_{\regZ}\ket{\sigma}_{\regA',\regB}}{\sqrt{2}}
\end{align*}
and 
\begin{align*}
\ket{\widetilde{\phi}}=\frac{(\ket{x}-\ket{y})_{\regA}\ket{\tau}_{\regZ}\ket{\sigma}_{\regA',\regB}}{\sqrt{2}}
\end{align*}
with advantage $\frac{\Gamma^2}{4}$. 
Moreover, $\widetilde{\A}$ does not act on qubits on which $\widetilde{U}$ does not act. 
In particular, $\widetilde{\A}$ does not act on qubits of $\regA$ and $\regZ$ on which $U_{\regA,\regZ}$ does not act since $\widetilde{U}$ acts on $\regA$ and $\regZ$ only through $U_{\regA,\regZ}$ and $U^{\dagger}_{\regA',\regZ}$. 
Thus, by considering $\widetilde{\A}$ as a distinguisher $\A$ with advice $\ket{\tau'}=\ket{\tau}_{\regZ}\ket{\sigma}_{\regA',\regB}$ that distinguishes $\ket{\psi}=\frac{\ket{x}+\ket{y}}{\sqrt{2}}$ and  $\ket{\phi}=\frac{\ket{x}-\ket{y}}{\sqrt{2}}$, 
\Cref{item:AAS_one_generalized} is proven. 
Below, we prove \Cref{cla:for_AAS_one_generalized}.
\begin{proof}[Proof of \Cref{cla:for_AAS_one_generalized}]
For $(a,b)\in \{(x,x),(x,y),(y,x),(y,y)\}$, we define 
\begin{align*}
    \ket{\tau'_{ab}}_{\regZ}:=(\bra{b}_{\regA}\otimes I_{\regZ}) U_{\regA,\regZ}\ket{a}_{\regA}\ket{\tau}_{\regZ}.
\end{align*}
Then, we have 
\begin{align}\label{eq:Gamma}
    \Gamma=\left\|\ket{\tau'_{xy}}_\regZ+\ket{\tau'_{yx}}_\regZ\right\|
\end{align}
and 
\begin{align}
    &U_{\regA,\regZ}\ket{x}_\regA\ket{\tau}_\regZ=\ket{x}_{\regA}\ket{\tau'_{xx}}_{\regZ}+\ket{y}_{\regA}\ket{\tau'_{xy}}_{\regZ}+ \ket{\garbage_x}_{\regA,\regZ} \label{eq:U_x}\\
    &U_{\regA,\regZ}\ket{y}_\regA\ket{\tau}_\regZ=\ket{x}_{\regA}\ket{\tau'_{yx}}_{\regZ}+\ket{y}_{\regA}\ket{\tau'_{yy}}_{\regZ}+ \ket{\garbage_y}_{\regA,\regZ} \label{eq:U_y}
\end{align}
where $\ket{\garbage_x}_{\regA,\regZ}$ and $\ket{\garbage_y}_{\regA,\regZ}$ are (not necessarily normalized) states such that 
\begin{align}
    &(\bra{x}_{\regA}\otimes I_{\regZ})\ket{\garbage_x}_{\regA,\regZ}= (\bra{y}_{\regA}\otimes I_{\regZ})\ket{\garbage_x}_{\regA,\regZ}=0,  \label{eq:orthogonal_garbage_x}
\\
&(\bra{x}_{\regA}\otimes I_{\regZ})\ket{\garbage_y}_{\regA,\regZ}= (\bra{y}_{\regA}\otimes I_{\regZ})\ket{\garbage_y}_{\regA,\regZ}=0.
\label{eq:orthogonal_garbage_y}
\end{align}
Then, 
\begin{align}
\begin{split} \label{eq:ip_one}
&\bra{y}_{\regA}\bra{\tau}_{\regZ}\bra{\sigma}_{\regA',\regB}
    \widetilde{U}\ket{x}_{\regA}\ket{\tau}_{\regZ}\ket{\sigma}_{\regA',\regB}\\
=
&\bra{y}_{\regA}\bra{\tau}_{\regZ}\bra{\sigma}_{\regA',\regB}X_{\regB}U^{\dagger}_{\regA',\regZ}(\ket{x}_{\regA}\ket{\tau'_{xx}}_{\regZ}+\ket{y}_{\regA}\ket{\tau'_{xy}}_{\regZ}+ \ket{\garbage_x}_{\regA,\regZ})\ket{\sigma}_{\regA',\regB}\\
=&\bra{\tau}_{\regZ}\bra{\sigma}_{\regA',\regB}X_{\regB}U^{\dagger}_{\regA',\regZ}\ket{\tau'_{xy}}_{\regZ}\ket{\sigma}_{\regA',\regB}
\end{split}
\end{align}
where the first equality follows from \Cref{eq:U_x} and the second equality follows from \Cref{eq:orthogonal_garbage_x} and the assumption that $\ket{x}$ and $\ket{y}$ are orthogonal.
By \Cref{eq:def_sigma,eq:U_x,eq:U_y}, 
it holds that 
\begin{align}
\begin{split} \label{eq:UZtausigma}
    &
    U_{\regA',\regZ}X_{\regB}\ket{\tau}_{\regZ}\ket{\sigma}_{\regA',\regB}\\
=&
U_{\regA',\regZ}\frac{\ket{\tau}_{\regZ}((\ket{x}_{\regA'}\ket{1}_{\regB}+\ket{y}_{\regA'}\ket{0}_{\regB})}{\sqrt{2}}\\
=&\frac{1}{\sqrt{2}}\left(
\begin{array}{l}
\left(\ket{x}_{\regA'}\ket{\tau'_{xx}}_{\regZ}+\ket{y}_{\regA'}\ket{\tau'_{xy}}_{\regZ}+ \ket{\garbage_x}_{\regA',\regZ}\right)\ket{1}_\regB\\
+\left(\ket{x}_{\regA'}\ket{\tau'_{yx}}_{\regZ}+\ket{y}_{\regA'}\ket{\tau'_{yy}}_{\regZ}+ \ket{\garbage_y}_{\regA',\regZ}\right)\ket{0}_\regB
\end{array}
\right).
\end{split}
\end{align}
Then, it holds that 
\begin{align}
\begin{split}\label{eq:ip_two}
  &\bra{\tau}_{\regZ}\bra{\sigma}_{\regA',\regB}X_{\regB}U^{\dagger}_{\regA',\regZ}\ket{\tau'_{xy}}_{\regZ}\ket{\sigma}_{\regA',\regB}\\
  =&\frac{1}{2}\left(
\begin{array}{l}
\left(\bra{x}_{\regA'}\bra{\tau'_{xx}}_{\regZ}+\bra{y}_{\regA'}\bra{\tau'_{xy}}_{\regZ}+ \bra{\garbage_x}_{\regA',\regZ}\right)\bra{1}_\regB\\
+\left(\bra{x}_{\regA'}\bra{\tau'_{yx}}_{\regZ}+\bra{y}_{\regA'}\bra{\tau'_{yy}}_{\regZ}+ \bra{\garbage_y}_{\regA',\regZ}\right)\bra{0}_\regB
\end{array}
\right)
\left(\ket{x}_{\regA'}\ket{0}_{\regB}+\ket{y}_{\regA'}\ket{1}_{\regB}\right)\ket{\tau'_{xy}}_{\regZ}
\\
  =&\frac{1}{2}
  (\bra{\tau'_{xy}}+\bra{\tau'_{yx}})_{\regZ}\ket{\tau'_{xy}}_{\regZ},
 \end{split}
\end{align}
where the first equality follows from  \Cref{eq:def_sigma,eq:UZtausigma} and the second equality follows from \Cref{eq:orthogonal_garbage_x,eq:orthogonal_garbage_y} and the assumption that $\ket{x}$ and $\ket{y}$ are orthogonal.

By \Cref{eq:ip_one,eq:ip_two}, we have 
\begin{align}\label{eq:ip_x}
    \bra{y}_{\regA}\bra{\tau}_{\regZ}\bra{\sigma}_{\regA',\regB}
    \widetilde{U}\ket{x}_{\regA}\ket{\tau}_{\regZ}\ket{\sigma}_{\regA',\regB}=\frac{1}{2}
  (\bra{\tau'_{xy}}+\bra{\tau'_{yx}})_{\regZ}\ket{\tau'_{xy}}_{\regZ}.
\end{align}
By a similar calculation, we have 
\begin{align}\label{eq:ip_y}
    \bra{x}_{\regA}\bra{\tau}_{\regZ}\bra{\sigma}_{\regA',\regB}
    \widetilde{U}\ket{y}_{\regA}\ket{\tau}_{\regZ}\ket{\sigma}_{\regA',\regB}=\frac{1}{2}
  (\bra{\tau'_{xy}}+\bra{\tau'_{yx}})_{\regZ}\ket{\tau'_{yx}}_{\regZ}.
\end{align}
By \Cref{eq:ip_x,eq:ip_y}, we have 
\begin{align*}
   & \bra{y}_{\regA}\bra{\tau}_{\regZ}\bra{\sigma}_{\regA',\regB}
    \widetilde{U}\ket{x}_{\regA}\ket{\tau}_{\regZ}\ket{\sigma}_{\regA',\regB}
+
\bra{x}_{\regA}\bra{\tau}_{\regZ}\bra{\sigma}_{\regA',\regB}
    \widetilde{U}\ket{y}_{\regA}\ket{\tau}_{\regZ}\ket{\sigma}_{\regA',\regB}\\
    =&\frac{1}{2}\left\|\ket{\tau'_{xy}}_\regZ+\ket{\tau'_{yx}}_\regZ\right\|^2.
\end{align*}
By combining the above with \Cref{eq:Gamma}, we complete the proof of \Cref{cla:for_AAS_one_generalized}. 
\end{proof}
This completes the proof of \Cref{thm:AAS_generalized}.
\end{proof}

%% file: conversion.tex
\section{Our Conversion for Commitments}\label{sec:conversion}
\takashi{The new version is commented out for now.}

In this section, we give a conversion for canonical quantum bit commitments that converts the flavors of security using \Cref{thm:AAS_generalized}.

\begin{theorem}[Converting Flavors]
\label{thm:conversion}
Let  $\{Q_0,Q_1\}$ be a canonical quantum bit commitment scheme. 
Let $\{Q'_0,Q'_1\}$ be a canonical quantum bit commitment scheme described as follows:
\begin{itemize}
 \item The roles of commitment and reveal registers are swapped from $\{Q_0,Q_1\}$ and the commitment register is augmented by an additional one-qubit register. 
  That is, if $\regC$ and $\regR$ are the commitment and reveal registers of  $\{Q_0,Q_1\}$,  then the commitment and reveal registers of $\{Q'_0,Q'_1\}$ are defined as $\regC':= (\regR,\regD)$ and $\regR':= \regC$ where $\regD$ is a one-qubit register.  
  \item For $b\in \bit$, the unitary $Q_b'$ is defined as follows:
  \begin{align*}
    Q_b'\seteq \left(Q_0\otimes \ket{0}\bra{0}_{\regD}+Q_1\otimes \ket{1}\bra{1}_{\regD}\right)\left(I_{\regR,\regC}\otimes Z^b_{\regD}H_{\regD}\right)
\end{align*}
where $Z_{\regD}$ and $H_{\regD}$ denote the Pauli $Z$ and the Hadamard operators on $\regD$. 
\end{itemize}


Then, the following hold for $\mathtt{X},\mathtt{Y}\in \{\text{computationally,statistically,perfectly}\}$:
\begin{enumerate}
   \item \label{item:hiding_to_binding}
   If $\{Q_0,Q_1\}$ is $\mathtt{X}$ hiding, then $\{Q'_0,Q'_1\}$ is $\mathtt{X}$ binding. 
    \item \label{item:binding_to_hiding}
    If $\{Q_0,Q_1\}$ is $\mathtt{Y}$ binding, then $\{Q'_0,Q'_1\}$ is $\mathtt{Y}$ hiding.  
\end{enumerate}
\end{theorem}

Note that we have 
\begin{align*}
    Q_b'\ket{0}_{\regC',\regR'}= \frac{1}{\sqrt{2}}\left((Q_{0}\ket{0})_{\regC,\regR}\ket{0}_{\regD}
    +(-1)^{b} (Q_{1}\ket{0})_{\regC,\regR}\ket{1}_{\regD} \right)
\end{align*}
for $b\in \bit$ 
where $(\regC',\regR')$ is rearranged as $(\regC,\regR,\regD)$.  

\ifnum\llncs=0
\input{proof_conversion}
\else
We defer the proof of \cref{thm:conversion} to the full version
since it easily follows from \cref{thm:AAS_generalized} as explained in \cref{sec:overview_conversion}. 
\fi


\ifnum\llncs=1
\smallskip\noindent\textbf{Applications.} 
We give applications of \cref{thm:conversion} in the full version. 
\fi

%% file: proof_conversion.tex
\ifnum\llncs=1
\section{Proof of \Cref{thm:conversion}}\label{sec:proof_conversion}
\fi
\begin{proof}[Proof of \Cref{thm:conversion}.]~
Since the proofs are almost identical for all the three cases of $\mathtt{X},\mathtt{Y}\in \{\text{computationally,statistically,}$  $\text{perfectly}\}$, we focus on the case of $\mathtt{X}=\mathtt{Y}=``\text{computationally}"$.

\smallskip
\noindent\textbf{Proof of \Cref{item:hiding_to_binding}.}
Suppose that $\{Q'_0,Q'_1\}$ is not computationally binding. Then, there exists a polynomial-time computable unitary $U$ over $\regR'=\regC$  and an ancillary register $\regZ$ and a state $\ket{\tau}_{\regZ}$ such that 
\begin{align*}
    \left\|((\bra{0}{Q_1'}^\dagger)_{\regC',\regR'}\otimes I_{\regZ})(I_{\regC'}\otimes U_{\regR',\regZ})((Q'_0\ket{0})_{\regC',\regR'}\ket{\tau}_{\regZ})\right\|
\end{align*}
is non-negligible. 
We observe that $U$ does not act on $\regD$ (since that is not part of the reveal register $\regR'$ of $\{Q'_0,Q'_1\}$), and thus it cannot cause any interference between states that take $0$ and $1$ in $\regD$. 
Based on this observation and the denifition of $\{Q'_0,Q'_1\}$, we have 
\begin{align*}
    &((\bra{0}{Q_1'}^\dagger)_{\regC',\regR'}\otimes I_{\regZ})(I_{\regC'}\otimes U_{\regR',\regZ})((Q'_0\ket{0})_{\regC',\regR'}\ket{\tau}_{\regZ})\\
    =&
    \frac{1}{2}\left(
    \begin{array}{l}
    ((\bra{0}Q_0^\dagger)_{\regC,\regR}\bra{0}_{\regD}\otimes I_{\regZ})(I_{\regR,\regD}\otimes U_{\regC,\regZ})((Q_0\ket{0})_{\regC,\regR}\ket{0}_{\regD}\ket{\tau}_{\regZ})\\
    -((\bra{0}Q_1^\dagger)_{\regC,\regR}\bra{1}_{\regD}\otimes I_{\regZ})(I_{\regR,\regD}\otimes U_{\regC,\regZ})((Q_1\ket{0})_{\regC,\regR}\ket{1}_{\regD}\ket{\tau}_{\regZ})
    \end{array}
    \right).
\end{align*}
Similarly, we have 
\begin{align*}
    &((\bra{0}{Q_0'}^\dagger)_{\regC',\regR'}\otimes I_{\regZ})(I_{\regC'}\otimes U_{\regR',\regZ})((Q'_1\ket{0})_{\regC',\regR'}\ket{\tau}_{\regZ})\\
    =&
    \frac{1}{2}\left(
    \begin{array}{l}
    ((\bra{0}Q_0^\dagger)_{\regC,\regR}\bra{0}_{\regD}\otimes I_{\regZ})(I_{\regR,\regD}\otimes U_{\regC,\regZ})((Q_0\ket{0})_{\regC,\regR}\ket{0}_{\regD}\ket{\tau}_{\regZ})\\
    -((\bra{0}Q_1^\dagger)_{\regC,\regR}\bra{1}_{\regD}\otimes I_{\regZ})(I_{\regR,\regD}\otimes U_{\regC,\regZ})((Q_1\ket{0})_{\regC,\regR}\ket{1}_{\regD}\ket{\tau}_{\regZ})
    \end{array}
    \right).
\end{align*}
In particular, 
\begin{align*}
    &((\bra{0}{Q_1'}^\dagger)_{\regC',\regR'}\otimes I_{\regZ})(I_{\regC'}\otimes U_{\regR',\regZ})((Q'_0\ket{0})_{\regC',\regR'}\ket{\tau}_{\regZ})\\
    =&((\bra{0}{Q_0'}^\dagger)_{\regC',\regR'}\otimes I_{\regZ})(I_{\regC'}\otimes U_{\regR',\regZ})((Q'_1\ket{0})_{\regC',\regR'}\ket{\tau}_{\regZ}).\\
\end{align*}

Therefore, 
\begin{align*}
    \left\|
    \begin{array}{l}
    ((\bra{0}{Q_1'}^\dagger)_{\regC',\regR'}\otimes I_{\regZ})(I_{\regC'}\otimes U_{\regR',\regZ})((Q'_0\ket{0})_{\regC',\regR'}\ket{\tau}_{\regZ})\\
    +
    ((\bra{0}{Q_0'}^\dagger)_{\regC',\regR'}\otimes I_{\regZ})(I_{\regC'}\otimes U_{\regR',\regZ})((Q'_1\ket{0})_{\regC',\regR'}\ket{\tau}_{\regZ})
    \end{array}
    \right\|
\end{align*}
is non-negligible. 
If we set $\ket{x}=Q'_0\ket{0}_{\regC',\regR'}$ and $\ket{y}=Q'_1\ket{0}_{\regC',\regR'}$, then $\ket{x}$ and $\ket{y}$ are orthogonal. 
Then, by \Cref{item:AAS_one_generalized} of \Cref{thm:AAS_generalized}, there exists a non-uniform QPT distinguisher $\A$ with a polynomial-size advice $\ket{\tau'}$ that does not act on $\regC'=(\regR,\regD)$ and distinguishes 
\begin{align*}
    \ket{\psi}=\frac{\ket{x}+\ket{y}}{\sqrt{2}}=(Q_0\ket{0})_{\regC,\regR}\ket{0}_{\regD}
\end{align*}
and 
\begin{align*}
    \ket{\phi}=\frac{\ket{x}-\ket{y}}{\sqrt{2}}=(Q_1\ket{0})_{\regC,\regR}\ket{1}_{\regD}
\end{align*}
with a non-negligible advantage.
This means that the computational hiding property of $\{Q_0,Q_1\}$ is broken, which contradicts the assumption. 
Thus, $\{Q'_0,Q'_1\}$ is computationally binding.
This completes the proof of  \Cref{item:hiding_to_binding}.

\smallskip
\noindent\textbf{Proof of \Cref{item:binding_to_hiding}.}
  Suppose that $\{Q'_0,Q'_1\}$ is not computationally hiding. Then, there exists a non-uniform QPT distinguisher $\A$ with advice $\ket{\tau}_\regZ$ that does not act on $\regR'=\regC$ and
distinguishes 
\begin{align*}
    \ket{\psi}:=\frac{1}{\sqrt{2}}\left((Q_{0}\ket{0})_{\regC,\regR}\ket{0}_{\regD}+ (Q_{1}\ket{0})_{\regC,\regR}\ket{1}_{\regD} \right) 
\end{align*}
and 
\begin{align*}
    \ket{\phi}:=\frac{1}{\sqrt{2}}\left((Q_{0}\ket{0})_{\regC,\regR}\ket{0}_{\regD}- (Q_{1}\ket{0})_{\regC,\regR}\ket{1}_{\regD} \right) 
\end{align*}
with a non-negligible advantage $\Delta$. 

Since $\ket{\psi}$ and $\ket{\phi}$ are orthogonal, by \Cref{item:AAS_two_generalized} of \Cref{thm:AAS_generalized}, there exists a polynomial-time computable unitary $U$ over $(\regR,\regD,\regZ)$ such that
 \begin{align*}
        \frac{|\bra{y}_{\regC,\regR,\regD}\bra{\tau}_{\regZ}(U_{\regR,\regD,\regZ}\otimes I_{\regC})\ket{x}_{\regC,\regR,\regD}\ket{\tau}_{\regZ}+\bra{x}_{\regC,\regR,\regD}\bra{\tau}_{\regZ}(U_{\regR,\regD,\regZ}\otimes I_{\regC})\ket{y}_{\regC,\regR,\regD}\ket{\tau}_{\regZ}|}{2}=\Delta
    \end{align*}
where 
\begin{align*}
    \ket{x}=\frac{\ket{\psi}+\ket{\phi}}{\sqrt{2}}=
    (Q_{0}\ket{0})_{\regC,\regR}\ket{0}_{\regD} 
\end{align*}
and 
\begin{align*}
    \ket{y}=\frac{\ket{\psi}-\ket{\phi}}{\sqrt{2}}=
    (Q_{1}\ket{0})_{\regC,\regR}\ket{1}_{\regD}.
\end{align*}

Thus, we must have $|\bra{y}_{\regC,\regR,\regD}\bra{\tau}_{\regZ}(U_{\regR,\regD,\regZ}\otimes I_{\regC})\ket{x}_{\regC,\regR,\regD}\ket{\tau}_{\regZ}|\geq \Delta$ or $|\bra{x}_{\regC,\regR,\regD}\bra{\tau}_{\regZ}(U_{\regR,\regD,\regZ}\otimes I_{\regC})\ket{y}_{\regC,\regR,\regD}\ket{\tau}_{\regZ}|\geq \Delta$. We assume the former w.l.o.g., i.e., we have 
\begin{align*}
\left|
\left(
(\bra{0}Q_{1}^\dagger)_{\regC,\regR}\bra{1}_{\regD}\bra{\tau}_{\regZ}
\right)
(U_{\regR,\regD,\regZ}\otimes I_{\regC})
    \left((Q_{0}\ket{0})_{\regC,\regR}\ket{0}_{\regD} \ket{\tau}_{\regZ}\right)\right|\geq \Delta.
\end{align*}

In particular, we have 
\begin{align*}
    \left\|(Q_1\ket{0}\bra{0}Q_1^\dagger)_{\regC,\regR}(U_{\regR,\regD,\regZ}\otimes I_{\regC})((Q_0\ket{0})_{\regC,\regR}\ket{0}_{\regD} \ket{\tau}_{\regZ})\right\|\ge\Delta.
\end{align*}

This means that $U$ with the auxiliary state $\ket{0}_{\regD} \ket{\tau}_{\regZ}$ breaks the computational binding property of $\{Q_0,Q_1\}$, which contradicts the assumption. 
Thus, $\{Q'_0,Q'_1\}$ is computationally hiding.
This completes the proof of \Cref{item:binding_to_hiding}.

\end{proof} 

%% file: application.tex
\section{Applications of Our Conversion}\label{sec:application}
In this section, we show applications of our conversion (\Cref{thm:conversion}). More applications can be found in \Cref{sec:more_application}.

When we describe a canonical quantum bit commitment scheme $\{Q_0,Q_1\}$, we only describe how $Q_0$ and $Q_1$ act on $\ket{0}$ for simplicity. Quantum circuits that implement $Q_0$ and $Q_1$ can be defined in a natural way. 
\takashi{If this is ambiguous, we may write that in Appendix.}

\subsection{Construction from PRG}\label{sec:construction_Naor}
Naor~\cite{JC:Naor91} constructed a two-message \emph{classical} commitment scheme that is computationally hiding and statistically binding based on PRGs. 
Yan et al.~\cite{YWLQ15} constructed a quantum \emph{non-interactive} version of Naor's commitment.\footnote{Yan~\cite[Appendix C]{AC:Yan22} shows an alternative more direct translation of Naor's commitment to the quantum setting. We could also apply our conversion to that scheme, but we focus on the scheme of \cite{YWLQ15} since that is simpler.}
Let $G:\bit^n\rightarrow \bit^{3n}$ be a PRG.  
Then Yan et al.'s commitment scheme $\{Q_{\ywlq,0},Q_{\ywlq,1}\}$ is described as follows:

  \begin{align*}
    &Q_{\ywlq,0}\ket{0}_{\regC,\regR}:=\frac{1}{\sqrt{2^n}}\sum_{x\in \bit^n}\ket{G(x)}_{\regC}\ket{x, 0^{2n}}_{\regR}\\
    &Q_{\ywlq,1}\ket{0}_{\regC,\regR}:=\frac{1}{\sqrt{2^{3n}}}\sum_{y\in \bit^{3n}}\ket{y}_{\regC}\ket{y}_{\regR}.
\end{align*}
Yan et al.~\cite{YWLQ15} proved the following theorem.
\begin{theorem}[{\cite{YWLQ15}}]\label{thm:YWLQ}
If $G$ is a PRG, $\{Q_{\ywlq,0},Q_{\ywlq,1}\}$ is computationally hiding and statistically binding. 
\end{theorem}

By applying our conversion to $\{Q_{\ywlq,0},Q_{\ywlq,1}\}$, we obtain the following scheme $\{Q'_{\ywlq,0},Q'_{\ywlq,1}\}$.
  \begin{align*}
    &Q'_{\ywlq,b}\ket{0}_{\regC',\regR'}:=\frac{1}{\sqrt{2^{n+1}}}\sum_{x\in \bit^n}\ket{0,x, 0^{2n}}_{\regC'}\ket{G(x)}_{\regR'}+(-1)^b\frac{1}{\sqrt{2^{3n+1}}}\sum_{y\in \bit^{3n}}\ket{1,y}_{\regC'}\ket{y}_{\regR'}.
\end{align*}

By \Cref{thm:conversion,thm:YWLQ}, we obtain the following theorem.
\begin{theorem}
If $G$ is a PRG, $\{Q'_{\ywlq,0},Q'_{\ywlq,1}\}$ is statistically hiding and computationally binding. 
\end{theorem}

This is the first statistically hiding and computationally binding quantum bit commitment scheme from PRG that makes only a single call to the PRG. 
If we apply existing conversions~\cite{EC:CreLegSal01,AC:Yan22} to $\{Q_{\ywlq,0},Q_{\ywlq,1}\}$ (or other PRG-based schemes), they result in schemes that make $\Omega(\secp^2)$ calls to the PRG.

Note that it is known that PRG exists assuming the existence of one-way functions~\cite{SIAM:HILL99}.\footnote{Though the original security proof in \cite{SIAM:HILL99} only considers classical adversaries, it also works against quantum adversaries as well assuming quantum-secure one-way functions.} 
In the current state of the art, a construction of PRG makes at least $\Omega(\secp^3)$ calls to the base one-way function~\cite{STOC:HaiReiVad10,STOC:VadZhe12}. 
Thus, if we construct a PRG from a one-way function and count the number of calls to the one-way function,  $\{Q'_{\ywlq,0},Q'_{\ywlq,1}\}$ makes $\Omega(\secp^3)$ calls to the one-way function. We observe that this is asymptotically the same number as that of Koshiba and Odaira~\cite{KO11}. However, it does not seem possible to instantiate the scheme of~\cite{KO11} with a single call to a PRG instead of $\Omega(\secp^3)$ calls to a one-way function. Also, our security analysis is much simpler than theirs once we establish \Cref{thm:conversion}.

\subsection{Construction from Pseudorandom State Generators}\label{sec:construction_PRS}
Ananth, Qian, Yuen \cite{C:AnaQiaYue22}, and Morimae and Yamakawa \cite{C:MorYam22} concurrently showed that a primitive called pseudorandom state generators (PRSGs)~\cite{C:JiLiuSon18} can be used to construct computationally hiding and statistically binding quantum bit commitments. 
Especially, Morimae and Yamakawa \cite[footnote 12]{C:MorYam22} mentioned that simply replacing PRGs with single-copy secure PRSGs in $\{Q_{\ywlq,0},Q_{\ywlq,1}\}$ yields a computationally hiding and statistically binding scheme.

Let $\StateGen$ be a single-copy-secure PRSG that, on input $k\in\bit^n$,
outputs an $m$-qubit state $|\phi_k\rangle$ where $m=3n$. 
Then, Morimae and Yamakawa's commitment scheme $\{Q_{\my,0},Q_{\my,1}\}$ is described as follows:

  \begin{align*}
    &Q_{\my,0}\ket{0}_{\regC,\regR}:=\frac{1}{\sqrt{2^n}}\sum_{k\in \bit^n}\ket{\phi_k}_{\regC}\ket{k, 0^{2n}}_{\regR}\\
    &Q_{\my,1}\ket{0}_{\regC,\regR}:=\frac{1}{\sqrt{2^{3n}}}\sum_{r\in \bit^{3n}}\ket{r}_{\regC}\ket{r}_{\regR}.
\end{align*}

\begin{theorem}\label{thm:MY}
If $\StateGen$ is single-copy-secure, then $\{Q_{\my,0},Q_{\my,1}\}$ is computationally hiding and statistically binding.
\end{theorem}

Since the above is not the main construction of \cite{C:MorYam22}, they did not give a security proof. Thus, we give a security proof for completeness.

\begin{proof}[Proof of \Cref{thm:MY}] 
We let $\ket{\psi_b}_{\regC,\regR}\seteq Q_{\my,b}\ket{0}_{\regC,\regR}$. 

\smallskip\noindent\textbf{Computational hiding.}
Note that $\Tr_{\regR}\left(\ket{\psi_1}\bra{\psi_1}_{\regC,\regR}\right)$ is a maximally mixed state, which is a Haar random state when given a single copy. 
On the other hand, we have $\Tr_{\regR}\left(\ket{\psi_0}\bra{\psi_0}_{\regC,\regR}\right)=\frac{1}{2^n}\sum_{k\in\bit^n}\ket{\phi_k}\bra{\phi_k}$. 
Thus, the computational hiding property immediately follows from the single-copy security of $\StateGen$. 

\smallskip\noindent\textbf{Statistical binding.}
The proof is similar to the proof of binding in \cite{C:MorYam22}. Let $F(\rho,\sigma)$ be the fidelity between $\rho$ and $\sigma$.  Then, we have

\begin{align*}
   &F\Big(\mbox{Tr}_{\regR}(|\psi_0\rangle\langle\psi_0|_{\regC,\regR}),\mbox{Tr}_{\regR}(|\psi_1\rangle\langle\psi_1|_{\regC,\regR})\Big) \\
   &=F\Big(
\frac{1}{2^n}\sum_k|\phi_k\rangle\langle\phi_k|, 
\frac{I^{\otimes m}}{2^m}
\Big)\\
&=\Big\|\sum_{i=1}^\xi\sqrt{\lambda_i}\frac{1}{\sqrt{2^m}}|\lambda_i\rangle\langle\lambda_i|\Big\|_1^2\\
&=\left(\sum_{i=1}^\xi\sqrt{\lambda_i}\frac{1}{\sqrt{2^m}}\right)^2\\
&\leq \left(\sum_{i=1}^\xi\lambda_i\right)\left(\sum_{i=1}^\xi\frac{1}{2^m}\right)\\
&\leq 2^{-2n},
\end{align*}
where in the second equality,
$\sum_{i=1}^\xi\lambda_i|\lambda_i\rangle\langle\lambda_i|$
is the diagonalization of $\frac{1}{2^n}\sum_k|\phi_k\rangle\langle\phi_k|$, 
in the first inequality, we have used 
Cauchy–Schwarz inequality,  and
in the final inequality, we have used $\xi\le2^n$ and $m=3n$. 
This means that $\{Q_{\my,0},Q_{\my,1}\}$ is statistically binding~\cite{AC:Yan22}. 
\end{proof}

By applying our conversion to $\{Q_{\my,0},Q_{\my,1}\}$, we obtain the following scheme $\{Q'_{\my,0},Q'_{\my,1}\}$.\footnote{We could also apply our conversion to the main construction of~\cite{C:MorYam22} to obtain a similar scheme.}
  \begin{align*}
    &Q'_{\my,b}\ket{0}_{\regC',\regR'}:=\frac{1}{\sqrt{2^{n+1}}}\sum_{k\in \bit^n}\ket{0,k, 0^{2n}}_{\regC'}\ket{\phi_k}_{\regR'}+(-1)^b\frac{1}{\sqrt{2^{3n+1}}}\sum_{r\in \bit^{3n}}\ket{1,r}_{\regC'}\ket{r}_{\regR'}.
\end{align*}

By \Cref{thm:conversion,thm:MY}, we obtain the following theorem.
\begin{theorem}
If  $\StateGen$ is single-copy-secure, then $\{Q'_{\my,0},Q'_{\my,1}\}$ is statistically hiding and computationally binding. 
\end{theorem}

This is the first statistically hiding and computationally binding quantum bit commitment scheme from PRSGs that makes only a single call to the PRSG. 
If we apply existing conversions~\cite{EC:CreLegSal01,AC:Yan22} to $\{Q_{\my,0},Q_{\my,1}\}$ (or other PRSG-based schemes~\cite{C:AnaQiaYue22}), they result in a schemes that make $\Omega(\secp^2)$ calls to the PRSG.

\subsection{Construction from Injective One-Way Functions}\label{sec:new_inj}
In this section, we show simple constructions of commitments based on any injective one-way functions. 

\smallskip
\noindent\textbf{Perfectly hiding and computationally binding commitment.}
We first construct a perfectly hiding and computationally binding quantum bit commitment scheme from injective one-way function. 
We note that such a commitment is already known from any one-way \emph{permutations} in \cite{EC:DumMaySal00}. (See \Cref{sec:construction_DMS} for details.) Our construction is more general since every  permutation is also injective but the converse is not true.

Let $f:\bit^n \rightarrow \bit^m$ be an injective one-way function. 
Then, we define a canonical quantum bit commitment scheme $\{Q_{\inj,0},Q_{\inj,1}\}$ as follows:

\begin{align*}
    &Q_{\inj,0}\ket{0}_{\regC,\regR}:=\frac{1}{\sqrt{2^n}}\sum_{x\in \bit^n}\ket{x}_{\regC}\ket{f(x)}_{\regR}\\
    &Q_{\inj,1}\ket{0}_{\regC,\regR}:=\frac{1}{\sqrt{2^n}}\sum_{x\in \bit^n}\ket{x}_{\regC}\ket{x, 0^{m-n}}_{\regR}.
\end{align*}
\begin{theorem}\label{thm:security_inj}
If $f$ is an injective one-way function, 
$\{Q_{\inj,0},Q_{\inj,1}\}$ is perfectly hiding and computationally binding.
\end{theorem}
\begin{proof}~

\smallskip
\noindent\textbf{Perfect hiding.}
By the injectivity of $f$, if we trace out $\regR$, then the reduced state in $\regC$ is $\sum_{x\in \bit^n}\ket{x}\bra{x}$ for both $b=0,1$. This implies perfect hiding. 

\smallskip
\noindent\textbf{Computational binding.}
Suppose that the $\{Q_{\inj,0},Q_{\inj,1}\}$ is not computationally binding. Then there exists a polynomial-time computable unitary $U$ over $(\regR,\regZ)$ and an auxiliary state $\ket{\tau}_{\regZ}$ such that 
\begin{align*}
    \left\|(Q_{\inj,1}\ket{0}\bra{0}Q_{\inj,1}^\dagger)_{\regC,\regR}(I_{\regC}\otimes U_{\regR,\regZ})((Q_{\inj,0}\ket{0})_{\regC,\regR}\ket{\tau}_{\regZ})\right\|
\end{align*}
is non-negligible. In particular, its square is also non-negligible. 
It holds that 
\begin{align}
   & \left\|(Q_{\inj,1}\ket{0}\bra{0}Q_{\inj,1}^\dagger)_{\regC,\regR}(I_{\regC}\otimes U_{\regR,\regZ})((Q_{\inj,0}\ket{0})_{\regC,\regR}\ket{\tau}_{\regZ})\right\|^2 \notag\\
  =& \frac{1}{2^{2n}}\left\|\sum_{x\in \bit^n}\bra{x,0^{m-n}}_{\regR}U_{\regR,\regZ}\ket{f(x)}_{\regR}\ket{\tau}_{\regZ}\right\|^2 \notag\\
  \leq & \frac{1}{2^{2n}}\left(\sum_{x\in \bit^n}\left\|\bra{x,0^{m-n}}_{\regR}U_{\regR,\regZ}\ket{f(x)}_{\regR}\ket{\tau}_{\regZ}\right\|\right)^2 \notag\\
  \leq & \frac{1}{2^{n}}\sum_{x\in \bit^n}\left\|\bra{x,0^{m-n}}_{\regR}U_{\regR,\regZ}\ket{f(x)}_{\regR}\ket{\tau}_{\regZ}\right\|^2,
  \label{eq:break_binding}
\end{align}
where the first equality follows from the definition of  $\{Q_{\inj,0},Q_{\inj,1}\}$, the first inequality follows from the triangle inequality, and the second inequality follows from the Cauchy–Schwarz inequality. Thus, the value of \Cref{eq:break_binding} is non-negligible. 

Then, we can construct an adversary $\A$ that breaks the one-wayness of $f$ with advice $\ket{\tau}$ as follows:

\begin{description}
\item[$\A(y;\ket{\tau})$:] Given an instance $y$ and advice $\ket{\tau}$, it generates a state $U\ket{y}_{\regR}\ket{\tau}_{\regZ}$  and measures $\regR$. If the measurement outcome is $(x, 0^{m-n})$ such that $f(x)=y$, it outputs $x$ and otherwise $\bot$.  
\end{description}
We can see that the probability that $\A$ outputs the correct preimage $x$ is the value of \Cref{eq:break_binding}, which is non-negligible.
This contradicts the one-wayness of $f$. Thus, $\{Q_{\inj,0},Q_{\inj,1}\}$ is computationally binding. 
\end{proof}

This is the first \emph{perfectly} hiding quantum bit commitment scheme from injective one-way functions that makes only a single quantum call to the base function.  
Before our work, such a commitment scheme was only known to exist from one-way \emph{permutations}~\cite{EC:DumMaySal00} (which is described in \Cref{sec:construction_DMS}). We remark that Koshiba and Odaira~\cite{KO09,KO11} generalized \cite{EC:DumMaySal00} to make the assumption weaker than the existence of injective one-way functions, but those constructions only achieve \emph{statistical} hiding. 

Alternatively, we can also construct such a commitment scheme by applying our conversion to the (purified version of) construction of computationally hiding and perfectly binding commitment scheme based on Goldreich-Levin theorem~\cite{STOC:GolLev89}. See \Cref{sec:construction_GL} for details.

\smallskip
\noindent\textbf{Computationally hiding and perfectly binding commitment.} 
Next, we apply our conversion to $\{Q_{\inj,0},Q_{\inj,1}\}$ to obtain the following scheme $\{Q'_{\inj,0},Q'_{\inj,1}\}$:

\begin{align*}
    &Q'_{\inj,b}\ket{0}_{\regC',\regR'}:=
    \frac{1}{\sqrt{2^{n+1}}}
    \sum_{x\in \bit^n}
   \left(
    \ket{0}\ket{f(x)}
    +(-1)^b
    \ket{1}\ket{x,0^{m-n}}
    \right)_{\regC'}\otimes
    \ket{x}_{\regR'}.
\end{align*}

By \Cref{thm:conversion,thm:security_inj}, we obtain the following theorem.
\begin{theorem}\label{thm:security_inj_dual}
If $f$ is an injective one-way function, 
$\{Q'_{\inj,0},Q'_{\inj,1}\}$ is computationally hiding and perfectly binding.
\end{theorem}

\smallskip
\noindent\textbf{Comparison with classical construction.}
It is well-known that we can \emph{classically} construct a computationally hiding and perfectly binding non-interactive commitment scheme from injective one-way functions by using Goldreich-Levin theorem~\cite{STOC:GolLev89}. (See \Cref{sec:construction_GL} for more details.) The construction also only makes a single call to the base function. 
Then, one may wonder if it is meaningful to give a \emph{quantum} construction for that. We argue this by remarking the following two points. 

The first is a minor parameter improvement. Our construction has a shorter commitment size than the classical construction (albeit with the apparent disadvantage of the usage of quantum communication). 
Specifically, commitment length of our construction is $m+1$ whereas it is $n+m+1$ in the classical construction. The additional $n$-bit is needed to send the seed for the hardcore bit function in the classical construction.  
We remark that the decommitment length is the same, $n$ for both constructions. Though the improvement is somewhat minor, we believe that it is still worthwhile to show that the quantum communication can reduce the communication complexity of such an important construction of commitments from injective one-way functions.

The second is rather conceptual. We remark that our construction does not make use of any sort of classical hardcore predicates. On the other hand, to our knowledge, the only known way to classically construct a commitment scheme from injective one-way functions (or even one-way permutations) is to rely on some hardcore predicates~\cite{STOC:GolLev89,GRS00,C:HolMauSjo04}. Thus, the source of the pseudorandomness of our construction seems conceptually very different from that for classical constructions. In a nutshell, we interpret the theorem shown by \cite{AAS20} in a completely irrelevant context as a kind of search-to-decision reduction. We believe that this new search-to-decision reduction technique is interesting and will be useful in the future work.

\smallskip
\noindent\textbf{Construction from keyed injective one-way functions.} 
Unfortunately, there is no known candidate of post-quantum injective one-way functions based on standard assumptions.\footnote{Candidate constructions of post-quantum injective one-way functions based on hash functions or block ciphers can be found in \cite[Section 5]{EC:Unruh12}.} On the other hand, there are many candidates of \emph{keyed} injective one-way functions. We remark that our construction can be easily extended to one based on keyed injective one-way functions. 
Let $\{f_k:\bit^n\rightarrow \bit^m\}_{k\in \mathcal{K}}$ be a keyed injective one-way function.
Then, we construct a modified scheme $\{Q_{\injsetup,0},Q_{\injsetup,1}\}$ as follows:

\begin{align*}
    &Q_{\injsetup,0}\ket{0}_{\regC,\regR}:=\frac{1}{\sqrt{2^n|\mathcal{K}|}}\sum_{x\in \bit^n,k\in\mathcal{K}}\ket{x,k}_{\regC}\ket{f_k(x),k}_{\regR}\\
    &Q_{\injsetup,1}\ket{0}_{\regC,\regR}:=\frac{1}{\sqrt{2^n|\mathcal{K}|}}\sum_{x\in \bit^n}\ket{x,k}_{\regC}\ket{x, 0^{m-n},k}_{\regR}
\end{align*}

Then, we can show that $\{Q_{\injsetup,0},Q_{\injsetup,1}\}$ is perfectly hiding and computationally binding similarly to the proof of \Cref{thm:security_inj}.

By applying our conversion, we obtain the following scheme $\{Q'_{\injsetup,0},Q'_{\injsetup,1}\}$.

\begin{align*}
    &Q'_{\injsetup,b}\ket{0}_{\regC',\regR'}:=
    \frac{1}{\sqrt{2^{n+1}|\mathcal{K}|}}
    \sum_{x\in \bit^n,k\in\mathcal{K}}
    \left(
    \ket{0}\ket{f_k(x),k}
    +(-1)^b
    \ket{1}\ket{x,0^{m-n},k}
    \right)_{\regC'}\otimes
    \ket{x,k}_{\regR'}.
\end{align*}

By \Cref{thm:conversion}, $\{Q'_{\injsetup,0},Q'_{\injsetup,1}\}$ is computationally binding and statistically hiding. 

We remark that we can also view it as a quantum-ciphertext PKE (\Cref{def:QPKE}) if we assume that $f_k$ is a trapdoor function. 
That is, we can use $\Tr_{\regR'}\left(Q'_{\injsetup,b}\ket{0}_{\regC',\regR'}\bra{0}_{\regC',\regR'}{Q'}^{\dagger}_{\injsetup,b}\right)$ as an encryption of $b$. We can decrypt it with a trapdoor for $f_k$ by applying a unitary $\ket{x,0^{m-n},k}\mapsto \ket{f_k(x),k}$ on the second register of $\regC'$  controlled on the first register of $\regC'$ (which is efficiently computable with the trapdoor) and then measuring the first register of $\regC'$ in the Hamadard basis. The IND-CPA security directly follows from the computational hiding property of $\{Q'_{\injsetup,0},Q'_{\injsetup,1}\}$. 
This gives a conceptually different way to construct (quantum-ciphertext) PKE from trapdoor functions than that based on hardcore predicates. 


\subsection{Construction from Collapsing Functions}\label{sec:new_collapsing}
In this section, we show simple constructions of commitments based on collapsing functions. Interestingly, the constructions are almost identical to those based on injective one-way functions given in \Cref{sec:new_inj}, but they achieve the other flavors of security than those based on injective one-way functions. 

\smallskip
\noindent\textbf{Computationally hiding and statistically binding commitment.}
We first construct a computationally hiding and statistically binding quantum bit commitment scheme from collapsing functions (\Cref{def:collapsing}). 

Let $\{H_k:\bit^n \rightarrow \bit^m\}_{k\in \mathcal{K}}$ be a family of collapsing functions such that $n\geq m+\secp$.  
Then, we define a canonical quantum bit commitment scheme $\{Q_{\col,0},Q_{\col,1}\}$ as follows:

\begin{align*}
    &Q_{\col,0}\ket{0}_{\regC,\regR}:=\frac{1}{\sqrt{2^n|\mathcal{K}|}}\sum_{x\in \bit^n,k\in \mathcal{K}}\ket{x,k}_{\regC}\ket{H_k(x),0^{n-m},k}_{\regR}\\
    &Q_{\col,1}\ket{0}_{\regC,\regR}:=\frac{1}{\sqrt{2^n|\mathcal{K}|}}\sum_{x\in \bit^n,k\in \mathcal{K}}\ket{x,k}_{\regC}\ket{x,k}_{\regR}.
\end{align*}
\begin{theorem}\label{thm:security_col}
If $\{H_k:\bit^n \rightarrow \bit^m\}_{k\in \mathcal{K}}$ is a family of collapsing functions such that $n\geq m+\secp$, 
$\{Q_{\inj,0},Q_{\inj,1}\}$ is computationally hiding and statistically binding.
\end{theorem}
\begin{proof}~

\smallskip
\noindent\textbf{Computational hiding.}
We have 
$$\Tr_{\regR}(Q_{\col,0}\ket{0}_{\regC,\regR})=
\frac{1}{|\mathcal{K}|}\sum_{y\in \bit^m,k\in \mathcal{K}}
\frac{|S_{k,y}|}{2^n}
\left(\frac{1}{\sqrt{|S_{k,y}|}}\sum_{x\in S_{k,y}}\ket{x,k}\right)\left(\frac{1}{\sqrt{|S_{k,y}|}}\sum_{x'\in S_{k,y}}\bra{x',k}\right)
$$
where 
$$
S_{k,y}\seteq \{x\in \bit^n: H_k(x)=y\}.
$$

Then, by the collapsing property of $\{H_k\}_{k\in \mathcal{K}}$, we can show that $\Tr_{\regR}(Q_{\col,0}\ket{0}_{\regC,\regR})$ is computationally indistinguishable from 
$$
\frac{1}{2^n|\mathcal{K}|}\sum_{x\in \bit^n,k\in \mathcal{K}}
\ket{x,k}\bra{x,k}.
$$

The above state is exactly the same as $\Tr_{\regR}(Q_{\col,1}\ket{0}_{\regC,\regR})$. Thus, the computational hiding property is proven.

\smallskip
\noindent\textbf{Statistical binding.}
Suppose that the $\{Q_{\col,0},Q_{\col,1}\}$ is not statistically binding. Then by a similar argument to that for the proof of computational binding of $\{Q_{\inj,0},Q_{\inj,1}\}$ in \Cref{sec:new_inj}, we can construct an unbounded-time adversary $\A$ such that  
\begin{align*}
    \Pr[\A(k,H_k(x))=x:k\gets \mathcal{K},x\gets \bit^n]
\end{align*}
is non-negligible. However, this is information-theoretically impossible since $n\geq m+\secp$. Thus, $\{Q_{\col,0},Q_{\col,1}\}$ is statistically binding. 
\end{proof}

This is the first \emph{statistically} binding quantum bit commitment scheme from collapsing functions that makes only a single quantum call to the base function. 
To our knowledge, the only known way to construct statistically binding (classical or quantum) commitments from collapsing functions (or collision-resistant functions in the classical case) is to first construct PRGs regarding collapsing (or collision-resistant) functions as one-way functions and then convert it to commitments by~\cite{JC:Naor91}. This requires super-constant number of calls to the base function since known constructions of PRGs from one-way functions require super-constant number of calls~~\cite{SIAM:HILL99,STOC:HaiReiVad10,STOC:VadZhe12}.   

Note that post-quantum statistically \emph{hiding} commitments from collapsing functions are known~\cite{C:HalMic96,EC:Unruh16}. Thus, by applying our conversion to the purified version of the scheme, we can obtain an alternative construction of statistically binding commitments from collapsing functions. See \Cref{sec:construction_collapsing_HM} for details.

\smallskip
\noindent\textbf{Statistically hiding and computationally binding commitment.} 
Next, we apply our conversion to $\{Q_{\col,0},Q_{\col,1}\}$ to obtain the following scheme $\{Q'_{\col,0},Q'_{\col,1}\}$:

\begin{align*}
    &Q'_{\col,b}\ket{0}_{\regC',\regR'}:=
    \frac{1}{\sqrt{2^{n+1}}|\mathcal{K}|}
    \sum_{x\in \bit^n,k\in\mathcal{K}}
   \left(
    \ket{0}\ket{H_k(x),0^{n-m},k}
    +(-1)^b
    \ket{1}\ket{x,k}
    \right)_{\regC'}\otimes
    \ket{x,k}_{\regR'}.
\end{align*}

By \Cref{thm:conversion,thm:security_col}, we obtain the following theorem.
\begin{theorem}\label{thm:security_col_dual}
If $\{H_k\}_{k\in \mathcal{K}}$ is a family of collapsing functions, 
$\{Q'_{\col,0},Q'_{\col,1}\}$ is statistically hiding and computationally binding.
\end{theorem}

As already mentioned, statistically hiding and computationally binding commitments from collapsing functions are known even without using quantum communications~\cite{C:HalMic96,EC:Unruh16}. The above theorem gives an alternative construction for such commitments albeit with quantum communications.

%% file: Appendix_One-Shot_Signatures.tex
\ifnum\llncs=0
\section{Proof of \Cref{lem:Gap_CF_and_CH}}\label{sec:proof_Gap_CF_and_CH}
\else
\subsection{Proof of \Cref{lem:Gap_CF_and_CH}}\label{sec:proof_Gap_CF_and_CH}
\fi
We give a proof of \Cref{lem:Gap_CF_and_CH}. Before giving the proof, we clarify definitions of terms that appear in the statement of the lemma. First, we define (infinitely-often) uniform conversion hardness for group actions.
\begin{definition}[(Infinitely-often) uniform conversion hardness]\label{def:conv_hard_uniform}
We say that an STF $(\setup,\eval,\swap)$ is uniform conversion hard if for any uniform QPT adversary $\A$, we have 
\begin{align*}
    \Pr[f_1(x_1)=y:(\pp,\td)\gets \setup(1^\secp),x_0\gets \calX, y\seteq f_0(x_0),  x_1\gets \A(\pp,\ket{f_0^{-1}(y)})]=\negl(\secp).
\end{align*}
We say that it is infinitely-often uniform conversion hard if the above holds for infinitely many security parameters $\secp\in \mathbb{N}$. 
\end{definition}

Next, we define
(infinitely-often) one-shot signatures. We focus on the case of single-bit messages for simplicity. The message space can be extended to multiple bits by a simple parallel repetition as shown in \cite{STOC:AGKZ20}.

\begin{definition}[One-shot signatures]\label{def:OSS}
A one-shot signature scheme consists of algorithms $(\setup,\keygen,\sign,\vrfy)$.
\begin{description}
\item[$\setup(1^\secp)\rightarrow \crs$:] This is a PPT algorithm that takes the security parameter $1^\secp$ as input, and outputs a classical public parameter $\pp$.
\item[$\keygen(\pp)\rightarrow (\vk,\qsk)$:] This is a QPT algorithm that takes a public parameter $\pp$ as input, and outputs a classical verification key $\vk$ and a quantum signing key $\qsk$.
\item[$\sign(\pp,\qsk,b)\rightarrow \sigma$:] This is a QPT algorithm that takes a public parameter $\pp$, a signing key $\qsk$ and a message $b\in \bit$ as input, and outputs a classical signature $\sigma$. 
\item[$\vrfy(\pp,\vk,b,\sigma)\rightarrow \top/\bot$:] This is a PPT algorithm that takes a public parameter $\pp$, a verification key $\vk$, a message $b$, and a signature $\sigma$ as input, and outputs the decision $\top$ or $\bot$.
\end{description}
We require a one-shot signature scheme to satisfy the following properties.

\smallskip
\noindent\textbf{Correctness.}
For any $b\in \bit$, we have 
\begin{align*}
    \Pr\left[\vrfy(\pp,\vk,b,\sigma)\rightarrow \top 
    :
    \pp\gets \setup(1^\secp),
    (\pk,\qsk)\gets\keygen(\pp),
    \sigma\gets \sign(\pp,\qsk,b)
    \right]= 1-\negl(\secp).
\end{align*}

\smallskip
\noindent\textbf{(Infinitely-often) Security.}
We say that a one-shot signature scheme is secure 
if for any non-uniform QPT adversary $\A$, we have 
\begin{align*}
    \Pr\left[
    \forall b\in \bit~\vrfy(\pp,\vk,b,\sigma_b)=\top
    :
    \pp\gets \setup(1^\secp),
    (\vk,\sigma_0,\sigma_1)\gets \A(\pp)
    \right]=\negl(\secp).
\end{align*}

We say that it is infinitely-often secure if the above holds for infinitely many security parameters $\secp\in \mathbb{N}$. 
\end{definition}

Then, we give a proof of \Cref{lem:Gap_CF_and_CH}.

\begin{proof}[Proof of \Cref{lem:Gap_CF_and_CH}]
Since the proof is almost identical for both \Cref{item:gap_one,item:gap_two}, we first prove \Cref{item:gap_one} and then explain how to modify it to prove \Cref{item:gap_two}. 

\smallskip
\noindent\textbf{Proof of \Cref{item:gap_one}.}
Let $(\setup,\eval,\swap)$ be an STF that is claw-free but not infinitely-often uniform conversion hard. Then, there is a uniform QPT algorithm $\A$ and a polynomial $\poly$ such that 
\begin{align}
\label{eq:break_eow}
  \Pr[f_1(x_1)=y:(\pp,\td)\gets \setup(1^\secp),x_0\gets \calX, y\seteq f_0(x_0),  x_1\gets \A(\pp,\ket{f_0^{-1}(y)})] > 1/\poly(\secp)
\end{align}
for all $\secp$. 
Then, we construct a one-shot signature scheme as follows. Let $N\seteq \poly(\secp)\cdot \secp$. 

\begin{description}
\item[$\setup(1^\secp)$:] For $i\in [N]$, generate $(\pp_i,\td_i)\gets \setup(1^\secp)$, and output
$\pp\seteq \{\pp_i\}_{i\in [N]}$. We write $f_{i,0}$ and $f_{i,1}$ to mean $f_0^{(\pp_i)}$ and $f_1^{(\pp_i)}$, respectively.  
\item[$\keygen(\pp)$:] 
Given $\pp=\{\pp_i\}_{i\in [N]}$, 
for $i\in [N]$, 
generate 
\begin{align*}
    \ket{\calX}=\frac{1}{|\calX|^{1/2}}\sum_{x\in \calX}\ket{x}, 
\end{align*}
coherently compute $f_{i,0}$ in another register to get 
\begin{align*}
    \ket{\calX}=\frac{1}{|\calX|^{1/2}}\sum_{x\in \calX}\ket{x}\ket{f_{i,0}(x)}, 
\end{align*}
measure the second register to get $y_i$. At this point, the first register collapses to $\ket{f_{i,0}^{-1}(y_i)}$. 
Output $\vk\seteq \{y_i\}_{i\in [N]}$ and $\qsk\seteq \{y_i,\ket{f_{i,0}^{-1}(y_i)}\}_{i\in[N]}$. 

\item[$\sign(\pp,\qsk,b)\rightarrow \sigma$:] 
Given $\pp=\{\pp_i\}_{i\in [N]}$, $\qsk= \{y_i,\ket{f_{i,0}^{-1}(y_i)}\}_{i\in[N]}$, and $b\in \bit$, do the following.
\begin{itemize}
    \item If $b=0$, for $i\in [N]$, measure $\ket{f_{i,0}^{-1}(y_i)}$ to get $x_i\in f_{i,0}^{-1}(y_i)$ and output $\sigma\seteq \{x_i\}_{i\in [N]}$.
    \item If $b=1$, for $i\in [N]$, run $\A(\pp_i,\ket{f_{i,0}^{-1}(y_i)})$ to get $x'_i$. If $f_{i,1}(x'_i) \neq y_i$ for all $i\in [N]$, it aborts. Otherwise, it outputs $\sigma \seteq (i^*,x'_{i^*})$ where $i^*$ is the smallest index such that $f_{i^*,1}(x'_{i^*}) = y_{i^*}$. 
\end{itemize}
\item[$\vrfy(\pp,\vk,b,\sigma)\rightarrow \top/\bot$:] Given $\pp=\{\pp_i\}_{i\in [N]}$, $\vk=\{y_i\}_{i\in [N]}$, $b\in \bit$, and a signature $\sigma$, do the following.
\begin{itemize}
    \item If $b=0$, parse $\sigma=\{x_i\}_{i\in [N]}$, and output $\top$ if $f_{i,0}(x_i) = y_i$ for all $i\in [N]$ and $\bot$ otherwise. 
    \item If $b=1$, parse $\sigma=(i,x'_i)$, and output $\top$ if $f_{i,1}(x'_i) = y_i$ and $\bot$ otherwise. 
\end{itemize}
\end{description}

\noindent\textbf{Correctness.}
It is easy to see that the signing algorithm outputs a valid signature whenever it does not abort. By \Cref{eq:break_eow}, the probability that the signing algorithm abort (when $b=1$) is 
\begin{align*}
    (1-1/\poly)^{N} = \negl(\secp)
\end{align*}
by $N= \poly(\secp)\cdot \secp$.

\smallskip
\noindent\textbf{Security.}
Suppose that there is a non-uniform QPT adversary that breaks the above one-shot signature scheme. The adversary is given $\pp=\{\pp_i\}_{i\in [N]}$ and  
finds 
$\vk=\{y_i\}_{i\in [N]}$,
$\sigma_0=\{x_i\}_{i\in [N]}$,  and
$\sigma_1=(i^*,x'_{i^*})$ such that 
$f_{i,0}(x_i) = y_i$ for all $i\in [N]$ and 
$f_{i^*,1}(x'_{i^*}) = y_{i^*}$ 
with a non-negligible probability. 
In particular, 
when the above happens, $(x_{i^*},x'_{i^*})$ forms a claw, i.e., 
we have $f_{i^*,0}(x_{i^*})=f_{i^*,1}(x'_{i^*})$.  
Thus, by randomly guessing $i^*$ and embedding a problem instance of the claw-freeness into the $i^*$-th coordinate, we can break the claw-freeness of the STF $(\setup,\eval,\swap)$, which is a contradiction. 
Thus, the above one-shot signature scheme is secure.

This completes the proof of \Cref{item:gap_one}.

\smallskip
\noindent\textbf{Proof of \Cref{item:gap_two}.}
The proof is similar to that of \Cref{item:gap_one}. The difference is that since we only assume the STF is not uniform conversion hard, we can only assume that \Cref{eq:break_eow} holds for infinitely many $\secp$ rather than all $\secp$. In this case, the correctness of the above one-shot signature scheme only holds for infinitely many $\secp$. To deal with this, we modify the verification algorithm 
so that it approximates $\A$'s success probability up to additive error $1/(4\poly(\secp))$ (except for a negligible probability) and simply accepts if the approximated success probability is smaller than $1/(2\poly(\secp))$. 
Then, the correctness holds on all $\secp\in \mathbb{N}$ because
\begin{itemize}
    \item if the real success probability is smaller than $1/(4\poly(\secp))$, the estimated success probability is smaller than $1/(2\poly(\secp))$ with overwhelming probability, and thus the verification algorithm accepts with overwhelming probability on these security parameters, and
    \item if the real success probability is larger than $1/(4\poly(\secp))$, the signing algorithm should succeed in generating a valid proof with overwhelming probability and thus the  verification algorithm accepts with overwhelming probability on these security parameters. 
\end{itemize}
For the security, we observe that the estimated success probability is smaller than $1/(2\poly(\secp))$ with a negligible probability when the real success probability is larger than $1/\poly(\secp)$. Thus, for those security parameters, the adversary should find valid signatures in the original scheme. Since there are infinitely many such $\secp$, this is not possible by the claw-freeness of the STF. 
\end{proof}

%% file: Appendix_Commitments.tex
\section{More Applications of Our Conversion}\label{sec:more_application}

\subsection{Construction from One-Way Permutations 
via Dumais-Mayers-Salvail Commitment}\label{sec:construction_DMS}
Dumais, Mayers and Salvail~\cite{EC:DumMaySal00} constructed a perfectly hiding and computationally binding commitment from one-way permutations as follows.\footnote{We describe it in the canonical form as in \cite[Section 5]{AC:Yan22}.}
Let $f:\bit^n \rightarrow \bit^n$ be a one-way permutation.  
Then, we define a canonical quantum bit commitment scheme $\{Q_{\dms,0},Q_{\dms,1}\}$ as follows:

\begin{align*}
    Q_{\dms,b}\ket{0}_{\regC,\regR}:=\frac{1}{\sqrt{2^n}}\sum_{x\in \bit^n}((H^b)^{\otimes n}\ket{f(x)})_{\regC}\ket{x}_{\regR}.
\end{align*}
\begin{theorem}[{\cite{EC:DumMaySal00}}]\label{thm:DMS}
If $f$ is a one-way permutation, $\{Q_{\dms,0},Q_{\dms,1}\}$ is perfectly hiding and computationally binding.  
\end{theorem}

By applying our conversion to $\{Q_{\dms,0},Q_{\dms,1}\}$, we obtain the following scheme $\{Q'_{\dms,0},Q'_{\dms,1}\}$: 
\begin{align*}
    Q'_{\dms,b}\ket{0}_{\regC',\regR'}:=
    \frac{1}{\sqrt{2^{n+1}}}
    \left(
    \sum_{x\in \bit^n}
    \left(
    \ket{0,x}_{\regC'}
    \ket{f(x)}_{\regR'}
    +(-1)^b
     \ket{1,x}_{\regC'}
     (H^{\otimes n}\ket{f(x)})_{\regR'}
     \right)
    \right).
\end{align*}

By \Cref{thm:DMS} and \Cref{thm:conversion}, we obtain the following theorem.

\begin{theorem}\label{thm:dual_DMS}
If $f$ is a one-way permutation, $\{Q'_{\dms,0},Q'_{\dms,1}\}$ is computationally hiding and perfectly binding. 
\end{theorem}

\subsection{Constructions from Injective One-Way Functions via Goldreich-Levin Theorem}\label{sec:construction_GL} 
\takashi{I moved this here because there seems no advantage over the construction in \Cref{sec:new_inj}.}
Goldreich and Levin~\cite{STOC:GolLev89} showed that for any one-way function $f:\bit^n\rightarrow \bit^m$, $r\cdot x$ is computationally indistinguishable from a uniform bit given $(f(x),r)$ for $x,r\sample \bit^n$. Here, $r\cdot x:=\sum_{i\in \bit^n}r_i x_i \mod 2$ where $r_i$ and $x_i$ are the $i$-th bits of $r$ and $x$, respectively.   
It is well-known that the above theorem gives us a simple \emph{classical} non-interactive  commitment scheme that is computationally hiding and perfectly binding from any injective one-way function: Let $f:\bit^n\rightarrow \bit^m$ be an injective one-way function. Then, a commitment to bit $b\in \bit$ is set to be $(f(x),r,r\cdot x \oplus b)$. By purifying this construction, we obtain the following canonical quantum bit commitment scheme $\{Q_{\gl,0},Q_{\gl,1}\}$:

\begin{align*}
    Q_{\gl,b}\ket{0}_{\regC,\regR}:=\frac{1}{2^n}\sum_{x\in \bit^n,r\in\bit^n}\ket{f(x),r,r\cdot x \oplus b}_{\regC}\ket{x,r}_{\regR}.
\end{align*}

By Goldreich-Levin theorem (or its quantum version with a better reduction loss shown by Adcock and Cleve~\cite{AC02}), it is straightforward to prove the following theorem. 
\begin{theorem}\label{thm:GL}
If $f$ is an injective one-way function, $\{Q_{\gl,0},Q_{\gl,1}\}$ is computationally hiding and perfectly binding. 
\end{theorem}
 
 By applying our conversion to $\{Q_{\gl,0},Q_{\gl,1}\}$, we obtain the following scheme $\{Q'_{\gl,0},Q'_{\gl,1}\}$.

\begin{align*}
    Q'_{\gl,b}\ket{0}_{\regC',\regR'}:=\frac{1}{2^n}\sum_{x\in \bit^n,r\in\bit^n}
    \left(
    \ket{0,x,r}_{\regC'}\ket{f(x),r,r\cdot x}_{\regR'}
    +(-1)^b \ket{1,x,r}_{\regC'}\ket{f(x),r,r\cdot x\oplus 1}_{\regR'}
    \right).
\end{align*}

By \Cref{thm:conversion,thm:GL}, we obtain the following theorem.

\begin{theorem}\label{thm:dual_GL}
If $f$ is an injective one-way function, $\{Q'_{\gl,0},Q'_{\gl,1}\}$ is perfectly hiding and computationally binding. 
\end{theorem}



\smallskip
\noindent\textbf{Construction from injective one-way functions with trusted setup.} 

Similarly to the schemes in \Cref{sec:new_inj}, the above schemes work based on injective one-way functions with trusted setup. 
Let $\mathcal{R}$ be the randomness space for the setup and $f_R$ be the (description of) injective one-way function generated from the randomness $R\in \mathcal{R}$. Then, we construct a modified scheme $\{Q_{\glsetup,0},Q_{\glsetup,1}\}$:

\begin{align*}
    Q_{\glsetup,b}\ket{0}_{\regC,\regR}:=\frac{1}{2^n\sqrt{|\mathcal{R}|}}\sum_{x\in \bit^n,r\in\bit^n,R\in\mathcal{R}}\ket{f_R(x),r,r\cdot x \oplus b,R}_{\regC}\ket{x,r,R}_{\regR}.
\end{align*}

It is easy to show that $\{Q_{\glsetup,0},Q_{\glsetup,1}\}$ is computationally hiding and perfectly binding. By applying our conversion to $\{Q_{\glsetup,0},Q_{\glsetup,1}\}$, we obtain the following scheme $\{Q'_{\glsetup,0},Q'_{\glsetup,1}\}$: 

\begin{align*}
    Q'_{\glsetup,b}\ket{0}_{\regC',\regR'}:=\frac{1}{2^n\sqrt{|\mathcal{R}|}}\sum_{x\in \bit^n,r\in\bit^n,R\in\mathcal{R}}
    \Big(
    \ket{0,x,r,R}_{\regC'}\ket{f_R(x),r,r\cdot x,R}_{\regR'}\\
    +(-1)^b \ket{1,x,r,R}_{\regC'}\ket{f_R(x),r,r\cdot x\oplus 1,R}_{\regR'}
    \Big).
\end{align*}

By \Cref{thm:conversion}, $\{Q'_{\glsetup,0},Q'_{\glsetup,1}\}$ is perfectly hiding and computationally binding.

\subsection{Construction from Collapsing Hash Functions via Halevi-Micali Commitments}\label{sec:construction_collapsing_HM}
Halevi and Micali~\cite{C:HalMic96} constructed a two-message statistically hiding and computationally binding commitment scheme from any collision-resistant hash functions in the classical setting. Unruh~\cite{EC:Unruh16} pointed out that the scheme may not be secure against quantum adversaries, and showed that it is secure if we assume a stronger security than the collision-resistance called \emph{collapsing} property~\Cref{def:collapsing}. 

In this section, we consider the canonical form of the scheme of \cite{C:HalMic96} and show that it is statistically hiding and computationally binding assuming the collapsing hash functions. Then, we convert it into computationally hiding and statistically binding one by our conversion. 

\smallskip\noindent\textbf{Preparation.}
Before describing the canonical form of the scheme of \cite{C:HalMic96}, we define universal functions and the leftover hash lemma.

\begin{definition}[Universal functions.]
A polynomial-time computable function family $\mathcal{F}=\{f_{k}:\bit^L\rightarrow \bit^{\ell}\}_{k\in \mathcal{K}_{\mathcal{F}}}$ is universal if for any $x,x'\in \bit^{L}$ such that $x\neq x'$, we have 
\begin{align*}
    \Pr_{k\gets \mathcal{K}_{\mathcal{F}}}[f_k(x)=f_k(x')]=2^{-\ell}.
\end{align*}
\end{definition}

For any polynomials $L,\ell$ in the security parameter, 
there unconditionally exists a universal function family from $\bit^{L}$ to $\bit^\ell$~\cite{CW79,SIAM:HILL99}.

\begin{lemma}[Leftover hash lemma~{\cite{SIAM:HILL99}}]\label{lem:LHL} 
Let $\mathcal{F}=\{f_{k}:\bit^L\rightarrow \bit^{\ell}\}_{k\in \mathcal{K}_{\mathcal{F}}}$ be a universal function family. Let $X$ be a random variable over $\bit^L$ such that $H_{\infty}(X)\geq \ell + 2\log \epsilon^{-1}$ where $H_{\infty}(X):= -\log \max_{x\in \bit^L}\Pr[X=x]$.
Then, we have 
\begin{align*}
    \Delta((k,f_k(X)),(k,U_{\ell}))\leq \epsilon
\end{align*}
where $\Delta$ denotes the statistical distance, $k\gets \mathcal{K}$, and $U_\ell\gets \bit^\ell$. 
\end{lemma}

\smallskip\noindent\textbf{Construction.} 
Let $\mathcal{H}=\{H_k:\bit^L\rightarrow \bit^\ell\}_{k\in \mathcal{K}_{\mathcal{H}}}$ be a collapsing function family and $\mathcal{F}=\{f_{k'}:\bit^L\rightarrow \bit\}_{k'\in \mathcal{K}_{\mathcal{F}}}$ be a universal function family where $L=\ell+2\secp +1$. 
The canonical form of the scheme of \cite{C:HalMic96}  $\{Q_{\halmic,0},Q_{\halmic,1}\}$ is described as follows:
\begin{align*}
Q_{\halmic,b}\ket{0}_{\regC,\regR}:=  
\frac{1}{\sqrt{2^L|\mathcal{K}_{\mathcal{H}}||\mathcal{K}_{\mathcal{F}}|}}
\sum_{k\in \mathcal{K}_\mathcal{H},k'\in \mathcal{K}_\mathcal{F},x\in \bit^L}
\ket{k,k',H_k(x),f_{k'}(x)\oplus b}_{\regC}\ket{k,k',x}_{\regR}
\end{align*}

\begin{theorem}\label{thm:halmic}
$\{Q_{\halmic,0},Q_{\halmic,1}\}$ is statistically hiding and computationally binding.
\end{theorem}
\begin{proof}
~

\smallskip
\noindent\textbf{Hiding.}
For $b\in \bit$, we have 
\begin{align*}
    &\Tr_{\regR}(Q_{\halmic,b}\ket{0}_{\regC,\regR})\\
    &=
    \frac{1}{2^L|\mathcal{K}_{\mathcal{H}}||\mathcal{K}_{\mathcal{F}}|}\sum_{k\in \mathcal{K}_\mathcal{H},k'\in \mathcal{K}_\mathcal{F},x\in \bit^L}
\ket{k,k',H_k(x),f_{k'}(x)\oplus b}_{\regC}\bra{k,k',H_k(x),f_{k'}(x)\oplus b}_{\regC}
\end{align*}
Thus,  breaking the hiding property is equivalent to distinguishing the classical distributions  
$\{(k,k',H_k(x),f_{k'}(x)\oplus b):k\gets \mathcal{K}_\mathcal{H},k'\gets \mathcal{K}_\mathcal{F},x\gets \bit^L \}$ 
for $b=0,1$. 
For any fixed $k\in \mathcal{K}_\mathcal{H}$ and $y\in \bit^\ell$, let $X_{k,y}$ be the conditional distribution of
$x\gets \bit^L$ conditioned on $H_k(x)=y$. Since $\ell$-bit side information can decrease the min-entropy by at most $\ell$, we have $H_{\infty}(X_{k,y})\geq L-\ell=2\secp + 1$. Thus, by the leftover hash lemma (\Cref{lem:LHL}), for any fixed $k$, we have  
\begin{align*}
    \Delta((k,k',H_k(x),f_{k'}(x)),(k,k',H_k(x),U_1))\leq 2^{-\secp}
\end{align*}
where $k'\gets \mathcal{K}_\mathcal{F}$, $x\gets \bit^L$, and $U_1\gets \bit$. 
Combined with the above observation, this implies that $\{Q_{\halmic,0},Q_{\halmic,1}\}$ is statistically hiding.

\smallskip
\noindent\textbf{Binding.}
Suppose that the $\{Q_{\halmic,0},Q_{\halmic,1}\}$ is not computationally binding. Then there exists a polynomial-time computable unitary $U$ over $(\regR,\regZ)$ and a auxiliary state $\ket{\tau}_{\regZ}$ such that 
\begin{align*}
    \left\|(Q_{\halmic,1}\ket{0}\bra{0}Q_{\halmic,1}^\dagger)_{\regC,\regR}(I_{\regC}\otimes U_{\regR,\regZ})((Q_{\halmic,0}\ket{0})_{\regC,\regR}\ket{\tau}_{\regZ})\right\|
\end{align*}
is non-negligible. In particular, its square is also non-negligible. 

Let $\regRK$ and $\regRX$ be sub-registers of $\regR$ that store $(k,k')$ and $x$, respectively. \takashi{I hope this explanation makes sense.} 
For $k\in \mathcal{K}_{\mathcal{H}}$, $k'\in \mathcal{K}_{\mathcal{F}}$, $y\in \bit^\ell$, and $z\in \bit$, we define a subset $S_{k,k',y,z}\subseteq \bit^L$ as 
\begin{align*}
    S_{k,k',y,z}:=\{x\in \bit^L: H_k(x)=y~\land~f_{k'}(x)=z\}
\end{align*}
and define the state $\ket{S_{k,k',y,z}}_{\regRX}$ as 
\begin{align*}
    \ket{S_{k,k',y,z}}_{\regRX}:=\frac{1}{\sqrt{|S_{k,k',y,z}|}}\sum_{x\in S_{k,k',y,z}}\ket{x}_{\regRX}.
\end{align*}
For notational convenience, we also define a (non-normalized) state $\ket{\widetilde{S}_{k,k',y,z}}_{\regRX}$ as 
\begin{align*}
    \ket{\widetilde{S}_{k,k',y,z}}_{\regRX}:=
    \sqrt{|S_{k,k',y,z}|}\ket{S_{k,k',y,z}}_{\regRX}=\sum_{x\in S_{k,k',y,z}}\ket{x}_{\regRX}.
\end{align*}
Clearly, for any  $k\in \mathcal{K}_\mathcal{H}$ and $k'\in \mathcal{K}_\mathcal{F}$, 
we have
\begin{align} \label{eq:sum_S}
   \sum_{
  y\in \bit^\ell, z\in \bit
  }
  |S_{k,k',y,z}|
  =
  2^L.
\end{align}
Then, it holds that 
\begin{align}
   & \left\|(Q_{\halmic,1}\ket{0}\bra{0}Q_{\halmic,1}^\dagger)_{\regC,\regR}(I_{\regC}\otimes U_{\regR,\regZ})((Q_{\halmic,0}\ket{0})_{\regC,\regR}\ket{\tau}_{\regZ})\right\|^2 \notag\\
  =& 
   \frac{1}{(2^{L}|\mathcal{K}_{\mathcal{H}}||\mathcal{K}_{\mathcal{F}}|)^2}
   \left\|\sum_{
  \substack{
  k\in \mathcal{K}_\mathcal{H},k'\in \mathcal{K}_\mathcal{F},\\
  y\in \bit^\ell, z\in \bit
  }
  }
 \bra{k,k'}_{\regRK}\bra{\widetilde{S}_{k,k',y,z\oplus 1}}_{\regRX}
 U_{\regR,\regZ}
 \ket{k,k'}_{\regRK}\ket{\widetilde{S}_{k,k',y,z}}_{\regRX}\ket{\tau}_{\regZ}
 \right\|^2 \notag\\
  \leq& 
   \frac{1}{(2^{L}|\mathcal{K}_{\mathcal{H}}||\mathcal{K}_{\mathcal{F}}|)^2}
   \left\|\sum_{
  \substack{
  k\in \mathcal{K}_\mathcal{H},k'\in \mathcal{K}_\mathcal{F},\\
  y\in \bit^\ell, z\in \bit
  }
  }
 \bra{\widetilde{S}_{k,k',y,z\oplus 1}}_{\regRX}
 U_{\regR,\regZ}
 \ket{k,k'}_{\regRK}\ket{\widetilde{S}_{k,k',y,z}}_{\regRX}\ket{\tau}_{\regZ}
 \right\|^2 \notag\\
=& 
   \frac{1}{(2^{L}|\mathcal{K}_{\mathcal{H}}||\mathcal{K}_{\mathcal{F}}|)^2}
   \left\|\sum_{
  \substack{
  k\in \mathcal{K}_\mathcal{H},k'\in \mathcal{K}_\mathcal{F},\\
  y\in \bit^\ell, z\in \bit
  }
  }
  \sum_{x\in S_{k,k',y,z\oplus 1}}
 \bra{x}_{\regRX}
 U_{\regR,\regZ}
 \ket{k,k'}_{\regRK}\ket{\widetilde{S}_{k,k',y,z}}_{\regRX}\ket{\tau}_{\regZ}
 \right\|^2 \notag\\
    \leq& 
   \frac{1}{(2^{L}|\mathcal{K}_{\mathcal{H}}||\mathcal{K}_{\mathcal{F}}|)^2}
   \left(\sum_{
  \substack{
  k\in \mathcal{K}_\mathcal{H},k'\in \mathcal{K}_\mathcal{F},\\
  y\in \bit^\ell, z\in \bit
  }
  }
  \sum_{x\in S_{k,k',y,z\oplus 1}}
  \left\|
 \bra{x}_{\regRX}
 U_{\regR,\regZ}
  \ket{k,k'}_{\regRK}\ket{\widetilde{S}_{k,k',y,z}}_{\regRX}\ket{\tau}_{\regZ}
 \right\|\right)^2 \notag\\
     \leq& 
   \frac{1}{2^{L}|\mathcal{K}_{\mathcal{H}}||\mathcal{K}_{\mathcal{F}}|}
   \sum_{
  \substack{
  k\in \mathcal{K}_\mathcal{H},k'\in \mathcal{K}_\mathcal{F},\\
  y\in \bit^\ell, z\in \bit
  }
  }
  \sum_{x\in S_{k,k',y,z\oplus 1}}
  \left\|
 \bra{x}_{\regRX}
 U_{\regR,\regZ}
  \ket{k,k'}_{\regRK}\ket{\widetilde{S}_{k,k',y,z}}_{\regRX}\ket{\tau}_{\regZ}
 \right\|^2 \notag\\
    \leq& 
   \frac{1}{|\mathcal{K}_{\mathcal{H}}||\mathcal{K}_{\mathcal{F}}|}
   \sum_{
  \substack{
  k\in \mathcal{K}_\mathcal{H},k'\in \mathcal{K}_\mathcal{F},\\
  y\in \bit^\ell, z\in \bit
  }
  }
  \frac{|S_{k,k',y,z}|}{2^L}
  \sum_{x\in S_{k,k',y,z\oplus 1}}
  \left\|
 \bra{x}_{\regRX}
 U_{\regR,\regZ}
   \ket{k,k'}_{\regRK}\ket{S_{k,k',y,z}}_{\regRX}\ket{\tau}_{\regZ}
 \right\|^2
  \label{eq:break_binding_HM}
\end{align}
where the first equality follows from the definition of $\{Q_{\halmic,0},Q_{\halmic,1}\}$, the second equality follows from the definition of $\ket{\widetilde{S}_{k,k',y,z\oplus 1}}_{\regRX}$, the second inequality follows from the triangle inequality,  the third inequality follows from the Cauchy–Schwarz inequality and \Cref{eq:sum_S}, and the fourth inequality follows from $L\geq \ell+1$. 
Therefore, the value of \Cref{eq:break_binding_HM} is non-negligible.
\takashi{The above argument looks a little bit awkward to me. Is there a better one?}

Below, we give an algorithmic interpretation for the value of \Cref{eq:break_binding_HM}. Let $\A$ be an algorithm that works as follows with advice $\ket{\tau}$. 
\begin{description}
\item[$\A(k,k';\ket{\tau})$:] Given $k\in \mathcal{K}_{\mathcal{H}}$ and $k'\in \mathcal{K}_{\mathcal{F}}$ as input and advice $\ket{\tau}$, it generates a state
\begin{align*}
    \sum_{x\in \bit^{L}}\ket{x}_{\regRX}\ket{H_k(x),f_{k'}(x)}_{\regA}
\end{align*}
where $\regA$ is an additional register and measures $\regA$. Let $(y,z)$ be the outcome. At this point, the state in $\regRX$ collapses to $\ket{S_{k,k',y,z}}_{\regRX}$. Then, it computes $U_{\regR,\regZ}\ket{k,k'}_{\regRK}\ket{S_{k,k',y,z}}_{\regRX}\ket{\tau}_{\regZ}$ and measures $\regRX$. Let $x$ be the outcome. If $x\in S_{k,k',y,z\oplus 1}$, then it outputs $1$. Otherwise, it outputs $0$.  
\end{description}
The probability that the measurement outcome of $\regA$ by $\A$ is $(y,z)$ is $\frac{|S_{k,k',y,z}|}{2^L}$. Therefore, $\Pr_{k\sample \mathcal{K}_{\mathcal{H}},k'\sample \mathcal{K}_{\mathcal{F}}}[\A(k,k';\ket{\tau})=1]$ is exactly the value of \Cref{eq:break_binding_HM}, which is non-negligible.

Next, we consider a modified algorithm $\A'$ that works similarly to $\A$ except that it measures $\regRX$ before applying $U_{\regR,\regZ}$. 
Equivalently, instead of generating the state $\sum_{x\in \bit^{L}}\ket{x}_{\regRX}\ket{H_k(x),f_{k'}(x)}_{\regA}$, $\A'$ classically samples $x\sample \bit^L$, computes $y:=H_k(x)$ and $z:=f_{k'}(y)$ and uses $\ket{x}_{\regRX}$ instead of $\ket{S_{k,k',y,z}}_{\regRX}$. 
By a straightforward reduction to the collapsing property of $\mathcal{H}$, 
\begin{align*}
    \left|\Pr_{k\sample \mathcal{K}_{\mathcal{H}},k'\sample \mathcal{K}_{\mathcal{F}}}[\A'(k,k';\ket{\tau})=1]
    -
    \Pr_{k\sample \mathcal{K}_{\mathcal{H}},k'\sample \mathcal{K}_{\mathcal{F}}}[\A(k,k';\ket{\tau})=1]
    \right|=\negl(\secp).
\end{align*}
Therefore, $\Pr_{k\sample \mathcal{K}_{\mathcal{H}},k'\sample \mathcal{K}_{\mathcal{F}}}[\A'(k,k';\ket{\tau})=1]$ is non-negligible.

By using the above, we construct a non-uniform QPT algorithm $\B$ that breaks the collision-resistance of $\mathcal{H}$ as follows. 

\begin{description}
\item[$\B(k;\ket{\tau})$:] Given an input $k\in \mathcal{K}_{\mathcal{H}}$ and advice $\ket{\tau}$, it picks $k'\sample \mathcal{K}_{\mathcal{F}}$ and $x\sample \bit^N$. 
It computes $U_{\regR,\regZ}\ket{k,k'}_{\regRK}\ket{x}_{\regRX}\ket{\tau}_{\regZ}$ and measures $\regRX$. Let $x'$ be the outcome. It outputs $(x,x')$.
\end{description}

It is easy to see that 
\begin{align*}
    \Pr_{k\sample \mathcal{K}_{\mathcal{H}}}[x'\in S_{k,k',H_k(x),f_{k'}(x)\oplus 1}:(x,x')\sample\B(k;\ket{\tau})]
    =
    \Pr_{k\sample \mathcal{K}_{\mathcal{H}},k'\sample \mathcal{K}_{\mathcal{F}}}[\A'(k,k';\ket{\tau})=1]
\end{align*}
where $k'$ in the LHS is the one picked by $\B$. 
Since the RHS is non-negligible, the LHS is non-negligible. 
Moreover, when $x'\in S_{k,k',H_k(x),f_{k'}(x)\oplus 1}$, we have $H_k(x)=H_k(x')$ and $x\neq x'$ by the definition of $S_{k,k',H_k(x),f_{k'}(x)\oplus 1}$. Therefore, 
\begin{align*}
    \Pr_{k\sample \mathcal{K}_{\mathcal{H}}}[H_k(x)=H_k(x')~\land~x\neq x':(x,x')\sample\B(k;\ket{\tau})]
\end{align*}
is non-negligible. 
This means that $\B$ with advice $\ket{\tau}$ breaks the collision-resistance of $\mathcal{H}$. This contradicts the assumption that $\mathcal{H}$ is collapsing since the collapsing property implies the collision-resistance. 
Thus, $\{Q_{\halmic,0},Q_{\halmic,1}\}$ is computationally binding. This completes the proof of \Cref{thm:halmic}.
\takashi{Is there a more direct proof?}
\end{proof}

By applying our conversion, we obtain the following scheme $\{Q'_{\halmic,0},Q'_{\halmic,1}\}$
\begin{align*}
Q'_{\halmic,b}\ket{0}_{\regC',\regR'}:=  
\frac{1}{\sqrt{2^{L+1}|\mathcal{K}_{\mathcal{H}}||\mathcal{K}_{\mathcal{F}}|}}
\sum_{k\in \mathcal{K}_\mathcal{H},k'\in \mathcal{K}_\mathcal{F},x\in \bit^L}
\bigg(
\ket{0,k,k',x}_{\regC'}\ket{k,k',H_k(x),f_{k'}(x)}_{\regR'}\\
+(-1)^b\ket{1,k,k',x}_{\regC'}\ket{k,k',H_k(x),f_{k'}(x)\oplus 1}_{\regR'}\bigg)
\end{align*}

By \Cref{thm:halmic,thm:conversion}, we obtain the following theorem.
\begin{theorem}
$\{Q'_{\halmic,0},Q'_{\halmic,1}\}$ is computationally hiding and statistically binding.
\end{theorem}


%% file: main.bbl
\newcommand{\etalchar}[1]{$^{#1}$}
\begin{thebibliography}{YWLQ15}

\bibitem[AAS20]{AAS20}
Scott Aaronson, Yosi Atia, and Leonard Susskind.
\newblock On the hardness of detecting macroscopic superpositions.
\newblock {\em Electron. Colloquium Comput. Complex.}, page 146, 2020.

\bibitem[AC02]{AC02}
Mark Adcock and Richard Cleve.
\newblock A quantum goldreich-levin theorem with cryptographic applications.
\newblock In Helmut Alt and Afonso Ferreira, editors, {\em {STACS} 2002, 19th
  Annual Symposium on Theoretical Aspects of Computer Science, Antibes - Juan
  les Pins, France, March 14-16, 2002, Proceedings}, volume 2285 of {\em
  Lecture Notes in Computer Science}, pages 323--334. Springer, 2002.

\bibitem[ADMP20]{AC:ADMP20}
Navid Alamati, Luca {De Feo}, Hart Montgomery, and Sikhar Patranabis.
\newblock Cryptographic group actions and applications.
\newblock In Shiho Moriai and Huaxiong Wang, editors, {\em ASIACRYPT~2020,
  Part~II}, volume 12492 of {\em {LNCS}}, pages 411--439. Springer, Heidelberg,
  December 2020.

\bibitem[AGKZ20]{STOC:AGKZ20}
Ryan Amos, Marios Georgiou, Aggelos Kiayias, and Mark Zhandry.
\newblock One-shot signatures and applications to hybrid quantum/classical
  authentication.
\newblock In Konstantin Makarychev, Yury Makarychev, Madhur Tulsiani, Gautam
  Kamath, and Julia Chuzhoy, editors, {\em 52nd ACM STOC}, pages 255--268.
  {ACM} Press, June 2020.

\bibitem[AGM21]{Alagic_2021}
Gorjan Alagic, Tommaso Gagliardoni, and Christian Majenz.
\newblock Can you sign a quantum state?
\newblock {\em Quantum}, 5:603, dec 2021.

\bibitem[AQY22]{C:AnaQiaYue22}
Prabhanjan Ananth, Luowen Qian, and Henry Yuen.
\newblock Cryptography from pseudorandom quantum states.
\newblock In Yevgeniy Dodis and Thomas Shrimpton, editors, {\em CRYPTO~2022,
  Part~I}, volume 13507 of {\em {LNCS}}, pages 208--236. Springer, Heidelberg,
  August 2022.

\bibitem[Bab16]{STOC:Babai16}
L{\'a}szl{\'o} Babai.
\newblock Graph isomorphism in quasipolynomial time [extended abstract].
\newblock In Daniel Wichs and Yishay Mansour, editors, {\em 48th ACM STOC},
  pages 684--697. {ACM} Press, June 2016.

\bibitem[BB84]{BB84}
Charles~H. Bennett and Gilles Brassard.
\newblock {Quantum cryptography: Public key distribution and coin tossing}.
\newblock In {\em Proceedings of IEEE International Conference on Computers,
  Systems, and Signal Processing}, pages 175--179, 1984.

\bibitem[BB21]{TCC:BitBra21}
Nir Bitansky and Zvika Brakerski.
\newblock Classical binding for quantum commitments.
\newblock In Kobbi Nissim and Brent Waters, editors, {\em TCC~2021, Part~I},
  volume 13042 of {\em {LNCS}}, pages 273--298. Springer, Heidelberg, November
  2021.

\bibitem[BC91]{C:BraCre90}
Gilles Brassard and Claude Cr{\'e}peau.
\newblock Quantum bit commitment and coin tossing protocols.
\newblock In Alfred~J. Menezes and Scott~A. Vanstone, editors, {\em CRYPTO'90},
  volume 537 of {\em {LNCS}}, pages 49--61. Springer, Heidelberg, August 1991.

\bibitem[BCKM21]{C:BCKM21b}
James Bartusek, Andrea Coladangelo, Dakshita Khurana, and Fermi Ma.
\newblock One-way functions imply secure computation in a quantum world.
\newblock In Tal Malkin and Chris Peikert, editors, {\em CRYPTO~2021, Part~I},
  volume 12825 of {\em {LNCS}}, pages 467--496, Virtual Event, August 2021.
  Springer, Heidelberg.

\bibitem[BCM{\etalchar{+}}18]{FOCS:BCMVV18}
Zvika Brakerski, Paul Christiano, Urmila Mahadev, Umesh~V. Vazirani, and Thomas
  Vidick.
\newblock A cryptographic test of quantumness and certifiable randomness from a
  single quantum device.
\newblock In Mikkel Thorup, editor, {\em 59th FOCS}, pages 320--331. {IEEE}
  Computer Society Press, October 2018.

\bibitem[BDS17]{BDS17}
Shalev Ben-David and Or~Sattath.
\newblock Quantum tokens for digital signatures.
\newblock Cryptology ePrint Archive, Paper 2017/094, 2017.
\newblock \url{https://eprint.iacr.org/2017/094}.

\bibitem[BHY09]{EC:BelHofYil09}
Mihir Bellare, Dennis Hofheinz, and Scott Yilek.
\newblock Possibility and impossibility results for encryption and commitment
  secure under selective opening.
\newblock In Antoine Joux, editor, {\em EUROCRYPT~2009}, volume 5479 of {\em
  {LNCS}}, pages 1--35. Springer, Heidelberg, April 2009.

\bibitem[BJ15]{C:BroJef15}
Anne Broadbent and Stacey Jeffery.
\newblock Quantum homomorphic encryption for circuits of low {T}-gate
  complexity.
\newblock In Rosario Gennaro and Matthew J.~B. Robshaw, editors, {\em
  CRYPTO~2015, Part~II}, volume 9216 of {\em {LNCS}}, pages 609--629. Springer,
  Heidelberg, August 2015.

\bibitem[BY91]{C:BraYun90}
Gilles Brassard and Moti Yung.
\newblock One-way group actions.
\newblock In Alfred~J. Menezes and Scott~A. Vanstone, editors, {\em CRYPTO'90},
  volume 537 of {\em {LNCS}}, pages 94--107. Springer, Heidelberg, August 1991.

\bibitem[CDMS04]{TCC:CDMS04}
Claude Cr{\'e}peau, Paul Dumais, Dominic Mayers, and Louis Salvail.
\newblock Computational collapse of quantum state with application to oblivious
  transfer.
\newblock In Moni Naor, editor, {\em TCC~2004}, volume 2951 of {\em {LNCS}},
  pages 374--393. Springer, Heidelberg, February 2004.

\bibitem[CKR11]{ICALP:ChaKerRos11}
Andr{\'e} Chailloux, Iordanis Kerenidis, and Bill Rosgen.
\newblock Quantum commitments from complexity assumptions.
\newblock In Luca Aceto, Monika Henzinger, and Jiri Sgall, editors, {\em ICALP
  2011, Part~I}, volume 6755 of {\em {LNCS}}, pages 73--85. Springer,
  Heidelberg, July 2011.

\bibitem[CLLZ21]{C:CLLZ21}
Andrea Coladangelo, Jiahui Liu, Qipeng Liu, and Mark Zhandry.
\newblock Hidden cosets and applications to unclonable cryptography.
\newblock In Tal Malkin and Chris Peikert, editors, {\em CRYPTO~2021, Part~I},
  volume 12825 of {\em {LNCS}}, pages 556--584, Virtual Event, August 2021.
  Springer, Heidelberg.

\bibitem[CLM{\etalchar{+}}18]{AC:CLMPR18}
Wouter Castryck, Tanja Lange, Chloe Martindale, Lorenz Panny, and Joost Renes.
\newblock {CSIDH}: An efficient post-quantum commutative group action.
\newblock In Thomas Peyrin and Steven Galbraith, editors, {\em ASIACRYPT~2018,
  Part~III}, volume 11274 of {\em {LNCS}}, pages 395--427. Springer,
  Heidelberg, December 2018.

\bibitem[CLS01]{EC:CreLegSal01}
Claude Cr{\'e}peau, Fr{\'e}d{\'e}ric L{\'e}gar{\'e}, and Louis Salvail.
\newblock How to convert the flavor of a quantum bit commitment.
\newblock In Birgit Pfitzmann, editor, {\em EUROCRYPT~2001}, volume 2045 of
  {\em {LNCS}}, pages 60--77. Springer, Heidelberg, May 2001.

\bibitem[Cou06]{Couveignes06}
Jean-Marc Couveignes.
\newblock Hard homogeneous spaces.
\newblock Cryptology ePrint Archive, Paper 2006/291, 2006.
\newblock \url{https://eprint.iacr.org/2006/291}.

\bibitem[CW79]{CW79}
Larry Carter and Mark~N. Wegman.
\newblock Universal classes of hash functions.
\newblock {\em J. Comput. Syst. Sci.}, 18(2):143--154, 1979.

\bibitem[DFS04]{C:DamFehSal04}
Ivan Damg{\aa}rd, Serge Fehr, and Louis Salvail.
\newblock Zero-knowledge proofs and string commitments withstanding quantum
  attacks.
\newblock In Matthew Franklin, editor, {\em CRYPTO~2004}, volume 3152 of {\em
  {LNCS}}, pages 254--272. Springer, Heidelberg, August 2004.

\bibitem[DH76]{DH76}
Whitfield Diffie and Martin~E. Hellman.
\newblock New directions in cryptography.
\newblock {\em {IEEE} Trans. Inf. Theory}, 22(6):644--654, 1976.

\bibitem[DMS00]{EC:DumMaySal00}
Paul Dumais, Dominic Mayers, and Louis Salvail.
\newblock Perfectly concealing quantum bit commitment from any quantum one-way
  permutation.
\newblock In Bart Preneel, editor, {\em EUROCRYPT~2000}, volume 1807 of {\em
  {LNCS}}, pages 300--315. Springer, Heidelberg, May 2000.

\bibitem[DS22]{cryptoeprint:2022/786}
Marcel Dall'Agnol and Nicholas Spooner.
\newblock On the necessity of collapsing.
\newblock Cryptology ePrint Archive, Paper 2022/786, 2022.
\newblock \url{https://eprint.iacr.org/2022/786}.

\bibitem[{ElG}85]{ElGamal85}
Taher {ElGamal}.
\newblock A public key cryptosystem and a signature scheme based on discrete
  logarithms.
\newblock {\em {IEEE} Transactions on Information Theory}, 31:469--472, 1985.

\bibitem[FGS19]{FGS19}
Vyacheslav Futorny, Joshua~A. Grochow, and Vladimir~V. Sergeichuk.
\newblock Wildness for tensors.
\newblock {\em Linear Algebra and its Applications}, 566:212--244, apr 2019.

\bibitem[FUYZ22]{FUYZ20}
Junbin Fang, Dominique Unruh, Jun Yan, and Dehua Zhou.
\newblock How to base security on the perfect/statistical binding property of
  quantum bit commitment?
\newblock In Sang~Won Bae and Heejin Park, editors, {\em 33rd International
  Symposium on Algorithms and Computation, {ISAAC} 2022, December 19-21, 2022,
  Seoul, Korea}, volume 248 of {\em LIPIcs}, pages 26:1--26:12. Schloss
  Dagstuhl - Leibniz-Zentrum f{\"{u}}r Informatik, 2022.

\bibitem[GJMZ22]{GJMZ22}
Sam Gunn, Nathan Ju, Fermi Ma, and Mark Zhandry.
\newblock Commitments to quantum states.
\newblock arXiv:2210.05138, 2022.

\bibitem[GL89]{STOC:GolLev89}
Oded Goldreich and Leonid~A. Levin.
\newblock A hard-core predicate for all one-way functions.
\newblock In {\em 21st ACM STOC}, pages 25--32. {ACM} Press, May 1989.

\bibitem[GMR84]{FOCS:GolMicRiv84}
Shafi Goldwasser, Silvio Micali, and Ronald~L. Rivest.
\newblock A ``paradoxical'' solution to the signature problem (extended
  abstract).
\newblock In {\em 25th FOCS}, pages 441--448. {IEEE} Computer Society Press,
  October 1984.

\bibitem[GPV08]{STOC:GenPeiVai08}
Craig Gentry, Chris Peikert, and Vinod Vaikuntanathan.
\newblock Trapdoors for hard lattices and new cryptographic constructions.
\newblock In Richard~E. Ladner and Cynthia Dwork, editors, {\em 40th ACM STOC},
  pages 197--206. {ACM} Press, May 2008.

\bibitem[GRS00]{GRS00}
Oded Goldreich, Ronitt Rubinfeld, and Madhu Sudan.
\newblock Learning polynomials with queries: The highly noisy case.
\newblock {\em {SIAM} J. Discret. Math.}, 13(4):535--570, 2000.

\bibitem[HILL99]{SIAM:HILL99}
Johan H{\aa}stad, Russell Impagliazzo, Leonid~A. Levin, and Michael Luby.
\newblock A pseudorandom generator from any one-way function.
\newblock {\em {SIAM} J. Comput.}, 28(4):1364--1396, 1999.

\bibitem[HM96]{C:HalMic96}
Shai Halevi and Silvio Micali.
\newblock Practical and provably-secure commitment schemes from collision-free
  hashing.
\newblock In Neal Koblitz, editor, {\em CRYPTO'96}, volume 1109 of {\em
  {LNCS}}, pages 201--215. Springer, Heidelberg, August 1996.

\bibitem[HMS04]{C:HolMauSjo04}
Thomas Holenstein, Ueli~M. Maurer, and Johan Sj{\"o}din.
\newblock Complete classification of bilinear hard-core functions.
\newblock In Matthew Franklin, editor, {\em CRYPTO~2004}, volume 3152 of {\em
  {LNCS}}, pages 73--91. Springer, Heidelberg, August 2004.

\bibitem[HR07]{STOC:HaiRei07}
Iftach Haitner and Omer Reingold.
\newblock Statistically-hiding commitment from any one-way function.
\newblock In David~S. Johnson and Uriel Feige, editors, {\em 39th ACM STOC},
  pages 1--10. {ACM} Press, June 2007.

\bibitem[HRV10]{STOC:HaiReiVad10}
Iftach Haitner, Omer Reingold, and Salil~P. Vadhan.
\newblock Efficiency improvements in constructing pseudorandom generators from
  one-way functions.
\newblock In Leonard~J. Schulman, editor, {\em 42nd ACM STOC}, pages 437--446.
  {ACM} Press, June 2010.

\bibitem[Imp95]{Impagliazzo95}
Russell Impagliazzo.
\newblock A personal view of average-case complexity.
\newblock In {\em Proceedings of the Tenth Annual Structure in Complexity
  Theory Conference, Minneapolis, Minnesota, USA, June 19-22, 1995}, pages
  134--147. {IEEE} Computer Society, 1995.

\bibitem[JD11]{PQCRYPTO:JaoDeFo11}
David Jao and Luca {De Feo}.
\newblock Towards quantum-resistant cryptosystems from supersingular elliptic
  curve isogenies.
\newblock In Bo-Yin Yang, editor, {\em Post-Quantum Cryptography - 4th
  International Workshop, PQCrypto 2011}, pages 19--34. Springer, Heidelberg,
  November~/~December 2011.

\bibitem[JLS18]{C:JiLiuSon18}
Zhengfeng Ji, Yi-Kai Liu, and Fang Song.
\newblock Pseudorandom quantum states.
\newblock In Hovav Shacham and Alexandra Boldyreva, editors, {\em CRYPTO~2018,
  Part~III}, volume 10993 of {\em {LNCS}}, pages 126--152. Springer,
  Heidelberg, August 2018.

\bibitem[JQSY19]{TCC:JQSY19}
Zhengfeng Ji, Youming Qiao, Fang Song, and Aaram Yun.
\newblock General linear group action on tensors: {A} candidate for
  post-quantum cryptography.
\newblock In Dennis Hofheinz and Alon Rosen, editors, {\em TCC~2019, Part~I},
  volume 11891 of {\em {LNCS}}, pages 251--281. Springer, Heidelberg, December
  2019.

\bibitem[KKNY05]{EC:KKNY05}
Akinori Kawachi, Takeshi Koshiba, Harumichi Nishimura, and Tomoyuki Yamakami.
\newblock Computational indistinguishability between quantum states and its
  cryptographic application.
\newblock In Ronald Cramer, editor, {\em EUROCRYPT~2005}, volume 3494 of {\em
  {LNCS}}, pages 268--284. Springer, Heidelberg, May 2005.

\bibitem[KO09]{KO09}
Takeshi Koshiba and Takanori Odaira.
\newblock Statistically-hiding quantum bit commitment from
  approximable-preimage-size quantum one-way function.
\newblock In Andrew~M. Childs and Michele Mosca, editors, {\em Theory of
  Quantum Computation, Communication, and Cryptography, 4th Workshop, {TQC}
  2009, Waterloo, Canada, May 11-13, 2009, Revised Selected Papers}, volume
  5906 of {\em Lecture Notes in Computer Science}, pages 33--46. Springer,
  2009.

\bibitem[KO11]{KO11}
Takeshi Koshiba and Takanori Odaira.
\newblock Non-interactive statistically-hiding quantum bit commitment from any
  quantum one-way function.
\newblock arXiv:1102.3441, 2011.

\bibitem[Kre21]{Kre21}
William Kretschmer.
\newblock Quantum pseudorandomness and classical complexity.
\newblock In Min{-}Hsiu Hsieh, editor, {\em 16th Conference on the Theory of
  Quantum Computation, Communication and Cryptography, {TQC} 2021, July 5-8,
  2021, Virtual Conference}, volume 197 of {\em LIPIcs}, pages 2:1--2:20.
  Schloss Dagstuhl - Leibniz-Zentrum f{\"{u}}r Informatik, 2021.

\bibitem[LC97]{LC97}
Hoi-Kwong Lo and Hoi~Fung Chau.
\newblock Is quantum bit commitment really possible?
\newblock {\em Physical Review Letters}, 78(17):3410, 1997.

\bibitem[Mah18]{FOCS:Mahadev18b}
Urmila Mahadev.
\newblock Classical homomorphic encryption for quantum circuits.
\newblock In Mikkel Thorup, editor, {\em 59th FOCS}, pages 332--338. {IEEE}
  Computer Society Press, October 2018.

\bibitem[May97]{May97}
Dominic Mayers.
\newblock Unconditionally secure quantum bit commitment is impossible.
\newblock {\em Physical review letters}, 78(17):3414, 1997.

\bibitem[MY22]{C:MorYam22}
Tomoyuki Morimae and Takashi Yamakawa.
\newblock Quantum commitments and signatures without one-way functions.
\newblock In Yevgeniy Dodis and Thomas Shrimpton, editors, {\em CRYPTO~2022,
  Part~I}, volume 13507 of {\em {LNCS}}, pages 269--295. Springer, Heidelberg,
  August 2022.

\bibitem[Nao91]{JC:Naor91}
Moni Naor.
\newblock Bit commitment using pseudorandomness.
\newblock {\em Journal of Cryptology}, 4(2):151--158, January 1991.

\bibitem[OTU00]{C:OkaTanUch00}
Tatsuaki Okamoto, Keisuke Tanaka, and Shigenori Uchiyama.
\newblock Quantum public-key cryptosystems.
\newblock In Mihir Bellare, editor, {\em CRYPTO~2000}, volume 1880 of {\em
  {LNCS}}, pages 147--165. Springer, Heidelberg, August 2000.

\bibitem[Reg09]{JACM:Regev09}
Oded Regev.
\newblock On lattices, learning with errors, random linear codes, and
  cryptography.
\newblock {\em J. {ACM}}, 56(6):34:1--34:40, 2009.

\bibitem[RS06]{cryptoeprint:2006/145}
Alexander Rostovtsev and Anton Stolbunov.
\newblock Public-key cryptosystem based on isogenies.
\newblock Cryptology ePrint Archive, Paper 2006/145, 2006.
\newblock \url{https://eprint.iacr.org/2006/145}.

\bibitem[Sho99]{Shor99}
Peter~W. Shor.
\newblock Polynomial-time algorithms for prime factorization and discrete
  logarithms on a quantum computer.
\newblock {\em {SIAM} Rev.}, 41(2):303--332, 1999.

\bibitem[Unr12]{EC:Unruh12}
Dominique Unruh.
\newblock Quantum proofs of knowledge.
\newblock In David Pointcheval and Thomas Johansson, editors, {\em
  EUROCRYPT~2012}, volume 7237 of {\em {LNCS}}, pages 135--152. Springer,
  Heidelberg, April 2012.

\bibitem[Unr16a]{AC:Unruh16}
Dominique Unruh.
\newblock Collapse-binding quantum commitments without random oracles.
\newblock In Jung~Hee Cheon and Tsuyoshi Takagi, editors, {\em ASIACRYPT~2016,
  Part~II}, volume 10032 of {\em {LNCS}}, pages 166--195. Springer, Heidelberg,
  December 2016.

\bibitem[Unr16b]{EC:Unruh16}
Dominique Unruh.
\newblock Computationally binding quantum commitments.
\newblock In Marc Fischlin and Jean-S{\'{e}}bastien Coron, editors, {\em
  EUROCRYPT~2016, Part~II}, volume 9666 of {\em {LNCS}}, pages 497--527.
  Springer, Heidelberg, May 2016.

\bibitem[VZ12]{STOC:VadZhe12}
Salil~P. Vadhan and Colin~Jia Zheng.
\newblock Characterizing pseudoentropy and simplifying pseudorandom generator
  constructions.
\newblock In Howard~J. Karloff and Toniann Pitassi, editors, {\em 44th ACM
  STOC}, pages 817--836. {ACM} Press, May 2012.

\bibitem[Yan21]{AC:Yan21}
Jun Yan.
\newblock Quantum computationally predicate-binding commitments with
  application in quantum zero-knowledge arguments for {NP}.
\newblock In Mehdi Tibouchi and Huaxiong Wang, editors, {\em ASIACRYPT~2021,
  Part~I}, volume 13090 of {\em {LNCS}}, pages 575--605. Springer, Heidelberg,
  December 2021.

\bibitem[Yan22]{AC:Yan22}
Jun Yan.
\newblock General properties of quantum bit commitments (extended abstract).
\newblock In Shweta Agrawal and Dongdai Lin, editors, {\em Advances in
  Cryptology - {ASIACRYPT} 2022 - 28th International Conference on the Theory
  and Application of Cryptology and Information Security, Taipei, Taiwan,
  December 5-9, 2022, Proceedings, Part {IV}}, volume 13794 of {\em Lecture
  Notes in Computer Science}, pages 628--657. Springer, 2022.

\bibitem[YWLQ15]{YWLQ15}
Jun Yan, Jian Weng, Dongdai Lin, and Yujuan Quan.
\newblock Quantum bit commitment with application in quantum zero-knowledge
  proof (extended abstract).
\newblock In Khaled~M. Elbassioni and Kazuhisa Makino, editors, {\em Algorithms
  and Computation - 26th International Symposium, {ISAAC} 2015, Nagoya, Japan,
  December 9-11, 2015, Proceedings}, volume 9472 of {\em Lecture Notes in
  Computer Science}, pages 555--565. Springer, 2015.

\bibitem[Zha19]{EC:Zhandry19b}
Mark Zhandry.
\newblock Quantum lightning never strikes the same state twice.
\newblock In Yuval Ishai and Vincent Rijmen, editors, {\em EUROCRYPT~2019,
  Part~III}, volume 11478 of {\em {LNCS}}, pages 408--438. Springer,
  Heidelberg, May 2019.

\end{thebibliography}
